\documentclass[acmsmall,screen]{acmart}

\AtBeginDocument{%
  }

\setcopyright{acmcopyright}
\copyrightyear{2018}
\acmYear{2018}
\acmDOI{XXXXXXX.XXXXXXX}

\acmJournal{JACM}
\acmVolume{37}
\acmNumber{4}
\acmArticle{111}
\acmMonth{8}

\newcommand{\macrospath}{macros}

\usepackage{ifthen}
\usepackage{xspace}

\newboolean{talk}
\setboolean{talk}{false}
\newboolean{paper}
\setboolean{paper}{false}

\newboolean{IEEEstyle}
\setboolean{IEEEstyle}{false}
\newboolean{lipicsstyle}
\setboolean{lipicsstyle}{false}
\newboolean{eptcsstyle}
\setboolean{eptcsstyle}{false}
\newboolean{entcsstyle}
\setboolean{entcsstyle}{false}
\newboolean{sigplanstyle}
\setboolean{sigplanstyle}{false}
\newboolean{easychairstyle}
\setboolean{easychairstyle}{false}
\newboolean{scrartclstyle}
\setboolean{scrartclstyle}{true}
\newboolean{lncsstyle}
\setboolean{lncsstyle}{false}

\newboolean{acmartstyle}
\setboolean{acmartstyle}{false}

\newboolean{needstheorems}
\setboolean{needstheorems}{false}
\newboolean{withimages}
\setboolean{withimages}{false}
\newboolean{withproofs}
\setboolean{withproofs}{true}

\newboolean{french}
\setboolean{french}{false}

\setboolean{acmartstyle}{true}
\usepackage{amsthm}
    \newtheorem{theorem}{Theorem}[section]
    \newtheorem{lemma}[theorem]{Lemma}
    
    \newtheorem{corollary}[theorem]{Corollary}
    \newtheorem{proposition}[theorem]{Proposition}
    \newtheorem{definition}[theorem]{Definition}
    \newtheorem{remark}[theorem]{Remark}
    
    \ifthenelse{\boolean{scrartclstyle}}{
  
}{}

\usepackage[utf8]{inputenc}

\ifthenelse{\boolean{acmartstyle}}{}{
\usepackage{amssymb}	
}

\usepackage{stmaryrd}
\usepackage{bussproofs}
\usepackage{mathtools}
\usepackage{bbold}
\usepackage{qtree}

\usepackage{ifthen}
\usepackage{xspace}
\usepackage{fancybox}
\usepackage{hhline}

\usepackage{soul}

\usepackage{xcolor}
\usepackage{multirow}
\usepackage{microtype}

\usepackage{hyperref}

\hypersetup{colorlinks=true}
\hypersetup{hidelinks}

\usepackage[colorinlistoftodos]{todonotes}
\usepackage{ifthen}
\usepackage{proof}
\usepackage{pifont}
\usepackage{appendix}
\capturecounter{theorem}

\newcommand{\sem}[1]{\interp{#1}}
\newcommand{\ignore}[1]{}
 
\newcommand{\colspace}{@{\hspace{.5cm}}}
\newcommand{\myinput}[1]{\ifthenelse{\boolean{withimages}}{\input{#1}}{}}

\newcommand{\reflemma}[1]{Lemma~\ref{l:#1}}

\newcommand{\refthm}[1]{Thm.~\ref{thm:#1}}
\newcommand{\refthmp}[2]{Thm.~\ref{thm:#1}.\ref{p:#1-#2}}

\newcommand{\refprop}[1]{Prop.~\ref{prop:#1}}
\newcommand{\refsect}[1]{Sect.~\ref{sect:#1}}

\newcommand{\refapp}[1]{Appendix~\ref{app:#1} (p.~\pageref{app:#1})}

\newcommand{\refdef}[1]{Definition~\ref{def:#1}}

\newcommand{\ie}{\textit{i.e.}\xspace}
\newcommand{\eg}{\textit{e.g.}\xspace}
\newcommand{\ih}{\textit{i.h.}\xspace}

\newcommand{\RED}[1]{{\color{red} {#1}}}
\newcommand{\blue}[1]{{\color{blue} {#1}}}

\newcommand{\adr}[1]{{\blue{#1}}} 
\newcommand{\cadr}[2]{{\blue{#2}}} 

\renewcommand{\adr}[1]{#1} 
\renewcommand{\cadr}[2]{#2} 

\newcommand{\defeq}{\coloneqq} 
\newcommand{\grameq}{\Coloneqq} 
\newcommand{\set}[1]{\{#1\}}

\newcommand{\size}[1]{|#1|}

\newcommand{\vsym}{\mathsf{v}}

\renewcommand{\l}{\lambda}
\newcommand{\isub}[2]{\{#1/#2\}}
\renewcommand{\isub}[2]{\{#1{\shortleftarrow}#2\}}
\newcommand{\esub}[2]{[#1/#2]}
\renewcommand{\esub}[2]{[#1{\shortleftarrow}#2]}
\newcommand{\fv}[1]{{\tt fv}(#1)}

\newcommand{\enf}[1]{{\esym}(#1)} 
\newcommand{\nf}{\mathsf{nf}} 
\newcommand{\nfwh}{\mathsf{nf}_{wh}} 

\newcommand{\rootRew}[1]{\mapsto_{#1}}
\newcommand{\Rew}[1]{\rightarrow_{#1}}

\newcommand{\lRew}[1]{\; \mbox{}_{#1}{\leftarrow}\ }

\newcommand{\mult}{\mathsf{m}} 
\newcommand{\expo}{\mathsf{e}} 
\newcommand{\expoabs}{\expo_{\abssym}} 
\newcommand{\expovar}{\expo_{\varsym}} 
\newcommand{\rtom}{\rootRew{\mult}} 
\newcommand{\rtoe}{\rootRew{\expo}} 

\newcommand{\rtobv}{\rootRew{\betav}} 

\newcommand{\tob}{\Rew{\beta}}

\newcommand{\betav}{{\beta_v}} 
\newcommand{\abssym}{\lambda} 
\newcommand{\varsym}{y}

\newcommand{\tobv}{\Rew{\betav}} 

\newcommand{\esym}{{\mathtt e}}
\newcommand{\isym}{i}

\newcommand{\msym}{{\mathtt m}}

\newcommand{\subsym}{{\mathsf{sub}}}

\newcommand{\wsym}{w} 

\newcommand{\shufeqext}{\shufeqext} 

\newcommand{\tom}{\Rew{\mult}}
\newcommand{\toe}{\Rew{\expo}}
\newcommand{\toeabs}{\Rew{\expoabs}}

\newcommand{\tm}{t}
\newcommand{\tmtwo}{u}
\newcommand{\tmthree}{s}
\newcommand{\tmfour}{r}

\newcommand{\tmrone}{a}
\newcommand{\tmrtwo}{b}

\newcommand{\tmronep}{\tmrone'}
\newcommand{\tmrtwop}{\tmrtwo'}

\newcommand{\tmrthree}{p}
\newcommand{\tmrfour}{q}
\newcommand{\tmrthreep}{\tmrthree'}
\newcommand{\tmrfourp}{\tmrfour'}

\newcommand{\tmfirst}{\tm_1}

\newcommand{\tmp}{\tm'}
\newcommand{\tmtwop}{\tmtwo'}

\newcommand{\tmfourp}{\tmfour'}

\newcommand{\tmpfirst}{\tmp_1}

\newcommand{\var}{x}
\newcommand{\vartwo}{y}
\newcommand{\varthree}{z}
\newcommand{\varfour}{w}

\newcommand{\vartwop}{\vartwo'}

\newcommand{\val}{v}
\newcommand{\valtwo}{\val'}

\newcommand{\valof}[1]{\val_#1} 

\newcommand{\valt}{v_t}
\newcommand{\valttwo}{\valt'}

\newcommand{\ctxholep}[1]{\langle #1\rangle}
\newcommand{\ctxhole}{\ctxholep{\cdot}}

\newcommand{\ctx}{C}
\newcommand{\ctxtwo}{\ctx'}

\newcommand{\ctxp}[1]{\ctx\ctxholep{#1}}
\newcommand{\ctxtwop}[1]{\ctxtwo\ctxholep{#1}}

\newcommand{\sctx}{L}
\renewcommand{\sctx}{S}
\newcommand{\sctxtwo}{\sctx'}

\newcommand{\sctxp}[1]{\sctx\ctxholep{#1}}
\newcommand{\sctxtwop}[1]{\sctxtwo\!\ctxholep{#1}}

\newcommand{\sctxONE}{L_{1}} 

\newcommand{\sctxONEp}[1]{\sctxONE\ctxholep{#1}}

\newcommand{\isctx}{L_\isym} 
\newcommand{\isctxtwo}{\isctx'}

\newcommand{\isctxp}[1]{\isctx\ctxholep{#1}}
\newcommand{\isctxtwop}[1]{\isctxtwo\ctxholep{#1}}

\newcommand{\isctxONE}{L_{\isym 1}} 
\newcommand{\isctxONEtwo}{\isctxONE'}

\newcommand{\isctxONEp}[1]{\isctxONE\ctxholep{#1}}
\newcommand{\isctxONEtwop}[1]{\isctxONEtwo\ctxholep{#1}}

\newcommand{\wctx}{W}

\newcommand{\wctxp}[1]{\wctx\ctxholep{#1}}

\newcommand{\arbctxp}[1]{\arbctxp{#1}}
\newcommand{\arbctxtwop}[1]{\arbctxtwop{#1}}

\newcommand{\tctx}{T}
\newcommand{\tctxtwo}{\tctx'}

\newcommand{\tctxp}[1]{\tctx\ctxholep{#1}}
\newcommand{\tctxtwop}[1]{\tctxtwo\ctxholep{#1}}

\newcommand{\evctx}{E}

\newcommand{\evctxp}[1]{\evctx\ctxholep{#1}}

\newcommand{\evctxONE}{\evctx_1}

\newcommand{\evctxONEp}[1]{\evctxONE\ctxholep{#1}}

\newcommand{\levctx}{L}
\newcommand{\levctxtwo}{\levctx'}

\newcommand{\levctxp}[1]{\levctx\ctxholep{#1}}
\newcommand{\levctxtwop}[1]{\levctxtwo\ctxholep{#1}}

\ifthenelse{\boolean{acmartstyle}}{
	
	}{
	
}

\newcommand{\itm}{i}
\newcommand{\itmtwo}{\itm'}

\newcommand{\itmONE}{i_1}
\newcommand{\itmONEtwo}{\itmONE'}

\newcommand{\itmTWO}{i_2}
\newcommand{\itmTWOtwo}{\itmTWO'}

\newcommand{\itmapp}{\mathbb{i}}
\newcommand{\itmapptwo}{\itmapp'}

\newcommand{\itmappONE}{\itmapp_1}
\newcommand{\itmappONEtwo}{\itmappONE'}

\newcommand{\fire}{f}
\newcommand{\firetwo}{g}

\newcommand{\firep}{\fire'}

\newcommand{\vsub}{{\vsym\subsym}} 

\newcommand{\tow}{\Rew{\wsym}} 

\newcommand{\la}[1]{\lambda #1.}

\newcommand{\myproof}[1]{
\ifthenelse{\boolean{omitproofs}}{\begin{IEEEproof} Proof available but omitted for readability. \end{IEEEproof}}{#1}}

\newcommand{\withproofs}[1]{\ifthenelse{\boolean{withproofs}}{#1}{}}

\newcommand{\withoutproofs}[1]{\ifthenelse{\boolean{withproofs}}{}{#1}}

\newcommand{\vsubterms}{\Lambda_\vsub}

\newcommand{\doubt}[1]{}

\newcommand{\letexp}{\sf{let}}

\newcounter{numberone}
\newcounter{numberoneroman}
\newcounter{numberonealph}

\newcommand{\cbn}{CbN\xspace}
\newcommand{\cbv}{CbV\xspace}

\newcommand{\ntype}{L}

\newcommand{\hastype}{\!:\!}

\makeatletter
\newsavebox{\@brx}
\newcommand{\llangle}[1][]{\savebox{\@brx}{\(\m@th{#1\langle}\)}%
  \mathopen{\copy\@brx\kern-0.7\wd\@brx\usebox{\@brx}}}
\newcommand{\rrangle}[1][]{\savebox{\@brx}{\(\m@th{#1\rangle}\)}%
  \mathclose{\copy\@brx\kern-0.7\wd\@brx\usebox{\@brx}}}
\makeatother

\newcommand{\symfont}[1]{\mathtt{#1}}

\newcommand{\vscsym}{\symfont{vsc}}

\newcommand{\ntm}{n}
\newcommand{\ntmtwo}{\ntm'}
\newcommand{\ntmthree}{\ntm''}

\newcommand{\bsvsctax}{{\Downarrow}\textsc{-ax}}
\newcommand{\bsvsctapm}{{\Downarrow}\textsc{-@m}}
\newcommand{\bsvsctapi}{{\Downarrow}\textsc{-@i}}
\newcommand{\bsvsctapvar}{{\Downarrow}\textsc{-@}\var}
\newcommand{\bsvsctese}{{\Downarrow}\textsc{-[e]}}\newcommand{\bsvsctesi}{{\Downarrow}\textsc{-[i]}}

\newcommand{\bswax}{{\Downarrow}\textsc{-ax}}
\newcommand{\bswbeta}{{\Downarrow}\textsc{-@}\betav}
\newcommand{\bswappnf}{{\Downarrow}\textsc{-@nf}}

\newcommand{\ntmONE}{n_1}
\newcommand{\ntmONEtwo}{\ntmONE'}
\newcommand{\ntmONEthree}{\ntmONE''}

\newcommand{\ntmTWO}{n_2}
\newcommand{\ntmTWOtwo}{\ntmTWO'}
\newcommand{\ntmTWOthree}{\ntmTWO''}

\newcommand{\lctx}{L}
\newcommand{\lctxtwo}{\lctx'}

\newcommand{\lctxp}[1]{\lctx\ctxholep{#1}}
\newcommand{\lctxtwop}[1]{\lctxtwo\ctxholep{#1}}

\newcommand{\rctx}{R}

\newcommand{\rctxp}[1]{\rctx\ctxholep{#1}}

\newcommand{\tolw}{\Rew{l}}
\newcommand{\torw}{\Rew{r}}

\newcommand{\Id}{\symfont{I}}

\newcommand{\letin}[3]{{\sf let}\ #1=#2\ {\sf in}\ #3}

\renewcommand{\fire}{\ntm}
\renewcommand{\firetwo}{\ntmtwo}

\let\cal\undefined
\newcommand{\cal}[1]{\mathcal{#1}}
\newcommand{\relsym}{{\cal R}}
\newcommand{\rel}{~\relsym~}
\newcommand{\lassenop}[1]{{#1}^{\cal L}}
\newcommand{\howeop}[1]{{#1}^{\cal H}}

\newcommand{\lasrelsym}{\lassenop\relsym  }
\newcommand{\lasrel}{\, \lasrelsym\,  }

\newcommand{\mlassenopsym}{{\cal L}^M}
\newcommand{\mlassenop}[1]{{#1}^{\mlassenopsym}}
\newcommand{\mlasrelsym}{\mlassenop\relsym  }
\newcommand{\mlasrel}{\, \mlasrelsym\,  }

\newcommand{\ievctx}{S}
\newcommand{\ievctxp}[1]{\ievctx\ctxholep{#1}}
\newcommand{\subred}{\textsc{s}}
\newcommand{\rtos}{\rootRew{\subred}} 
\newcommand{\tos}{\Rew{\subred}} 

\usepackage{proof}

\newcommand{\tmn}{n}

\newcommand{\red}[1]{\rightarrow}

\newcommand{\enft}{\mathit{enf}}
\newcommand{\renft}{\mathit{renf}}

\newcommand{\nafet}{\mathit{nafex}}
\newcommand{\nafext}{\mathit{vscx}}

\newcommand{\openfp}[1]{\ctxholep{#1}_{\enft}}

\newcommand{\relenf}{\openfp\relsym}

\newcommand{\isubst}[2]{\isub#2#1}

\newcommand{\opnafep}[1]{\ctxholep{#1}_{\nafet}}

\newcommand{\opnafexp}[1]{\ctxholep{#1}_{\nafext}}
\newcommand{\opnafex}{\ctxhole_{\nafext}}
\newcommand{\relnafex}{\opnafexp\relsym}

\newcommand{\mlasrelnafexsym}{\opnafexp\mlasrelsym}
\newcommand{\mlasrelnafex}{\, \mlasrelnafexsym \,}

\newcommand{\enfbisim}{\simeq_{\enft}}

\newcommand{\nafebisim}{\simeq_{\nafet}}

\newcommand{\tovsce}{\Rew{\equivsone}}

\newcommand{\symequivsone}{@\symfont{l}}
\newcommand{\equivsone}{\equiv_{\symequivsone}}
\newcommand{\equivsthree}{\equiv_{v@\symfont{r}}}
\newcommand{\equivexsthree}{\equiv_{@\symfont{r}}}
\renewcommand{\symequivsone}{@l}
\renewcommand{\equivsthree}{(\equiv_{@r})_{|ltr}}
\renewcommand{\equivexsthree}{\equiv_{@r}}

\newcommand{\equivcom}{\equiv_{{\mathtt{com}}}}
\renewcommand{\equivcom}{\equiv_{com}}

\newcommand{\equivx}{\equiv_{M}}
\newcommand{\equivlid}{\equiv_{lid}}
\newcommand{\equivass}{\equiv_{ass}}
\newcommand{\equivrad}{\equiv_{rad}}
\newcommand{\equivlad}{\equiv_{lad}}
\newcommand{\equivexrad}{\equiv_{exrad}}
\newcommand{\equivom}{\equiv_{\Omega}}
\newcommand{\equivomv}{\equiv_{\Omega_v}}

\newcommand{\equivdup}{\equiv_{dup}}
\newcommand{\streq}{\equiv_{str}}

\newcommand{\leqcbn}{\precsim_{cbn}}
\newcommand{\eqcbn}{\simeq_{cbn}}
\newcommand{\leqncbv}{\precsim_{nai}}
\newcommand{\eqncbv}{\simeq_{nai}}
\newcommand{\leqenf}{\precsim_{\enft}}
\newcommand{\leqrenf}{\precsim_{\renft}}
\newcommand{\eqenf}{\simeq_{\enft}}

\newcommand{\leqnafex}{\precsim_{\nafext}}

\newcommand{\tovsc}{\Rew{\vscsym}}

\newcommand{\lassenopsym}{{\cal L}}
\newcommand{\sclift}{\lassenopsym_{\textsc{lift}}}
\newcommand{\scvar}{\lassenopsym_{\textsc{var}}}
\newcommand{\scabs}{\lassenopsym_{\l}}
\newcommand{\scapp}{\lassenopsym_{@}}
\newcommand{\scsub}{\lassenopsym_{\textsc{sub}}}

\newcommand{\scesub}{\lassenopsym_{\textsc{esub}}}
\newcommand{\scequivx}{\lassenopsym_{\equivx}}

\newcommand{\msclift}{\mlassenopsym_{\textsc{lift}}}
\newcommand{\mscvar}{\mlassenopsym_{\textsc{var}}}
\newcommand{\mscabs}{\mlassenopsym_{\l}}

\newcommand{\mscapp}{\mlassenopsym_{@}}
\newcommand{\mscesub}{\mlassenopsym_{\textsc{esub}}}
\newcommand{\mscsub}{\mlassenopsym_{\textsc{sub}}}
\newcommand{\mscequivx}{\mlassenopsym_{\equivx}}

\newcommand{\leqc}{\precsim_C}
\newcommand{\eqcv}{\simeq_C^v}
\newcommand{\leqcv}{\precsim_C^v}

\newcommand{\curryfix}{Y_v}
\newcommand{\curryfixaux}{\Xi_v}
\newcommand{\turingfix}{\Theta_v}

\newcommand{\curryfixn}{Y}
\newcommand{\curryfixauxn}{\Xi}
\newcommand{\turingfixn}{\Theta}

\newcommand{\bssym}{\Downarrow} 

\newcommand{\bs}[1]{\bssym^{#1}}

\newcommand{\bswh}{\Downarrow_{wh}\,} 
\newcommand{\bswhdiv}{\not\bswh} 

\newcommand{\bsweak}{\Downarrow_{w}\,} 
\newcommand{\bsweakdiv}{\not\bsweak} 

\newcommand{\bswleftsym}{\Downarrow_{l}} 
\newcommand{\bswleft}[1]{\bswleftsym^{#1}}
\newcommand{\bswlefts}{\bswleftsym\,} 
\newcommand{\bswleftdiv}{\not\bswlefts}

\newcommand{\bswsym}{\Downarrow_{{w}}} 
\newcommand{\bsws}{\bswsym\,}
\newcommand{\bsw}[1]{\bswsym^{#1}}

\newcommand{\bsvsctsym}{\Downarrow_{\vscsym}}
\newcommand{\bsvscts}{\bsvsctsym\,}
\newcommand{\bsvsct}[1]{\bsvsctsym^{#1}}
\newcommand{\bsvsctdiv}{\not\bsvscts} 

\newcommand{\bsvscsym}{\Downarrow_{\vscsym}}
\newcommand{\bsvscs}{\bsvscsym\,}

\newcommand{\VSC}{\text{VSC}\xspace}

\newcommand{\bohm}{B{\"o}hm }

\renewcommand{\enf}{$\mathsf{enf}$\xspace}
\newcommand{\Enf}{$\mathsf{Enf}$\xspace}

\newcommand{\nafe}{$\mathsf{nafex}$\xspace}

\newcommand{\nafex}{$\mathsf{nafex}$\xspace}
\newcommand{\Nafex}{$\mathsf{Nafex}$\xspace}

\newcommand{\equivetav}{\equiv_{\eta_v}}
\newcommand{\equivbetav}{=_{\betav}}

\newcommand{\cbnfp}[1]{\ctxholep{#1}_{cbn}}
\newcommand{\relcbn}{\cbnfp\relsym}
\newcommand{\ncbvfp}[1]{\ctxholep{#1}_{nai}}
\newcommand{\relncbv}{\ncbvfp\relsym}

\newcommand{\lasrelncbvsym}{\ncbvfp\lasrelsym}
\newcommand{\lasrelncbv}{\, \lasrelncbvsym \,}

\newcommand{\etav}{\eta_v}

\newcommand{\contextualsym}{\mathcal{C}}

\newcommand{\equivcp}[1]{\simeq_\contextualsym^{#1}}
\newcommand{\leqcp}[1]{\precsim_\contextualsym^{#1}}

\newcommand{\mtype}{M}
\newcommand{\mtypetwo}{N}

\newcommand{\emptytype}{[~]}

\newcommand{\multitype}[2]{[{#2}_{1},\ldots,{#2}_{#1}]}

\newcommand{\ltype}{L}
\newcommand{\ltypetwo}{\ntype'}

\newcommand{\vartype}{X}

\newcommand{\typectx}{\Gamma}
\newcommand{\typectxtwo}{\Delta}

\newcommand{\typeder}{\pi}
\newcommand{\typederp}{\typeder'}
\newcommand{\typedertwo}{\sigma}
\newcommand{\typederthree}{\rho}

\newcommand{\typederthreep}{\typederthree'}

\newcommand{\typingruleApp}{@}
\newcommand{\typingruleAx}{\mathsf{ax}}
\newcommand{\typingruleAbs}{\lambda}
\newcommand{\typingruleMany}{\mathsf{many}}
\newcommand{\typingruleES}{\mathsf{es}}

\newcommand{\types}{\vdash}

\newcommand{\leqtype}{\precsim_{type}}
\newcommand{\equivtype}{\simeq_{type}}

\newcommand{\interp}[1]{\llbracket #1 \rrbracket}

\newcommand{\vscx}{\symfont{mir}_{M}}

\newcommand{\leqvscx}{\precsim_{\vscx}}

\newcommand{\opvscxp}[1]{\ctxholep{#1}_{\vscx}}

\newcommand{\relvscx}{\opvscxp\relsym}

\newcommand{\net}{$\mathsf{net}$\xspace}
\newcommand{\Net}{$\mathsf{Net}$\xspace}

\newcommand{\tderiv}{\Phi}

\newcommand{\derive}[2]{#1 \derives #2}
\newcommand{\concl}[4]{\derive{#1}{#2 \vdash #3 \hastype #4}}

\newcommand{\derives}{\vartriangleright}

\newcommand{\xmark}{\text{\ding{55}}}
\newcommand{\towh}{\Rew{wh}}

\usepackage{cleveref}

\renewcommand{\nafet}{\mathit{net}}
\renewcommand{\nafext}{\mathit{nafex}}
\newcommand{\nett}{\mathit{net}}

\newcommand{\leqnet}{\precsim_{\nett}}
\newcommand{\eqnet}{\simeq_{\nett}}

\renewcommand{\nafex}{mir$_{M}$\xspace}
\renewcommand{\Nafex}{Mir$_{M}$\xspace}

\newcommand{\nafexp}[1]{mir$_{#1}$\xspace}

\renewcommand{\leqnafex}{\precsim_{\vscx}}

\renewcommand{\nafext}{\vscx}

\renewcommand{\net}{net\xspace}
\renewcommand{\Net}{Net\xspace}

\begin{document}

\title{Normal Form Bisimulations by Value}

\author{Beniamino Accattoli}
\affiliation{%
  \institution{Inria \& LIX, École Polytechnique, UMR 7161}
  \country{France}}
\email{beniamino.accattoli@inria.fr}

\author{Adrienne Lancelot}
\affiliation{%
	\institution{Inria \& LIX, École Polytechnique, UMR 7161}
	\country{France}}
\email{adrienne.lancelot@inria.fr}

\author{Claudia Faggian}
\affiliation{%
  \institution{ IRIF, CNRS, Université Paris Cité, F-75013 Paris}
  \country{France}}
\email{faggian@irif.fr}

\renewcommand{\shortauthors}{Accattoli, Faggian, and Lancelot}

\begin{abstract}
Normal form bisimilarities are a natural form of program equivalence resting on open terms, first introduced by Sangiorgi in call-by-name. The literature contains a normal form bisimilarity for  Plotkin's call-by-value $\lambda$-calculus, Lassen's \emph{enf bisimilarity}, which validates all of Moggi's monadic laws and can be extended to validate $\eta$. It does not validate, however, other relevant principles, such as the identification of meaningless terms---validated instead by Sangiorgi's bisimilarity---or the commutation of $\letexp$s. These shortcomings are due to issues with open terms of Plotkin's calculus. We introduce a new call-by-value normal form bisimilarity, deemed \emph{net bisimilarity}, closer in spirit to Sangiorgi's and satisfying the additional principles. We develop it on top of an existing formalism designed for dealing with open terms in call-by-value. It turns out that enf and net bisimilarities are \emph{incomparable}, as net bisimilarity does not validate Moggi's laws nor $\eta$. Moreover, there is no easy way to merge them. To better understand the situation, we provide an analysis of the rich range of possible call-by-value normal form bisimilarities, relating them to Ehrhard's relational model.
\end{abstract}



\maketitle

\section{Introduction}

The study of program equivalences for $\l$-calculi is an important topic where semantical and operational techniques meet. Properties of program equivalences are notoriously difficult to prove.  Even the equivalence of two terms might be challenging to establish, if the notion of equivalence is \citeauthor{morris1968lambda}' contextual equivalence \citeyearpar{morris1968lambda}, or some variant still involving a universal quantification, such as \citeauthor{abramsky-lazy-90}'s applicative bisimilarity \citeyearpar{abramsky-lazy-90}. Another difficulty is the fact that properties of program equivalences are brittle, as they are not preserved by extensions of the calculus under study, nor by restrictions, and not even by changing the evaluation strategy within the same calculus.

It is well-known that applicative bisimilarity is fully abstract for contextual equivalence in the untyped call-by-name weak $\l$-calculus (where \emph{weak} stands for reduction only out of abstractions, which is standard in functional languages) \cite{abramsky-lazy-90}, as well as in Plotkin's call-by-value weak $\l$-calculus \cite{DBLP:journals/fuin/EgidiHR92,DBLP:books/cu/12/Pitts12}. Therefore, one might be led to think that the call-by-name/call-by-value switch is quite robust in the weak setting, even if it is known that its robustness breaks in a probabilistic setting, as shown by \citet{DBLP:conf/popl/LagoSA14}. 

This paper stems from the observation that another natural program equivalence, \citeauthor{SANGIORGI-normal-form-bisimulation}'s \emph{normal form bisimilarity} \citeyearpar{SANGIORGI-normal-form-bisimulation} (shortened to nf-bisimilarity), behaves differently in call-by-name (shortened to \cbn) and call-by-value (\cbv), already in the untyped effect-free weak case.

\paragraph{Normal Form Bisimilarity} Normal form bisimulations are program equivalences that, instead of comparing terms \emph{externally}, depending on how they behave \emph{in contexts}, compare them \emph{internally}, by looking at the structure of their (infinitary) \emph{normal forms}. A distinctive feature of nf-bisimulations is that they directly manipulate \emph{open terms}, to the point that \citeauthor{SANGIORGI-normal-form-bisimulation} rather used to call them \emph{open bisimulations} in his seminal paper \citeyearpar{SANGIORGI-normal-form-bisimulation}.

\cadr{It is known that Sangiorgi's}{Sangiorgi's} \cbn nf-bisimilarity is not fully abstract for contextual equivalence, being sound but not complete. 
 The failure of full abstraction is compensated by the fact that \cbn nf-bisimilarity is easier to establish than applicative similarity, because of the absence of quantification over arguments. Typically, it is easy to show that different fix-points combinators---which are the paradigmatic terms with infinitary normal forms---are nf-bisimilar, while it is hard to show that they are applicative bisimilar.

There exists a \cbv nf-bisimilarity, \citeauthor{LassenEnf}'s \emph{enf bisimilarity} $\eqenf$ \citeyearpar{LassenEnf}\cadr{, which---as Sangiorgi's---is  sound but not complete for \cbv contextual equivalence $\eqcv$, and which is considered the \cbv nf-bisimilarity of reference.}{, which is considered the \cbv nf-bisimilarity of reference. Like Sangiorgi's, it is sound but not complete for \cbv contextual equivalence $\eqcv$.} It is not \cadr{the}{an} off-the-shelf adaptation of Sangiorgi's to \cbv, as it has an \emph{extra clause} about evaluation contexts, exploiting the well-known issues of Plotkin's \cbv $\l$-calculus with open terms; issues that are discussed at length by \citet{accattoli+guerrieri-opencbv,Accattoli-Guerrieri-TypesFireballs,DBLP:journals/pacmpl/AccattoliG22}. 

The incompleteness of both \cbn and \cbv nf-bisimulations for the weak $\l$-calculus at first sight suggests that they are robust with respect to the \cbn/\cbv switch. In fact, they are \emph{not}. \cadr{The ways in which the two bisimilarities are incomplete, indeed, are very different. }{The two bisimilarities are indeed incomplete for very different reasons.}

\paragraph{Incompleteness of Enf.} \cadr{There are in fact various ways in which enf bisimilarity is incomplete. We can mention at least the following four.}{Here are four important reasons why enf bisimilarity is incomplete.}
\begin{enumerate}
\item \emph{Eta}: \cadr{it}{enf bisimilarity} does not validate\cadr{ (the \cbv variant $\etav$ of) }{ $\etav$, the \cbv variant of }$\eta$-equivalence, even if Lassen himself shows how to \cadr{correct this aspect}{refine enf to correct this aspect}. 

\item \emph{\cbn duplication}: enf bisimilarity does not equate duplications of terms that are not \emph{by value}: terms such as $(\la\var\vartwo\var\var)\tm$ and $\vartwo\tm\tm$, are contextually equivalent for any $\tm$, also for terms $\tm$ that are \emph{not} values, while in general they are not enf bisimilar. In richer \cbv settings with state or probability, contexts discriminate more, and those terms are not contextually equivalent, but in the pure case they are. This aspect, referred here to as \emph{cost-sensitiveness}, is not necessarily a drawback, as it keeps the program equivalence closer to \cbv intuition. 

\item \emph{Commutation of (independent) $\letexp$s}: in the pure case---\cadr{as well as}{or} in \adr{the }presence of commutative effects---contextual equivalence validates the following commutation $\equivcom$ (presented with $\letexp$s for readability, but easily presentable without them, expanding $\letexp$s as $\beta$-redexes):
\begin{center}
$\letin\var\tmtwo{\letin\vartwo\tmthree\tm} \ \ \equivcom\ \ \letin\vartwo\tmthree{\letin\var\tmtwo\tm}$\ \ \ \  if $\var\notin\fv\tmthree$ and $\vartwo\notin\fv\tmtwo$.
\end{center}
Enf bisimulations do not equate these terms.

\item \emph{$\Omega$-terms/meaningless terms}: enf bisimilarity does not equate terms contextually equivalent to the paradigmatic looping/meaningless $\l$-term $\Omega$, referred here to as \emph{$\Omega$-terms}. This is connected to the mentioned issues of Plotkin's calculus with open terms, as they cause \emph{false normal forms}, that is, $\Omega$-terms that are normal. In \cbn, instead, all $\Omega$-terms diverge.
\end{enumerate}
Point 2 is acceptable, because a bisimilarity validating \cbn duplication in \cbv needs to go beyond comparing the structure of normal forms, as the example shows, thus necessarily departing from the format of nf-bisimilarities. Thus, we shall consider cost-sensitivity as \emph{inherent} to \cbv nf-bisimilarities, which need not be corrected.
Point 3 is specific to \cbv, because in \cbn $\letexp$s disappear during evaluation, while in \cbv with open terms they might end up in normal forms. It is thus disappointing that commuting $\letexp$s are not captured. More precisely, it would be desirable to have a notion of nf-bisimilarity where the validation of commuting $\letexp$s can be added/removed modularly, as to adapt to different \cbv settings (pure/(non-)commutative-effects). Point 4 is a stark difference between \cbv and \cbn, because Sangiorgi's bisimilarity \emph{does equate} the $\Omega$-terms of weak \cbn. 

\paragraph{$\Omega$-terms} The failure of enf bisimilarity with respect to $\Omega$-terms is relevant, as all natural denotational models and program equivalences identify $\Omega$-terms (of the modelled notion of evaluation). Roughly, different semantics are built by distinguishing various partitions of non-$\Omega$-terms, but they put all terms contextually equivalent to the idle loop $\Omega$ in the same class---indeed, what would be the advantage of partitioning them in classes? For strong \cbn evaluation, in particular, the equational identification of $\Omega$-terms is the concept around which \citeauthor{Barendregt84}'s classic book \citeyearpar{Barendregt84} is built (therein $\Omega$-terms are called \emph{unsolvable terms}, see \refsect{preliminaries}). 
 
\paragraph{Two Possible Approaches.} To overcome some of the mentioned incompleteness via a nf-bisimilarity, one can act on two different levels: \emph{changing the underlying calculus}, which provides the normal forms to compare, or \emph{changing the nf-bisimilarity}, that is, the notion of comparison. For instance, \citet{DBLP:conf/fossacs/BiernackiLP19} address the problem of $\Omega$-terms (called therein \emph{deferred diverging terms}) by providing an alternative definition of enf bisimilarity. Commuting $\letexp$s, however, do not seem to be capturable by a similar tweak of enf-bisimilarity, as the commutation does not behave well in Plotkin's \cbv \adr{ calculus}. At first sight, changing the underlying calculus is less viable, as in general it  changes the notion of contextual equivalence. 

\paragraph{The Value Substitution Calculus.} The starting point of this paper is the \emph{value substitution calculus} (shortened to VSC), a \cbv $\l$-calculus due to \citet{accattoli+paolini-vsc} and related to linear logic proof nets \cite{Accattoli-proofnets}. The VSC solves the issues  of Plotkin's calculus with open terms via an extension of the rewriting rules, while--crucially--retaining the same notion of contextual equivalence \cite{DBLP:journals/pacmpl/AccattoliG22}. Moreover, the proof nets inception of the calculus makes it easy to deal with the commutation of $\letexp$s, which can be modularly added via a notion of \emph{structural equivalence}, compatible  with the rewriting rules by design. Therefore, the VSC is a natural candidate for designing a \cbv nf-bisimilarity improving on some of the incompleteness of Lassen's by changing the underlying calculus.

\paragraph{Net Bisimilarity.} \cadr{By using}{Using} the VSC, we introduce a \cbv nf-bisimilarity validating the commutation of $\letexp$s and identifying $\Omega$-terms. The obtained \emph{net bisimilarity} $\eqnet$ and the proof of its \emph{compatibility} (that is, its stability by context closure)---the challenging property to prove for bisimilarities---are the main contributions of this paper. Compatibility implies soundness with respect to contextual equivalence, and it is proved adapting \citeauthor{lassen1999bisimulation}'s variant \citeyearpar{lassen1999bisimulation} for nf-bisimilarities of Howe's method. As it is often the case for nf-bisimilarities, ours is sound but not complete. In particular---as for enf bisimilarity---it is \emph{cost-sensitive}. 

The crafting of net bisimilarity rests on a sophisticated analysis of \cbv and the VSC. We start with the off-the-shelf adaptation of Sangiorgi's nf-bisimilarity to the VSC, \emph{without Lassen's extra clause}, as it is unclear how to adapt the clause to the VSC. We then refine the adaptation by comparing normal forms modulo the \adr{(proof nets)} structural equivalence of the VSC---which includes commuting $\letexp$s---whence the name \emph{net} bisimilarity. We actually go further, introducing a \emph{parametric} nf-bisimilarity, where parts of the structural equivalences can be turned off and on at will---because some (such as commuting $\letexp$s) fail in extensions of \cbv with non-commutative effects---thus defining a \emph{family} of \cbv nf-bisimilarities, all proved compatible via a \emph{single abstract proof}. 

Our result is however more a new beginning than the end of the story: \net bisimilarity, indeed, is \emph{not} a refinement of Lassen's $\eqenf$. In fact, the two are \emph{incomparable}, because \net bisimilarity is incomplete in yet some other ways, which are instead validated by Lassen's. 

\paragraph{Moggi's Laws} 
It is well-known that Plotkin's \cbv $\l$-calculus is defective, and not just because of open terms. Plotkin himself showed the incompleteness of his continuation-passing translation \cite{PLOTKIN1975}. To both solve the issue and modeling extensions with effects, Moggi extended Plotkin's calculus with equations corresponding to laws for monads \cite{Moggi88tech,DBLP:conf/lics/Moggi89}, that are sound for contextual equivalence. Lassen's enf bisimulations verify these laws, showing that it is possible to capture the laws in the notion of nf-bisimulation rather than by changing the underlying calculus. In particular, it is Lassen's extra clause that allows enf to capture Moggi's laws.

Moggi's laws, however, are not rules of the VSC, and are not captured by net bisimilarity. Additionally, it is unclear how to extend net bisimulations as to satisfy Moggi's laws. Once more, there are two options: extending the underlying calculus (the VSC) or the nf-bisimilarity ($\eqnet$). Both options however break  properties that are crucial for the proof that net bisimilarity is compatible. Another disappointing fact is that the addition of $\etav$ to net bisimilarity requires one of Moggi's laws, so it is also unclear how to add $\etav$ to net bisimulations.

\paragraph{Impasse and Beyond} Summing up, it seems that one cannot have the cake (a framework for nf-bisimilarity satisfying either commuting $\letexp$s and $\Omega$-terms, or Moggi's laws and $\etav$) and eat it too (extend the framework as to capture the missing half), or at least it is far from evident how to do it. The exploration of such an \emph{impasse} is the other main contribution of the paper. We study two further program equivalences, a naive bisimilarity $\eqncbv$ and the program equivalence $\equivtype$ induced by a model, which are sort of the intersection and the union of enf and net bisimilarities.

\paragraph{Naive Bisimilarity $\eqncbv$ = No Cake and No Eating} We consider the off-the-shelf adaptation of Sangiorgi's \cbn nf-bisimulations to 
Plotkin's \cbv, obtaining \emph{naive nf-bisimulations}, which are strictly weaker than both enf and net bisimulations, as they do not identify commuting $\letexp$s, $\Omega$-terms, Moggi's laws, nor $\etav$. The experiment is instructive because we show that naive bisimulations, despite their weakness, are enough to provide easy proofs of bisimilarity for fix-point combinators---Lassen's extra clause plays no role in that. Naive bisimilarity also gives us the opportunity to gently introduce the proof technique that we use for net bisimilarity.

\paragraph{Type Equivalence $\equivtype$ = Cake and Eating, Universally} At the other end of the spectrum, we investigate the program equivalence given by the equational theory of Ehrhard's \cbv relational model \cite{DBLP:conf/csl/Ehrhard12}. We call it \emph{type equivalence} because the model is presented as a multi type system (a variant of intersection types). Such a model was already extensively studied in connection with the VSC by \citet{Accattoli-Guerrieri-TypesFireballs,DBLP:journals/pacmpl/AccattoliG22}. Its equational theory does not have a presentation via nf-bisimulations, nor any other characterization, but it is nonetheless possible to study it via the multi type system. It turns out that type equivalence, similarly to nf-bisimilarities, is compatible and sound, but not complete for contextual equivalence, because it is cost-sensitive. It is not an easily usable equivalence, as it is based on a \emph{universal} quantification over the typings for a term, but it provides interesting insights.

Our results are that both enf and net bisimilarities are \emph{included} in type equivalence. Therefore, the two bisimilarities are joinable. Since both are sound, they are obviously joinable in a cost-insensitive setting, as they are both included in contextual equivalence. Our results show that they are also joinable in a \emph{cost-sensitive} program equivalence, thus suggesting that a nf-bisimilarity joining the two might be possible. Crafting it, and especially proving that it is compatible, is left to future work.


%
%
%
%

The following table sums up the situation.

\begin{center}
	\begin{tabular}{ |r@{\hspace{.3cm}}|c@{\hspace{.3cm}}c@{\hspace{.3cm}}c@{\hspace{.3cm}}c@{\hspace{.3cm}}c@{\hspace{.3cm}}| } 
		\hline
		&  $\eqncbv$ & $\eqenf$ & $\eqnet$ & $\equivtype$ & $\eqcv$\\
		\hline
		Moggi's left identity law $\Id \tm \equivlid \tm$ ? & \RED{\xmark} & $\blue\checkmark$ & \RED{\xmark} & $\blue\checkmark$& $\blue\checkmark$ \\ 

		\hline
		Identification of $\Omega$-terms and commuting $\letexp$s? & \RED{\xmark} &  \RED{\xmark} & $\blue\checkmark$ & $\blue\checkmark$& $\blue\checkmark$ \\ 

		\hline
		Universal quantification in the definition? & \RED{\xmark}  & \RED{\xmark} & \RED{\xmark}  & $\blue\checkmark$& $\blue\checkmark$ \\ 
		\hline
		CbN duplication $(\la\var\vartwo\var\var)\tm\equivdup\vartwo\tm\tm$ ? & \RED{\xmark}  & \RED{\xmark} & \RED{\xmark}  & \RED{\xmark}& $\blue\checkmark$ \\ 
		\hline
	\end{tabular}
\end{center}

\paragraph{Related Work} Beyond the already cited papers, nf-bisimilarity is studied in variants and extensions of the $\l$-calculus by \citet{lassen1999bisimulation,DBLP:journals/entcs/Lassen06,DBLP:conf/lics/Lassen06}, \citet{DBLP:conf/csl/LassenL07,DBLP:conf/lics/LassenL08}, \citet{DBLP:conf/flops/BiernackiL12} and \citet{DBLP:journals/taosd/JagadeesanPR09}, and in relationship to game semantics by \citet{DBLP:conf/csl/LevyS14}, \citet{DBLP:conf/lics/JaberM21}, and \citet{DBLP:conf/csl/JaberS22}. The presence of state in \cite{DBLP:conf/birthday/StovringL09,DBLP:conf/fossacs/BiernackiLP19} makes enf-like bisimilarities fully abstract (as it makes contextual equivalence cost-sensitive and commutation of $\letexp$s invalid). Lassen's enf bisimilarity is studied with respect to $\eta$-equivalence by \citet{DBLP:journals/lmcs/BiernackiLP19}, extensions with effects by \citet{DBLP:conf/esop/LagoG19}, and the $\pi$-calculus by \citet{DBLP:journals/tcs/DurierHS22}. 

\cbv multi types are also used by Kesner and co-authors \cite{DBLP:conf/flops/BucciarelliKRV20,DBLP:conf/fscd/KesnerP22,DBLP:conf/csl/KesnerV22,DBLP:journals/pacmpl/ArrialGK23} and \citet{DBLP:conf/lfcs/Diaz-CaroMP13}.

A notion of \cbv B\"ohm tree, inducing a program equivalence similar to nf-bisimilarity is proposed by \citet{bohmtree-cbv}. Their equivalence is in between our naive and net bisimilarities. They conjecture that it characterizes type equivalence. Our results \emph{refute} such a conjecture: with respect to the benchmarks of \refsect{benchmarks}, their equivalence equates all $\Omega$-terms but not commuting $\letexp$s---which are validated by net bisimilarity---nor Moggi's left identity law $\equivlid$.

Very recently, \citet{Nikos2023} introduced a \emph{complete} \cbv bisimilarity, presented as a \emph{nf-bisimilarity}. It is however different from Sangiorgi's and Lassen's, as Koutaval et al.'s bisimilarity considers more than the structure of normal forms, which---as already pointed out---is mandatory to validate \cbn duplication and be complete in \cbv. Their definition, indeed, is rather based on game semantics tools and environmental bisimulations. Moreover, it addresses the different setting of \cbv PCF, which is typed, and it is not clear whether the result smoothly adapts to the untyped pure \cbv $\l$-calculus.


\paragraph{Proofs} Omitted proofs are in the additional material associated with this submission on HotCRP. In case of acceptance, a version of this paper with all proofs will be put on arXiv.org.

\section{Preliminaries}
\label{sect:preliminaries}
\paragraph{Contexts} All along the
paper we use (many notions of) \emph{contexts}, \ie terms with exactly one hole, noted $\ctxhole$. Plugging a term $\tm$ in a context $\ctx$, noted $\ctxp{\tm}$, possibly~captures free variables of $\tm$. For instance $(\la\var\ctxhole)\ctxholep\var = \la\var\var$, while $(\la\var\vartwo)\isub\vartwo\var = \la\varthree\var$.

\paragraph{Preorders} We shall mostly deal with simulations, rather than bisimulations, since the result for equivalences shall always follow by simply considering the symmetric notions. Given a preorder/similarity $\precsim_X$ for some $X$, we  denote with $\simeq_X$ the corresponding equivalence/bisimilarity.

\paragraph{(In)Equational Theories and Compatibility} Good program preorder/equivalences are \emph{(in)equational theories}, that is, they contain the reduction of the calculus and they are \emph{compatible}, defined as: if $\tm\precsim\tmp$ then $\ctxp\tm \precsim\ctxp\tmp$ for all contexts $\ctx$, that is, that they are stable by context closure. Reduction is usually trivially included in similarities while compatibility is usually non-trivial to prove.

\paragraph{Contextual Equivalence} The standard of reference for program equivalences is contextual equivalence, that can be defined abstractly as follows.
\begin{definition}[Contextual Preorder and Equivalence] Given a language of terms $\mathcal{T}$ with its associated notion of contexts $\ctx$ and predicate stating the termination of evaluation $\bs{}_{e}$, we define the associated \emph{contextual preorder} $\leqcp{e}$ and \emph{contextual equivalence} $\equivcp{e}$ as follows:
	\begin{itemize}
		\item $\tm \leqcp{e} \tmp$ if $\ctxp\tm ~{\bs{}_{e}}$ implies $\ctxp\tmp ~{\bs{}_{e}}$ for all contexts $\ctx$ such that $\ctxp{\tm}$ and $\ctxp\tmp$ are closed terms. 
		\item $\tm \equivcp{e} \tmp$ is the equivalence relation induced by $\leqc$, that is, $\tm \equivcp{e} \tmp \iff \tm \leqcp{e} \tmp$ and $\tmp \leqcp{e} \tm$.
	\end{itemize}
\end{definition}

A relation $\relsym$ is sound for the contextual preorder $\leqcp{e}$ when $\relsym \subseteq \leqcp{e}$. Soundness follows from compatibility and \emph{adequacy for $\bs{}_e$}, defined as: if $\tm\rel\tmp$  then $\tm \bs{}_e$ implies $\tmp \bs{}_e$.

\begin{proposition}
	\label{prop:congruence-included-contextual-equivalence}
	Let $\precsim$ be a compatible and adequate preorder. Then $\precsim \subseteq \leqcp{e}$.
\end{proposition}

\begin{proof}
	Suppose $\tmtwo \precsim \tmtwop$.
	Let $\ctx$ any closing context for $\tmtwo$ and $\tmtwop$.
	By compatibility, $\ctxp{\tmtwo} \precsim \ctxp{\tmtwop}$.
	By adequacy, $\ctxp{\tmtwo} \bs{}_e$ implies $\ctxp{\tmtwop} \bs{}_e$, that is, $\tmtwo \leqcp e \tmtwop$.
\end{proof}

We shall see that normal form simulations are defined in such a way that they are adequate, so that soundness follows directly from compatibility. Note that soundness without compatibility is useless: the relation $\relsym \defeq \set{(\var\vartwo,\var\vartwo)}$ is sound but not compatible. 

\paragraph{Diamond} A rewriting notion that shall play a role is the \emph{diamond property}, which is the following one-step strengthening of confluence for a reduction $\to$: if $\tm \to \tmtwo_1$, $\tm\to\tmtwo_2$, and $\tmtwo_1 \neq \tmtwo_2$ then exists $\tmthree$ such that $\tmtwo_1 \to \tmthree$ and $\tmtwo_2 \to \tmthree$. Some well-known facts: the diamond property implies confluence but not vice-versa; if $\to$ is diamond and there is a terminating reduction from $\tm$ then there are no diverging reductions from $\tm$; all reductions to normal form, if any, have the same length. Roughly, the diamond property is a relaxed form of determinism, where non-deterministic choices have no impact on the result nor on the length of the evaluation leading to it.

\paragraph{$\Omega$-terms, Solvability, and Scrutability} A cornerstone of the theory of the (\cbn) $\l$-calculus is the study of what here we call \emph{$\Omega$-terms}, that is, terms that are contextually equivalent to $\Omega$, the paradigmatic looping term. Such a notion has been extensively studied for \emph{head} (\cbn) contextual equivalence, for which $\Omega$-terms are better known as \emph{unsolvable terms}, and often labeled as \emph{meaningless terms}. Actually, the definition of unsolvable term (here omitted) is different, as it does not mention $\Omega$ nor contextual equivalence, but equivalent to the one of $\Omega$-term. Equational theories that identify all unsolvable terms were first studied by \citet{Wad:SemPra:71,DBLP:journals/siamcomp/Wadsworth76} and \citet{DBLP:books/daglib/0016519,solvability-barendregt}, and it is the leading theme of \citeauthor{Barendregt84}'s book \citeyearpar{Barendregt84}. Unsolvable terms can be characterized as those terms that diverge with respect to head reduction, as proved by Wadsworth. 

Adopting \emph{weak} head (\cbn) contextually equivalence gives a different set of $\Omega$-terms. As for the head case, these $\Omega$-terms coincide with a natural set of terms defined without mentioning $\Omega$, namely \emph{(\cbn) inscrutable terms} (for the definition see \refapp{app-vsc}), first studied by \citet{parametricBook} (under the name \emph{not potentially valuable terms}), who provide the following useful characterization, akin to Wadsworth's for unsolvable terms. 
\begin{theorem}[Diverging characterization of \cbn scrutability]
\label{thm:cbn-scrutability-characterization}
A term $\tm$ is \cbn inscrutable (thus a weak \cbn $\Omega$-term) if and only if the weak head reduction of $\tm$ diverges.
\end{theorem}
For instance, $\Omega$ is inscrutable but $\la\var\Omega$ and $\var\Omega$ are not (while $\la\var\Omega$ is unsolvable). From the characterization, it follows that every inscrutable term is unsolvable, but not vice-versa.

In \cbv, there are analogous notions of unsolvable and inscrutable terms, but they lack analogous \emph{diverging characterizations} in Plotkin's \cbv $\l$-calculus. They have been thoroughly studied by \citet{DBLP:journals/pacmpl/AccattoliG22} in the VSC, where instead they admit diverging characterizations. Surprisingly, an equational theory identifying \cbv unsolvable terms is necessarily inconsistent. The right notion of meaningless term in \cbv is actually given by \cbv inscrutable terms. In \refapp{app-vsc}, one can find the technical definition of \cbv inscrutable terms as well as the proof that they coincide with \cbv $\Omega$-terms (relying on the VSC introduced in \refsect{vsc}). 

In this paper, we shall not consider unsolvable terms, as we only consider \emph{weak} reductions. Thus, \cbn/\cbv \emph{$\Omega$-terms} shall always be relative to weak \cbn/\cbv contextually equivalence, and coinciding with \cbn/\cbv inscrutable terms.

\section{Background About Normal Form Bisimulations}
Normal form bisimulations are program equivalences that, instead of comparing terms \emph{externally}, depending on how they behave \emph{in contexts}, they compare them \emph{internally}, by looking at the structure of their (infinitary) \emph{normal forms}. Let us give the simplest possible example.

Let $\towh$ be weak head reduction, also known as \cbn evaluation, which is a deterministic reduction and it is defined as  $(\la\var\tm)\tmtwo \tmthree_1 \ldots \tmthree_k\towh \tm\isub\var\tmtwo \tmthree_1 \ldots \tmthree_k$ with $k\geq 0$. Normal form simulations are usually based on a big-step presentation of $\towh$: we write $\tm\bswh\tmtwo$ if the $\towh$-evaluation of $\tm$ terminates on $\tmtwo$, and $\tm\bswhdiv$ otherwise.
The following notion of simulation was first considered by \citet{SANGIORGI-normal-form-bisimulation}. Our presentation is slightly different but equivalent.
\begin{definition}[\cbn normal form simulations, \cite{SANGIORGI-normal-form-bisimulation}]
\label{def:cbn-nfs}
	A relation $\relsym$ is a \cbn normal form simulation if $\relsym\subseteq\relcbn$, where $\tm \relcbn \tmp$ holds whenever $\tm,\tmp$ satisfy one of the following clauses:
	\begin{center}
		$\begin{array}{r@{\hspace{.3cm}}r@{\hspace{.3cm}}l@{\hspace{.3cm}}l@{\hspace{.3cm}}lll}
	\textup{(cbn 1)} & && \tm\bswhdiv & \ie ~ \text{has no} \towh \text{-normal form.}
	\\
	\textup{(cbn 2)} & \tm \bswh \var & \text{and} & \tmp \bswh \var &
	\\
	\textup{(cbn 3)} & \tm \bswh \la\var\tm_1 & \text{and} &\tmp \bswh \la\var\tmp_1 & \text{with}~ \tm_1 \rel \tmp_1
	\\
	\textup{(cbn 4)} & \tm \bswh \ntm = \ntmONE\tmtwo&\text{and}& \tmp \bswh \ntmtwo = \ntmONEtwo \tmtwop &\text{with}~\ntmONE\rel\ntmONEtwo ~ \text{and}~ \tmtwo\rel\tmtwop
	\end{array}$
	\end{center}
	\emph{\cbn normal form similarity} $\leqcbn$ is the largest \cbn normal form simulation, that is, it is the union of all normal form simulations. We shorten \emph{normal form} as \emph{nf}.
	\end{definition}
Nf-bisimulations and bisimilarity are the symmetric variants of simulations and similarity, defined as expected.
A simple way of proving soundness of $\leqcbn$ is to show compatibility via the variant of Howe's method developed by \citet{lassen1999bisimulation}, technique that we shall recall in \refsect{naive-compatibility}. 

\paragraph{Partiality and Divergence.} Nf-simulations rely on a partial notion of evaluation (with respect to full $\beta$-reduction), such as weak head reduction $\towh$. The key point is that the partial reduction leaves some sub-terms not evaluated (arguments and abstraction bodies for $\towh$). The derived simulations compare the $\towh$-normal forms $\nfwh(\tm)$ and $\nfwh(\tmtwo)$ of $\tm$ and $\tmtwo$ by: 
\begin{itemize}
\item Checking that they have the same structure for the partially evaluated part of the term, and
\item Asking that the non-evaluated sub-terms of $\nfwh(\tm)$ and $\nfwh(\tmtwo)$ in the same positions are pairwise nf-similar.
\end{itemize}
The use of a partial notion of evaluation is crucial, as it allows fine discriminations related to divergence, which would be blurred if one would only consider full normal forms. \cbn nf-similarity, indeed, discriminates between the following forms of divergence:
\begin{enumerate}
\item Looping as $\Omega$;
\item Looping only after having received an argument, as for $\la\var\Omega$;
\item Having a looping argument, as for $\var \Omega$;
\item The finite iterations of 2 and 3, and of their combination;
\item The \emph{infinite} iteration of 2 or 3, that never actually loop as $\Omega$ as they keep producing an infinite amount of head variables and/or abstractions. There exists terms $\tm$, indeed, the $\towh$ normal form of which is, for instance, $\la\var\tmtwo$ or $\var\tmtwo$ or $\var \la\vartwo\tmtwo$, and such that $\tmtwo$ has the same property. Therefore, they give rise to infinite normal forms such as $\la\var\la\var\la\var\ldots$ or $\var (\var (\var \ldots$ or $\var (\la\vartwo \var (\la\vartwo \ldots$.
\end{enumerate}

%

\paragraph{Open Terms.} Nf-simulations are different from most other notions of program equivalence in that they have to deal with open terms, because, even when the terms to compare are closed, the sub-terms on which the comparison is iterated might be open.

\paragraph{Easier to Use.} With respect to other notions of equivalence such as applicative bisimilarity, nf-bisimilarity is often simpler to establish, because it removes the quantification over arguments. A typical example is the proof of the equivalence of the Curry and Turing fix-point combinators $\curryfixn$ and $\turingfixn$, which is particularly simple with \cbn nf-bisimilarity, as Lassen explained in his first article relating \bohm tree equivalence and head nf-bisimulations \cite{lassen1999bisimulation}, a variant of Sangiorgi's. We can easily adapt his argument on head nf-bisimulations to weak head nf-bisimulations.
Let $\curryfixn = \la\var{\curryfixauxn\curryfixauxn}\text{, where } \curryfixauxn= \la\varthree{\var({\varthree\varthree})}$ and $ \turingfixn = (\la\varthree{\la\var{\var({\varthree\varthree\var})}})(\la\varthree{\la\var{\var({\varthree\varthree\var})}})$. It is easy to check that the following relation $\relsym$ is a \cbn nf-(bi)simulation relating $\curryfixn$ and $\turingfixn$.
\begin{center}
$\rel \defeq \{(\curryfixn,\turingfixn),(\curryfixauxn\curryfixauxn,\var(\turingfixn\var)), (\var(\curryfixauxn\curryfixauxn),\var(\turingfixn\var)),(\curryfixauxn\curryfixauxn,\turingfixn\var) \mid \var\text{ a variable}\}$
\end{center}

\paragraph{Equational Benchmarks} We say that a notion of program equivalence $\simeq_X$ \emph{validates} an equivalence $\equiv_Y$ if $\tm\equiv_Y\tmtwo$ implies $\tm\simeq_X\tmtwo$. Program equivalences in (weak/head/strong) \cbn roughly can differ only along two axes:
\begin{enumerate}
\item \emph{$\eta$-equivalence}: the amount of $\eta$ equivalence that they validate;
\item \emph{$\Omega$-terms}: the amount of identifications among (variants of) $\Omega$-terms.
\end{enumerate}

Sangiorgi's \cbn nf-bisimilarity does not validate $\eta$-equivalence, because his bisimilarity is \emph{rigid}, it cannot equate different normal forms such as $\var$ and $\la\vartwo\var\vartwo$. Since weak \cbn contextual equivalence validates some cases of $\eta$ equivalence, $\eqcbn$ is not fully abstract \cite{SANGIORGI-normal-form-bisimulation}. More generally, nf-bisimilarities tend to not be fully abstract. Intuitively, $\tm$ and $\tmp$ might be externally equivalent, that is, behave the same in all contexts, and yet be internally different, by having different (potentially infinitary) normal forms.

About $\Omega$-terms, it follows from their diverging characterization (\refthm{cbn-scrutability-characterization}) that \cbn nf-bisimilarity equates all \cbn $\Omega$-terms, that is, it validates the following equivalence. 
\begin{itemize}
\item \emph{$\Omega$-equivalence}: $\tm \equivom \tmtwo$ if $\tm$ and $\tmtwo$ are \cbn $\Omega$-terms.
\end{itemize}
Instead, $\eqcbn$ does not equate all unsolvable terms, as it distinguishes $\Omega$ and $\la\var\Omega$ (namely $\Omega \leqcbn \la\var\Omega$ but $\la\var\Omega \not\leqcbn \Omega$). We shall see that in \cbv the reference \cbv nf-bisimulation does not equate all \cbv $\Omega$-terms.


\section{Plotkin's Call-by-Value and Open Terms}
\label{sect:plotkin}
Here we recall Plotkin's \cbv $\l$-calculus, and pay attention to some aspects that are often neglected, as they shall be relevant in the following sections. 

\paragraph{The \cbv $\l$-Calculus.} A \emph{value} $\val$ is a variable or an abstractions. At the rewriting level, we consider only weak evaluation, that is, out of abstractions. We define it in three variants, from left to right, from right to left, and in an unspecified order, which are discussed next.
\begin{center}
\begin{tabular}{c@{\hspace{.7cm}}c}
	$\begin{array}{r@{\hspace{.4cm}}rlll}
	\textsc{Terms} & \tm, \tmtwo, \tmthree & \grameq & \val \mid \tm\tmtwo 
	\\
	\textsc{Values} & \val  & \grameq  & \var \mid \la\var\tm

	\end{array}$
	
	&
	
	$\begin{array}{r@{\hspace{.4cm}}rlll}
	\textsc{(General) Contexts} & \ctx & \grameq &  \ctxhole\mid \tm\ctx\mid \ctx\tm \mid \la\var\ctx
	\\
	\textsc{Weak Contexts} & \wctx & \grameq &  \ctxhole\mid \tm\wctx\mid \wctx\tm
	\\
	\textsc{Left Contexts} & \lctx & \grameq &  \ctxhole\mid \val\lctx\mid \lctx\tm
	\\
	\textsc{Right Contexts} & \rctx & \grameq &  \ctxhole\mid \tm\rctx\mid \rctx\val
	\end{array}$
\end{tabular}
\end{center}

\begin{definition}[Reductions]
Let root $\betav$ reduction $\rtobv$ defined as $(\la\var\tm)\val \rtobv \tm\isub\var\val$. Then 
	$\betav$, weak, left (to right), and right (to left) reduction, noted $\tobv$, $\tow$, $\tolw$, and $\torw$, are defined as follows: if $\tm \rtobv \tmtwo$ then
	\begin{center}
		$\begin{array}{c@{\hspace{1cm}}c@{\hspace{1cm}}c@{\hspace{1cm}}c}
		\ctxp \tm \tobv \ctxp \tmtwo 
		&\wctxp \tm \tow \wctxp \tmtwo 
		&
		\lctxp \tm \tolw \lctxp \tmtwo 
		 &
		\rctxp \tm \torw \rctxp \tmtwo
		\end{array}
		$\end{center}
\end{definition}
It is standard that $\tolw$ and $\torw$ are deterministic and contained in $\tow$, which is non-deterministic but diamond, while $\tobv$ is non-deterministic and confluent but not diamond.

\paragraph{Closed Terms.} If terms are closed and evaluation is weak, which are a standard assumption in the study of functional languages, then call-by-value evaluation has some very nice properties, summed up by the next proposition.

\begin{proposition}
Let $\tm$ be a closed term.
\begin{enumerate}
\item $\tm$ is a weak/left/right normal form if and only if $\tm$ is an abstraction.
\item On $\tm$, left and right reduction are full with respect to weak evaluation, that is, if $\tm \tow \tmtwo$ then there exist $\tmthree$ and $\tmfour$ such that $\tm \tolw \tmthree$ and $\tm \torw \tmfour$.
\end{enumerate}
\end{proposition}
Intuitively, left and right reduction are two equivalent ways of turning weak reduction into a deterministic reduction, on closed terms. 
In \cbv, contextual preorder $\leqcv$ is defined with respect to evaluation to a value, which can be equivalently expressed as termination of weak, left, or right reduction. Importantly, for $\leqcv$ one considers only the reduction of closed terms, as the definition of $\leqcv$ is based on contexts closing the terms to compare.

\paragraph{Open Terms, Stuck Redexes, and The Inequivalence of the Three Strategies.} As soon as one considers open terms, the good properties of the system break. The only ones which survive are the diamond property of weak reduction and the fact that left and right reduction are deterministic. But weak normal forms now have a complex shape, as there can be \emph{stuck redexes} such as $(\la\var\tm) (\vartwo\vartwo)$ which cannot be reduced because their argument is normal and not a value.

Stuck redexes unfortunately break the equivalence of left, right, and weak reduction: left, right, and weak normal forms are all different notions, because left and right reduction are no longer full with respect to weak reduction. For instance, $\tm \defeq \var\var(\Id\Id)$ is a left normal form which is not a weak/right normal form, because $\tm \tow \var\var\Id$, which is weak/right normal. Similarly, $\Id\Id(\var\var)$ is a right normal form which is not left/weak normal, and $\var\var(\Id\Id)(\var\var)$ is a left/right normal form which is not weak normal. Perhaps more worrying is the fact that stuck redexes introduce suspicious distinctions between contextually equivalent terms: $\Omega$ and $\Omega^L \defeq (\la\var\delta)(\vartwo\varthree)\delta$ are contextually equivalent, but the first one diverges while the second one is normal, because of the stuck redex.
\section{Equational Benchmarks for \cbv Program Equivalences}
\label{sect:benchmarks}

In \cbn, there are only two 'equational benchmarks', or degrees of freedom, for program equivalences, namely $\eta$-equivalence and the $\Omega$-equivalence $\equivom$. In \cbv, the situation is richer, there are various equivalences that can be validated or not. Here we list the most relevant ones. We start by discussing the \cbv variants of $\eta$-equivalence and $\Omega$-equivalence, and then present the equivalences that are found in known extensions of Plotkin's calculus.

For all the equivalences, we simply give the root axioms defining them, assuming that they are closed by all contexts. All the given equivalences (but \cbn erasure) are validated by \cbv contextual equivalence $\eqcv$. 
The meaning of some of the equivalences from extended \cbv calculi might seem obscure. They shall make more sense after the introduction of explicit substitutions in \refsect{vsc}.

In \cbv, $\eta$-equivalence has to be restricted, otherwise it turns non-values into values. At first sight, the by value version of $\eta$ seems to be $\val \equivetav \la\var{\val\var}$ if $\var\notin\fv\val$. But since any \cbv program equivalence validates $\betav$-reduction, the case $\val=\la\vartwo\tm$ is actually caught by $\betav$-reduction  (because $\la\var (\la\vartwo\tm)\var \tobv \la\var\tm\isub\vartwo\var =_\alpha \la\vartwo\tm$), so that $\equivetav$ simply amounts to the variable case.
\begin{itemize}
\item \emph{$\etav$ equivalence}: $\vartwo \equivetav \la\var{\vartwo\var}$ for every variable $\vartwo$.
\end{itemize}

\paragraph{\cbv $\Omega$-terms} $\Omega$-terms adapt to \cbv by simply considering \cbv contextual equivalence.
The equivalence to be validated here is the following one:
\begin{itemize}
\item \emph{$\Omega$-equivalence by value}: $\tm \equivomv \tmtwo$ if $\tm$ and $\tmtwo$ are \cbv $\Omega$-terms.
\end{itemize}
In Plotkin's calculus, \cbv $\Omega$-terms cannot have a diverging characterization akin to that of \cbn $\Omega$-terms (\refthm{cbn-scrutability-characterization})---they shall have one in the VSC. For instance, an $\Omega$-term such as $\Omega^L \defeq (\la\var\delta)(\vartwo\varthree)\delta$ is normal for Plotkin, while it should diverge if a good characterization existed. 

The equivalence $\equivomv$ has a special role among those listed here because whether a term is an $\Omega$-term is \emph{undecidable}, so that the equivalence cannot be seen as computational principle to be tested via a rewriting rule. It is then all the more relevant that a program equivalence validates it.

\paragraph{Moggi} We now turn to equivalences found in extensions of Plotkin's calculus. The equivalences enriching $\betav$-conversion in Moggi's untyped computational $\l$-calculus are the following ones, here reformulated without $\letexp$-expressions:
\begin{itemize}
	\item \emph{Left identity}: $\Id \tm \equiv_{lid} \tm$, where $\Id=\la\var\var$ is the identity combinator;
	\item \emph{Associativity of lets}: $(\la\var\tm)((\la\vartwo\tmtwo)\tmthree) \equivass (\la\vartwo(\la\var\tm)\tmtwo)\tmthree)$ if $\vartwo \not \in \fv{\adr\tm}$;
	\item \emph{Left decomposition of applications}: $\tm\tmtwo \equivlad (\la\var\var\tmtwo)\tm$ if $\var\not\in\fv\tmtwo$;
	\item \emph{Right decomposition of applications}: $\val\tm \equivrad (\la\var\val\var)\tm$ if $\var\not\in\fv\val$. This one exists also in an extended form: $\tmtwo\tm \equivexrad (\la\var\tmtwo\var)\tm$ if $\var\not\in\fv{\adr{\tmtwo}}$
\end{itemize}
Of them, the most interesting one for our study is $\equivlid$, which in Plotkin's calculus holds only for values, as for values it is an instance of $\betav$-conversion, while in Moggi's it holds for every term $\tm$.

\paragraph{Proof Nets} The \cbv translation of $\l$-calculus in linear logic proof nets, studied in detail by \citet{Accattoli-proofnets}, equates various pairs of terms. The induced equivalences are better expressed with explicit subsitution, as we shall see in \refsect{benchmarks-vsc}, but we anticipate them here anyway. They subsume Moggi's $\equivass$ rule and other presentations of proof nets equivalences such as the $\sigma$-rules of the shuffling calculus of \citet{shufflingcalculus}. Moreover, they include the following equivalences.
\begin{itemize}
	\item Left (Applied-$\l$) Application: $(\la\var\tm)\tmtwo\tmthree \equivsone (\la\var\tm\tmthree)\tmtwo$ if $\var\not\in\fv\tmthree$;
	
	\item Right (Applied-$\l$) Application: $\tm((\la\vartwo\tmtwo)\tmthree) \equivexsthree (\la\vartwo\tm\tmtwo)\tmthree$ if  $\vartwo\not\in\fv{\tm}$;
	\item Commutativity: $(\la\vartwo(\la\var\tm)\tmtwo)\tmthree \equivcom (\la\var(\la\vartwo\tm)\tmthree)\tmtwo$ if $\var\notin\fv\tmthree$ and $\vartwo\notin\fv\tmtwo$.
\end{itemize}
The first two rules correspond to possible commutations between applications and applied lambdas--which shall correspond to commutations between applications and $\letexp$s in \refsect{benchmarks-vsc}. The second one also can be seen as a generalization of Moggi's $\equiv_{ass}$ replacing the value $\la\var\tm$ with whatever term $\tm$. Commutativity swaps adjacent and unrelated redexes, and it is the equivalence that in the introduction is formulated with $\letexp$s. It is a special equivalence, for at least two reasons. 
\begin{enumerate}
\item \emph{Effects}: commutativity holds in the pure \cbv setting but it often fails in extensions of \cbv with effects, because many effects are order-dependent (think of the order of writes on a memory cell). Therefore, it is an equivalence that one might want to be able to modularly add or remove from a notion of bisimilarity, rather than always validate it. 
\item \emph{Unorientable}: being symmetric, commutativity cannot be oriented as a rewriting rule. Therefore, any nf-bisimulation validating it needs to be able to compare normal forms up to some deformation of terms.
\end{enumerate}
The reason why some equivalences at times appear in restricted forms is also related to effects. With non-commutative effects one has to fix a deterministic evaluation strategy, typically left-to-right, and this constrains the shape of equivalences forcing a sub-term to be a value $\val$ (resp. $\la\var\tm$), as in $\equivrad$ (resp. $\equiv_{ass}$). Similarly for proof nets equivalence, evaluating left-to-right would lead us to a restricted version of $\equivexsthree$ where the sub-term $\tm$ has to be a value. Proceeding right-to-left would relax them but force other dual constraints (such as $\tm$ being a value in $\equivlad$), and adopting a non-specified order induces the extended versions of Moggi's laws and the unrestricted version of the proof nets equivalences.

\paragraph{\cbn Duplication and \cbn Erasure} The last equivalences that we consider are \cbn duplication and \cbn erasure. We do not actually know how to characterize \cbn duplication independently of erasure with an axiom, or a set of axioms, but we discuss a specific case, to illustrate the idea.
\begin{itemize}
\item \emph{\cbn Duplication}: $(\la\var \vartwo \var \var) \tmtwo \equiv_{dup} \vartwo \tmtwo\tmtwo$;
\item \emph{\cbn Erasure}: $(\la\var \tm ) \tmtwo \equiv_{era} \tm$ with $\var\notin\fv\tm$.
\end{itemize}
In \cbv, duplication and erasure are included in $\betav$-conversion if $\tmtwo$ is a value, but not otherwise. For arbitrary terms, erasure is unsound in \cbv, because erasing a sub-term might turn divergence into termination: for instance $\tm\defeq (\la\var\varthree)(\vartwo\vartwo)$ is not contextually equivalent to $\tmtwo\defeq \varthree$ in \cbv, because $\ctxp\tm$ diverges while $\ctxp\tmtwo$ terminates with respect to the context $\ctx\defeq (\la\vartwo\ctxhole)\delta$. 

Duplication, instead, is sound for arbitrary terms, the idea being that terminating (resp. diverging) once, or terminating (resp. diverging) twice does not affect termination (resp. divergence). It is however a principle somewhat clashing with the nature of \cbv. The cornerstone of \cbv is the idea that one should reduce arguments \emph{before} substituting them, and \cbn duplication does exactly the opposite. Intuitively, a \cbv program equivalence validating \cbn duplication is \emph{qualitative} or \emph{cost-insensitive}, as it only observes termination, while one rejecting it is somehow \emph{cost-sensitive}, as it distinguishes between differently efficient ways of realizing the same qualitative behavior.

\paragraph{Summing Up} Of all the discussed equivalences, the most relevant ones for our study are the $\Omega$-equivalence $\equivomv$ and left identity $\equivlid$. Commutativity $\equivcom$ shall also be a source of inspiration for the modular \emph{mirror} approach of \refsect{net}. \cbn Duplication is validated by \cbv contextual equivalence while none of the program equivalences studied in this paper validates it. This fact shows that they may be sound but are \emph{not complete} with respect to contextual equivalence.
\section{Naive \cbv Bisimilarity}
\label{sect:naive}
If one takes Sangiorgi's \cbn nf-similarity $\leqcbn$ (\refdef{cbn-nfs}) and simply replaces weak head reduction with one of the weak \cbv reductions (weak, left, or right) then one obtains notions of \cbv nf-similarity. 

We define the similarity induced by weak reduction. Left or right reductions induce different similarities but with the same pros and cons that are discussed below. Let $\bsweak $ and $\bsweakdiv$ be big-step termination and divergence with respect to weak \cbv reduction $\tow$.
\begin{definition}[Naive nf-bisimulation for Call-by-Value]
A relation $\relsym$ is a \emph{naive (\cbv) nf-simulation} if $\relsym\subseteq\relncbv$, where $\tm \relncbv \tmp$ holds whenever $\tm,\tmp$ satisfy one of the following clauses:
	\begin{center}
	$\begin{array}{r@{\hspace{.3cm}}r@{\hspace{.3cm}}l@{\hspace{.3cm}}l@{\hspace{.3cm}}lll}
	\textup{(nai 1)} & && \tm\bsweakdiv & \ie ~ \text{has no} \tow \text{-normal form.}
	\\
	\textup{(nai 2)} & \tm \bsweak \var & \text{and} & \tmp \bsweak \var &
	\\
	\textup{(nai 3)} & \tm \bsweak \la\var\tm_1 & \text{and} &\tmp \bsweak \la\var\tmp_1 & \text{with}~ \tm_1 \rel \tmp_1
	\\
	\textup{(nai 4)} & \tm \bsweak \ntm = \ntmONE\ntmTWO&\text{and}& \tmp \bsweak \ntmtwo = \ntmONEtwo\ntmTWOtwo &\text{with}~\ntmONE\rel\ntmONEtwo ~ \text{and}~ \ntmTWO\rel\ntmTWOtwo
	\end{array}$
\end{center}
	\emph{Naive nf-similarity} $\leqncbv$ is the largest naive nf-simulation.
	\end{definition}

Naive (bi)similarity seems defined very naturally, and yet it does not validate \emph{any} of the equivalences of the previous section. Let's discuss $\equivomv$ and $\equivlid$.
\begin{enumerate}
\item
 \emph{$\Omega$-equivalence}: $\equivomv$ is not validated by $\leqncbv$, and $\Omega$-terms are not minimal for $\leqncbv$, in contrast with the fact that \cbn $\Omega$-terms are minimal for Sangiorgi's similarity $\leqcbn$. For instance, $\Omega$ is $\leqncbv$-minimal, but $\Omega^L$ is an $\Omega$-term and one has $\Omega \leqncbv \Omega^L$, because $\Omega\bsweakdiv$, but not $\Omega^L \not\leqncbv \Omega$, because $\Omega^L$ is $\tow$-normal.

\item \emph{Left identity}: the equivalence $\equivlid$ is not validated by $\leqncbv$. In Plotkin's calculus $\Id (\var\tm)$ does not reduce to $\var\tm$, because $\var\tm$ is not a value---more generally this happens for all normal open terms that are not values. Therefore, $\Id (\var\tm)$ and $\var \tm$ are $\leqncbv$-incomparable.
\end{enumerate}

\adr{\paragraph{$\beta_v$-Conversion}
	From the definition of $\leqncbv$, it immediately follows that naive bisimilarity contains the $\rtobv$ root rule, that is, that if $\tm \rtobv \tmtwo$ then $\tm\eqncbv \tmtwo$, simply because $\tm$ and $\tmtwo$ have the same left normal form. Since $\eqncbv$ is a compatible equivalence relation, it turns out that $\eqncbv$ contains the whole of $\betav$-conversion $=_{\betav}$\cadr{, thus it contains left reduction as well as weak and right reductions.}{.}
	
	\begin{proposition}[$\betav$-conversion is validated by naive bisimilarity]
		If $\tm =_{\betav} \tmtwo$ then $\tm \eqncbv \tmtwo$.
\end{proposition}}

\paragraph{Naive Similarity and Fix-Points} \cadr{It is nonetheless }{Despite its naivety, it is }possible to prove that the usual \cbv variants of Curry's and Turing's fix-point combinators $\curryfix$ and $\turingfix$ are naively similar, as we now show.
\begin{center}
$\begin{array}{c\colspace |\colspace  c}
\textsc{Curry's fix-point} & \textsc{Turing's fix-point}
\\
 \curryfix \defeq \la\var{\curryfixaux\curryfixaux}\text{, where } \curryfixaux= \la\varthree{\var\la\vartwo{\varthree\varthree\vartwo}}
 &
 \turingfix \defeq (\la\varthree{\la\var{\var\la\vartwo{\varthree\varthree\var\vartwo}}})(\la\varthree{\la\var{\var\la\vartwo{\varthree\varthree\var\vartwo}}}) 
 \end{array}$
 \end{center}
Let's build a naive (bi)simulation $\relsym$ relating $\curryfix$ and $\turingfix$. The relation $\relsym$ must contain the pair $(\curryfix,\turingfix)$.  Both terms $\tow$-evaluate to an abstraction ($\turingfix \to \la\var{\var\la\vartwo{\turingfix\var\vartwo}}$). Hence their weak normal forms are (abstractions) $\la\var{\curryfixaux\curryfixaux}$ and $\la\var{\var\la\vartwo{\turingfix\var\vartwo}}$ which must satisfy the third clause for $\leqncbv$, that is, their bodies under the abstraction $(\curryfixaux\curryfixaux,\var\la\vartwo{\turingfix\var\vartwo})$ must appear in $\relsym$. Since $\curryfixaux\curryfixaux \tow \var \la\vartwo{\curryfixaux\curryfixaux\vartwo}$, $\relsym$ must contain the fourth clause requirements, that is, $\var \rel \var$ and $\la\vartwo{\curryfixaux\curryfixaux\vartwo} \rel \la\vartwo{\turingfix\var\vartwo}$. As a result, $\var \rel \var$ needs to be added to the simulation (for all possible choice of variables when opening the abstraction at the first step), and we continue on adding $(\la\vartwo{\curryfixaux\curryfixaux\vartwo}, \la\vartwo{\turingfix\var\vartwo})$ to $\relsym$, then adding any pair of terms needed so that $\relsym\subseteq\relncbv$. By repeating this process, we eventually fall back to $\var\la\vartwo{\curryfixaux\curryfixaux\vartwo} \rel \var\la\vartwo{\turingfix\var\vartwo}$ which is already in the built relation $\relsym$, which means it now satisfies $\relsym \subseteq \relncbv$, that is, $\relsym$ is a naive simulation. Since similarly the symmetric relation $sym(\relsym)$ also satisfies $sym(\relsym) \subseteq \ncbvfp{sym(\relsym)}$, $\relsym$ is actually a naive \emph{bi}simulation.

The full relation proving the next proposition is:
\begin{center}$
\relsym \defeq \{(\var,\var) \mid \var \text{ a variable}\} ~\cup~ 
\{~~(\curryfix,\turingfix), (\la\var\curryfixaux\curryfixaux,\la\var{\var\la\vartwo{\turingfix\var\vartwo}}),\} ~~\cup $
\end{center} \begin{center}$
\{ (\curryfixaux\curryfixaux,{\var\la\vartwo{\turingfix\var\vartwo}}), ~~ (\var\la\vartwo{\curryfixaux\curryfixaux\vartwo},{\var\la\vartwo{\turingfix\var\vartwo}}),~~ (\la\vartwo{\curryfixaux\curryfixaux\vartwo},\la\vartwo{\turingfix\var\vartwo})\mid \var\text{ a variable}\}~~\cup
$\end{center}\begin{center}$
 \{({\curryfixaux\curryfixaux\vartwo},{\turingfix\var\vartwo}), ((\var\la\vartwop{\curryfixaux\curryfixaux\vartwop})\vartwo,({\var\la\vartwop{\turingfix\var\vartwop}})\vartwo)\mid \var\text{ and }\vartwo \text{ variables}\}$
\end{center} 

\begin{proposition}
\label{prop:naive-fix-points-equiv}
$\curryfix \eqncbv \turingfix$, that is, Curry's and Turing's fix-points are naive bisimilar.
\end{proposition}
%

\section{Compatibility and Lassen's method for naive bisimilarity}
\label{sect:naive-compatibility}
In this section, we prove compatibility for the naive bisimilarity introduced in the previous section. The aim is to provide a gentle introduction to \citeauthor{lassen1999bisimulation}'s variant \citeyearpar{lassen1999bisimulation} of Howe's method \cite{Howe1996method,DBLP:books/cu/12/Pitts12} for proving the compatibility of similarities, and to delay some of the technicalities that we shall need to address for the similarity we are really interested in.

We prove compatibility for weak naive similarity, but the proof technique easily adapts to left and right naive similarity.

\paragraph{From Small-Step to Big-Step.} Nf-similarities look at normal forms, and the crucial proof in Howe's method proceeds by induction on a big-step formulation of evaluation, where \emph{big-step} means that it relates a terminating term directly with its normal form, hiding the intermediate steps. Therefore, we need to reformulate the small-step reduction $\tm\tow^k\ntm$ where $\ntm$ is $\tow$-normal, in a big-step manner as $\tm \bsws \ntm$. 

For the technical development, we need to keep the information about the number $k$ of small steps, that is, we shall rather write $\tm \bsw k \ntm$. Such a quantitative information is needed both to prove the equivalence with small-step evaluation and for the crucial proof in the method. 

\begin{definition}[Big-step weak evaluation $\bsws$]
The big-step weak evaluation predicate $\tm \bsw k \ntm$, read \emph{$\tm$ (weak)-converges in $k$ steps to  a normal form $\ntm$}, is defined as follows.
\begin{center}
\begin{tabular}{cccc}
	\begin{tabular}{c\colspace cc}
	\infer[(\bswax)]{\val \bsw 0 \val}{}
	& 
	\infer[(\bswbeta)]{\tm\tmtwo \bsw {k+h+i+1} \ntm}{
		\tm \bsw k \la\var\tmthree
		&
		\tmtwo \bsw h \val
		&
		{\tmthree\isub\var\val} \bsw i \ntm}
&\infer[(\bswappnf)]{\tm\tmtwo \bsw {k+h} \ntmONE=\ntm\ntmtwo}{
	\tm \bsw k \ntm
	&
	\tmtwo \bsw h \ntmtwo}
	\end{tabular}
	\\[12pt]


\end{tabular}
\end{center}
\end{definition}

Notation: $\tm \bsws \ntm$ abbreviates \emph{there exists a $k$ such that $\tm \bsw k \ntm$}.
\begin{toappendix}
\begin{proposition}[Equivalence of small-step and big-step in weak]
	\label{l:ss-bs-equivalence_weak}
	$\tm \bsw k \ntm$ if and only if $\tm \tow^k \ntm$ with $\ntm$ normal.
\end{proposition}
\end{toappendix}
%
%
Let us stress an important point. We recall that $\tow$ is non-deterministic but diamond. The diamond property is here crucial, in order to make sense---at the big-step level---of the number of steps $k$, which for a diamond reduction does not depend on the reduction path to normal form.

\ignore{DELETE, not needed
\begin{toappendix}
	\begin{lemma}
		\label{l:grammar-nf-weak-cbv}
	$\tm$ is $\tow$-normal if and only if there exists $\ntm$, given by the following grammar, $\tm = \ntm$.
	\begin{center}
		$\begin{array}{r@{\hspace{.5cm}}rlll}
		\textsc{Inert terms}  &
		\itm,\itmtwo & \grameq &  \var\val \mid \val\itm \mid \itm\ntm
		\\
		{\tow}\textsc{-normal forms}  &
		\ntm,\ntmtwo & \grameq &  \val \mid \itm
		\end{array}
		$\end{center}
\end{lemma}
\end{toappendix}

	\paragraph{Making Inert Terms Explicit in the Clauses.} Case (nai 4) in the definition of naive simulations can be refined according to the grammar of weak normal forms: for normal forms of the shape $\ntmONE\ntmTWO$, there are three different possibilities, $\var\val$, $\val\itm$ or $\itm\ntm$.  The clause can then be equivalently replaced by the three different following clauses.
	\begin{center}
		$\begin{array}{r@{\hspace{.3cm}}r@{\hspace{.3cm}}l@{\hspace{.3cm}}l@{\hspace{.3cm}}lll}
		\text{(nai 4a)} & \tm \bsws \var\val &\textit{and}& \tmp \bsws \var\valtwo
		& \textit{with} ~ \val \rel \valtwo
		\\
		\text{(nai 4b)} & \tm \bsws \val\itm &\textit{and}& \tmp \bsws \valtwo\itmtwo
		& \textit{with} ~ \val \rel \valtwo ~\textit{and}~ \itm \rel \itmtwo
		\\
			\text{(nai 4c)} & \tm \bsws \itm \ntm &\textit{and}& \tmp \bsws \itmtwo \ntmtwo 
		& \textit{with} ~ \itm \rel \itmtwo ~\textit{and}~ \ntm \rel \ntmtwo
		\end{array}
		$\end{center}}
	
\paragraph{(Howe-)Lassen's Method} Proving that a behavioral preorder $\precsim$ is compatible often cannot be done directly, that is, just by induction on the contextual closure. The idea of Howe's method is that, instead of proving compatibility of $\precsim$, one introduces a derived preorder $\howeop\precsim$ where the compatible closure is enforced in the definition, and then proves that $\precsim$ and $\howeop\precsim$ coincide. \citeauthor{Howe1996method} introduced his method \citeyearpar{Howe1996method} to deal with \emph{applicative} similarities, \citeauthor{lassen1999bisimulation} adapted it \citeyearpar{lassen1999bisimulation} for \emph{normal form} similarities. The general idea is the same, but Lassen considers a different closure operation $\lassenop\precsim$.

\paragraph{Lassen's Closure.} The difficulty in proving directly that a similarity $\precsim$ is compatible comes from the applicative contextual closure, which may introduce a $\beta$-redex (when applying an abstraction to a term), that in turn can substitute over $\precsim$-related terms. The idea is to define the preorder $\lassenop\precsim$ as the compatible, substitutive, and reflexive closure of $\precsim$. 

\begin{definition}[Lassen closure]
The \emph{Lassen closure} $\lasrelsym$ of a relation $\relsym$ on terms is given by:
	\begin{center}
\begin{tabular}{cccccc} 
\begin{tabular}{c\colspace\colspace c\colspace\colspace cccc} 
	\infer[\sclift ]{\tmrone \lasrel \tmrtwo} {\tmrone \rel \tmrtwo}
	&
	\infer[\scvar]{\var \lasrel \var}	{}
	&
	\infer[\scabs ]{\la\var\tmrone \lasrel \la\var\tmrtwo} {\tmrone \lasrel \tmrtwo}
	
\end{tabular}
\\[10pt]
\begin{tabular}{c\colspace\colspace ccccc}
	
	\infer[\scapp ] {\tmrone\tmrthree  \lasrel  \tmrtwo\tmrfour} {\tmrone  \lasrel \tmrtwo & \tmrthree \lasrel \tmrfour }  
&
	\infer[\scsub ]{\tmrone\isub\var{\val} \lasrel \tmrtwo\isub\var{\valtwo}{}} {\tmrone \lasrel \tmrtwo & \val \lasrel \valtwo }	
\end{tabular}
\end{tabular}		
	\end{center}
\end{definition}
Note rule ($\scsub$): only values can be substituted, as this is Call-by-Value's mantra.

\paragraph{Lassen's Closure Preserves Simulations} The proof of equivalence of $\leqncbv$ and $\lassenop\leqncbv$ reduces to proving that the closure operator $\lassenop\cdot$ preserves $\leqncbv$ simulations, that is, that $\lassenop\relsym$ is a naive simulation if $\relsym$ is---it is often referred to as the \emph{main lemma} of the method. The proof is delicate and rests on two key intermediate properties. The first one concerns the evaluation level, and, when expressed at the big-step level, it is a sort of factorization property with respect to meta-level substitutions. In fact, it is nothing else but the substitutivity of small-step evaluation, rephrased at the big-step level.

\begin{proposition}[Small-step substitutivity]
	\label{l:stability_weak}
	If $\tm\tow\tmp$ then $\tm\isubst\val\var \tow \tmp\isubst\val\var$
\end{proposition}

\begin{proof}
	By induction on $\tm\tow\tmp$ (induction on contexts).
\end{proof}

\begin{lemma}[Big-step substitutivity]
	\label{l:splitting_weak}
	If $\tm\isubst\val\var \bsw k \ntm$ then there exist $k'$ and $\ntmtwo$ such that $ \tm \bsw {k'} \ntmtwo$ and $\ntmtwo\isubst\val\var\bsw {k-k'} \ntm$.
\end{lemma}

\begin{proof}
If  $\tm\isubst\val\var \bsw k \ntm$, then $ \tm \bsws \ntmtwo$ because if $\tm$ diverges then $\tm\isubst\val\var$ diverges as well by substitutivity of $\tow$ {(\reflemma{stability_weak})}.
	Then there exists $k'$ such that $\tm \bsw {k'} \ntmtwo$. Note that by substitutivity we have $ \tm\isubst\val\var \tow^{k'} \ntmtwo\isubst\val\var$, and so $\ntmtwo\isubst\val\var\bsw {k-k'} \ntm$ because the reduction is diamond, hence all normalizing reduction sequences have the same length.
\end{proof}

The second key intermediate property is the coherence of naive simulations with respect to reduction and substitution.
\begin{toappendix}
\begin{proposition}[{Coherence of simulation, reduction, and substitution}]
\label{prop:ncbv-coherence}
Let $\rel$ be a naive simulation, $\ntm \lasrelncbv \ntmtwo$, and $\val\lasrelncbv\valtwo$.
\begin{enumerate}
\item \emph{Normal forms}: if $\ntm\isub\var\val$ is $\tow$-normal then $\ntmtwo\isub\var\valtwo$ is $\tow$-normal and\\ $\ntm\isub\var\val \lasrelncbv \ntmtwo\isub\var\valtwo$.
\item \emph{Steps}: if $\ntm\isub\var{\val} \tow \tm$
	then $\ntmtwo\isub\var{\valtwo}  \tow \tmtwo$ and $\tm \lasrel \tmtwo$.
\end{enumerate}
\end{proposition}
\end{toappendix}
{Note that the second point has $\lasrel$ rather than $\lasrelncbv$ in the conclusion. This is because in general $\tm$ and $\tmtwo$ are not normal. In the proof of the next proposition, it is shown that the normal forms of $\tm$ and $\tmtwo$ are in fact $\lasrelncbv$-related.}

We can now prove the crucial property of Lassen's closure.
\begin{toappendix}
\begin{proposition}
	\label{prop:main-lemma_naive}
		Let $\relsym$ be a naive simulation.
		\begin{enumerate}
		\item \emph{Technical auxiliary statement}: if $\tmrone\lasrel\tmrtwo$ and $\tmrone \bsw k \ntm$ then $\tmrtwo\bsws \ntmtwo$ and $\ntm \lasrelncbv \ntmtwo$.		
		\item \emph{Lassen's closure preserves naive simulations}:  $\lassenop\relsym$ is a naive simulation.
		\end{enumerate}
\end{proposition}
\end{toappendix}
\begin{proof} 
	\hfill
	\begin{enumerate}
	\item \emph{Sketch} (complete proof in Appendix \ref{chapter:proof-compatibility-naive} of the additional material on HotCRP):	by induction on $(k,d)$ where $d$ is the size of the derivation of $\tmrone \lasrel \tmrtwo$. We proceed by case analysis on the last rule of the derivation $\tmrone \lasrel \tmrtwo$. Cases ($\sclift$), ($\scvar$), and ($\scabs$) are immediate by definition. Case ($\scapp$) relies on a second case analysis (on the last rule of the $\tmrone \bsw k \ntm$ derivation). The sub-cases are routine and may depend on the ($\scsub$) rule. Case ($\scsub$) is the core of the proof. It starts by applying big-step substitutivity (\reflemma{splitting_weak}) to $\tmrone = \tm\isub\var\val$ and then, depending on whether the obtained $\ntmtwo\isubst\val\var$ is normal, it applies the corresponding coherence property of naive simulations with respect to evaluation (\refprop{ncbv-coherence}).
	\item Unfolding the statement one obtains exactly the statement of point 1.\qedhere
	\end{enumerate}
\end{proof}
Finally, we can use the preservation property to prove the redundancy of the closure, from which the compatibility and the soundness of naive similarity follows.
\begin{theorem}[Compatibility and soundness of $\leqncbv$]
	\hfill
	\begin{enumerate}
	\item \emph{Redundancy of Lassen's closure}: $\leqncbv \,= \lassenop \leqncbv$.
	\item Naive similarity is compatible and included in the \cbv contextual preorder $\leqcv$.
	\end{enumerate}
\end{theorem}

\begin{proof}
\hfill
\begin{enumerate}
\item By construction of $\lassenop\ctxhole$, $\leqncbv \subseteq \lassenop\leqncbv$ (by rule $\sclift$). Preservation of naive simulations by Lassen's closure (\refprop{main-lemma_naive}) and the fact that $\leqncbv$ is a naive simulation give that $\lassenop\leqncbv$ is a naive simulation. By definition, $\leqncbv$ is the maximal naive simulation hence $\lassenop\leqncbv \subseteq \leqncbv$. 
\item Compatibility follows from point 1, because $\lassenop\leqncbv$ is compatible by definition. Inclusion in $\eqcv$ follows by \refprop{congruence-included-contextual-equivalence} and by adequacy of $\leqncbv$, which is trivial.\qedhere
\end{enumerate}
\end{proof}

%

\section{Lassen's Eager Normal Form Simulation}
\label{sect:enf}
The \cbv nf-simulation of reference in the literature is due to Lassen \cite{LassenEnf}.  Lassen's simulation is interesting because it is not defined by simply changing the notion of reduction in Sangiorgi's. Lassen indeed exploits stuck redexes in open terms, and defines a simulation that \emph{unstucks} them, a mechanism which we shall refer to as \emph{stop-and-go}. 

\paragraph{Grammar of Left Normal Forms} Lassen's similarity is built using \emph{left reduction} $\tolw$. The starting point is the observation that, despite the limits of left reduction, it admits a description of its normal forms via left contexts which is extremely simple and elegant.
\begin{lemma}[Unique decomposition, \cite{LassenEnf}]
\label{l:las-unique-dec}
	Any (possibly open) term is either a value or admits a unique decomposition $\levctxp{\val\valtwo}$. In particular, left normal forms can be described as follows:
	\begin{center}
	$\begin{array}{c@{\hspace{.5cm}}rcc}
	\textsc{Left normal forms}  &
	\ntm,\ntmtwo & \grameq &  \val\mid \levctxp{\var\val}
	\end{array}
	$\end{center}

\end{lemma}

The simple structure of left normal forms is then used by Lassen to define eager nf-simulation. The crucial clause is the fourth one, which realizes the stop-and-go mechanism. \emph{Notation}: let $\bswlefts $ and $\bswleftdiv$ be big-step termination and divergence with respect to left \cbv reduction $\tolw$.

\begin{definition}[\Enf simulation, \cite{LassenEnf}]
	A relation $\relsym$ is an \emph{eager normal form (\enf) simulation} if $\relsym\subseteq\relenf$, where $\tm \relenf\tmp$ holds whenever $\tm,\tmp$ satisfy one of the following clauses:
	\begin{center}
		$\begin{array}{r@{\hspace{.3cm}}r@{\hspace{.3cm}}l@{\hspace{.3cm}}l@{\hspace{.3cm}}lll}
		\textup{(enf 1)} & &&\tm\bswleftdiv & \ie ~ \text{has no} \tolw \text{-normal form.}
		\\
		\textup{(enf 2)} & \tm \bswlefts \var  &\text{and}& \tmp \bswlefts \var
		\\
		\textup{(enf 3)} & \tm \bswlefts \la\var\tmfirst &\text{and}& \tmp \bswlefts \la\var\tmpfirst 
		& \text{with} ~ \tmfirst \rel \tmpfirst
		\\
		\textup{(enf 4)} & \tm \bswlefts \levctxp{\var\val}  &\text{and}& \tmp \bswlefts \levctxtwop{\var\valtwo} 
		&
		\textup{with} ~ \val\rel\valtwo ~ \text{and}  ~ \levctxp{\varthree}\rel\levctxtwop{\varthree} 
		\\
		&& &&\text{where} ~ \varthree ~ \text{is not free in} ~ \levctx ~ \text{or} ~ \levctxtwo
		\end{array}
		$\end{center}
	\emph{\Enf similarity}, written $ \leqenf $, is defined by co-induction as the largest \enf simulation, that is, it is the union of all \enf simulations. We say that $\tm$ is \emph{\enf similar} to $\tmp$ if $\tm \leqenf \tmp$.
\end{definition}

\paragraph{Stop-and-Go, Double Task, and Left Identity.} The stop-and-go clause (enf 4) cleverly does two tasks at the same time, subsuming clause (nai 4) of naive similarity in a bottom-up way and capturing the left identity equivalence $\equivlid$. 

About (nai 4), consider for instance comparing $\tm \defeq \var\vartwo\varfour$ with itself via enf simulations. Clause (enf 4) reduces it to compare $\varthree\varfour$ and $\vartwo$ with themselves, since $\tm = \lctxp{\var\vartwo}$ with $\lctx\defeq \ctxhole\varfour$. Then, it reduces the first one to compare $\varthree'$ and $\varfour$ with themselves. In contrast, (nai 4) (or Sangiorgi's (cbn 4)) proceeds \emph{top-down}, by splitting $\var\vartwo\varfour$ into $\var\vartwo$ and $\varfour$, and then splitting $\var\vartwo$.

About the left identity equivalence, consider showing that $\tm \defeq \Id(\var\var)$ and $\tmtwo\defeq \var\var$ are enf similar. Evaluation is stuck on $\tm$, that is, it \emph{stops} because $\var\var$ is not a value. Note that $\tm$ has shape $\lctxp{\var\val}$ with $\lctx=\Id\ctxhole$ and $\val=\var$. The idea is that a term $\tmtwo$ that is $\leqenf$-similar to $\tm$ has to get stuck as well, or anyway decompose in a similar way. Now, $\tmtwo$ is not stuck, because there are no blocked redexes, but it has nonetheless shape $\lctxtwop{\var\valtwo}$ by taking $\lctxtwo = \ctxhole$ and $\valtwo = \var$. The comparison between $\tm$ and $\tmtwo$ is then reduced to compare the two pairs $\lctxp\varthree = \Id \varthree$ and $\lctxtwop\varthree=\varthree$, and $\val=\var$ and $\valtwo=\var$. The second pair trivially matches, because the identity relation is an enf simulation. About the first pair, note that $\Id\varthree$ is no longer stuck, that is, it can \emph{go}. And for the next round of comparison (of $\Id\varthree$ and $\varthree$) we have to first $\tolw$ reduce the terms, so that $\Id\varthree \tolw \varthree$ and thus also the first pair trivially matches.

Summing up, the enf simulation relating $\Id(\var\var)$ and $\var\var$ is the following  simulation $\relsym$: \[\relsym =\{(\Id(\var\var),\var\var), (\var,\var), (\Id\varthree,\varthree), (\varthree,\varthree)\}\]
This is  just an instance, but the following more general result holds.
\begin{proposition}[Enf bisimilarity validates left identity]
$\Id \tm \eqenf \tm$ for any term $\tm$.
\end{proposition}

\paragraph{Left/Right/Weak Non-Equivalent Variants} Replacing left reduction with right reduction, one obtains a unique decomposition lemma such as \reflemma{las-unique-dec} with respect to right contexts, and, accordingly, a notion of \emph{right Lassen similarity}---let us denote it with $\leqrenf$. It turns out that $\leqenf$ and $\leqrenf$ are different, incomparable similarities. For instance, $\Omega$ and $\Omega (\var\var)$ are enf bisimilar (because they are both $\tolw$-divergent) but not renf bisimilar, because $\leqrenf$ stops on $\Omega (\var\var)$ which is $\torw$ normal. Similarly, $\Omega$ and $\var\var\Omega $ are not enf bisimilar while they are renf bisimilar.

Replacing left reduction with weak reduction is instead problematic for another reason. Since $\tow$ is non-deterministic, the unique decomposition lemma (\reflemma{las-unique-dec}) fails for it. It is not clear then what would be the right definition of weak enf similarity, as the stop-and-go clause can be generalized in more than one way. An appropriate definition for weak enf similarity, as a generalization, should include terms related by enf similarity. However, it is also unclear (to us) how to prove the compatibility of some of such generalizations (we tried but failed\footnote{The definitions we can come up with for weak enf similarities that could be compatible are not able to relate as much terms as enf similarity does.}). 

The next paragraphs discuss the principles that are (in)validated by enf similarity.

\paragraph{$\beta_v$-Conversion}\cadr{
From the definition of $\leqenf$, it immediately follows that enf bisimilarity contains the $\rtobv$ root rule, that is, that if $\tm \rtobv \tmtwo$ then $\tm\eqenf \tmtwo$, simply because $\tm$ and $\tmtwo$ have the same left normal form. Since $\eqenf$ is a compatible equivalence relation, it turns out that $\eqenf$ contains the whole of $\betav$-conversion $=_{\betav}$, thus it contains left reduction as well as weak and right reductions.}{As for naive similarity, enf similarity contains $\beta_v$-conversion, thus it contains left reduction as well as weak and right reductions.}

\begin{proposition}[$\betav$-conversion is validated by enf bisimilarity]
If $\tm =_{\betav} \tmtwo$ then $\tm \eqenf \tmtwo$.
\end{proposition}

\paragraph{Curry and Turing fix-Point Combinators are Enf Bisimilar.} \cadr{It is easy to check that the relation $\relsym$ proving the naive similarity of Curry's and Turing's \cbv fix-point combinators is also a \enf bisimulation.}{The relation $\relsym$ proving that $\curryfix \eqncbv \turingfix$ is not an enf bisimulation. Nonetheless, it is easy to build an enf bisimulation that relates Curry's and Turing's \cbv fix-point combinators.} The fact that those combinators are naive bisimilar means that the \emph{\cadr{unstacking}{unstucking}} aspect of the stop-and-go clause plays no role for their equivalence. 

\begin{proposition}
$\curryfix \eqenf \turingfix$, that is, Curry's and Turing's fix-points are enf bisimilar.
\end{proposition}

\paragraph{$\Omega$-Equivalence} Enf bisimilarity does not validate the $\Omega$-equivalence $\equivomv$. For instance, the same counter-example used for naive similarity works for enf, as we have $\Omega \leqenf \Omega^L$ but $\Omega^L \not\leqenf \Omega$. In fact, $\eqenf$ equates some $\Omega$-terms that are separated by $\eqncbv$ and vice-versa. For instance, let $\delta_3 \defeq \la\var\var\var\var$. The $\Omega$-term $\Omega_3 \defeq \delta_3\delta_3$ is divergent. We have that $\Omega^L_3 \defeq (\la\var\delta_3)(\vartwo\varthree)\delta_3$ is enf bisimilar to $\Omega^L$, because they stop similarly and then both go to diverge. They are instead unrelated with respect to $\leqncbv$, because when compared as normal forms they do not have the same structure. For the vice-versa, consider $\var\var\Omega$ and $\Omega$, which are equated by $\eqncbv$ but separated by $\eqenf$, because $\var\var\Omega$ is left normal but weak divergent. 

\begin{proposition}
Enf bisimilarity does not validate $\Omega$-equivalence $\equivomv$.
\end{proposition}

 The issue of enf with $\equivomv$ concerns clause (enf 1), as left diverging terms are a strict subset of \cbv $\Omega$-terms. To solve it, as usual, there are two options, changing the calculus or the nf-bisimilarity. 
 \citet{DBLP:conf/fossacs/BiernackiLP19} change the nf-bisimilarity, extending (enf 1) to all \emph{deferred diverging terms}, which correspond exactly to their (\cbv+ state) $\Omega$-terms. Their approach does not address the reason why $\Omega$-terms are not equated (namely, the stuck normal forms of Plotkin's \cbv), it only circumvents it, and it does not help when that \cadr{reasons}{reason} affects other aspects such as commuting $\letexp$s (which happens when one of the two $\letexp$s is blocked deep under a stuck redex). Addressing the reason and removing stuck normal forms amounts in fact to \cadr{change}{changing} the calculus. In \refsect{net}, we shall introduce a nf-bisimilarity smoothly validating both $\equivomv$ and commuting $\letexp$s, as it is built on an extension of Plotkin's calculus not suffering from stuck normal forms.

\paragraph{Further Equivalences.} The next proposition sums up the benchmarks for enf.
\begin{toappendix}
\begin{proposition}
\label{prop:enf-validation-of-equivalences}
Enf bisimilarity validates Moggi's equivalences, the left-to-right restrictions of the proof nets equivalences, but it does not validate $\etav$, all of the unrestricted proof nets equivalences, nor \cbn duplication.
\end{proposition}
\end{toappendix}
 Enf similarity does not validate $\etav$, since $\var$ and $\la\vartwo\var\vartwo$ are handled by different clauses in the definition of $\leqenf$. It can however be adapted to validate $\etav$, see \cite{LassenEnf,lassen+strovring-bisimilarity-eta,DBLP:journals/lmcs/BiernackiLP19}. 
About Moggi's equivalences, we have already discussed the left identity. Enf similarity validates all the other ones. Note that in \citeyearpar{LassenEnf} \citeauthor{LassenEnf} claims that enf validates $\equiv_{exrad}$ \adr{(defined in \refsect{benchmarks}) }which is actually false, it only validates $\equiv_{rad}$ (in fact $\equiv_{exrad}$ does not correspond to Moggi's usual right decomposition rule, it shall be motivated by proof nets in \refsect{benchmarks-vsc}). Proofs of the validations are easy, one only needs to write the right relation (the identity relation $\cup$ the equivalence to validate) and show that it is indeed an \enf bisimulation.

The validation of the left-to-right restricted versions of the proof nets equivalences ($\equivsone$ and a restricted version of $\equivexsthree$) is an easy contribution of this paper. 
Simple inspections show that enf does not validate the (order-unspecified) proof nets equivalences $\equivexsthree$ and $\equivcom$, nor \cbn duplication. Consider now renf bisimilarity, the right variant of enf. With respect to proof nets equivalences, it has a sort of dual behavior: it  does not validate $\equivsone$ but instead validates the right-to-left proof nets equivalences ($\equivexsthree$ and a restricted version of $\equivsone$).


\section{Value Substitution Calculus}
\label{sect:vsc}
\begin{figure}
\begin{tabular}{c}
	\!\!\!\!\!\!\!\!\!
$\arraycolsep=3pt
\begin{array}{rrll@{\hspace{1cm}}rrll}
\multicolumn{4}{c}{\textsc{Language}} & \multicolumn{4}{c}{\textsc{Lists of ESs}}
\\[-2pt]
\textsc{Terms } & \vsubterms \ni \tm,\tmtwo, \tmthree & \grameq& \val \mid \tm\tmtwo 
\mid \tm \esub\var\tmtwo  & 
\textsc{Sub. ctxs } &\sctx,\sctxtwo  &\grameq &\ctxhole \mid \sctx \esub\var\tm
\\
\textsc{Values } & \val,\valtwo & \grameq & \var \mid  \la\var\tm 

\end{array}$

\\[4pt]
\hline
\\[-12pt]
\tabcolsep = 2pt
		\!\!\!\!\!\!\!\!\!\!
	\begin{tabular}{c}
\textsc{Weak reduction}
	\\[3pt]
	$\begin{array}{r@{\hspace{.15cm}}r@{\hspace{.1cm}}l@{\hspace{.1cm}}ll}
	
	\textsc{Evaluation Contexts} & \evctx & \grameq &  \ctxhole\mid \tm\evctx\mid \evctx\tm \mid \evctx\esub{\var}{\tmtwo} \mid \tm\esub{\var}{\evctx}
	
	\end{array}$
	\\[2pt]
	\begin{tabular}{c @{\hspace{.3cm}}|@{\hspace{.3cm}} c}
		
		$\begin{array}{rr@{\ }l@{\ }l}
		\multicolumn{4}{c}{\textsc{Root rules}}
		\\
		\textsc{Mult. } & \sctxp{\la\var\tm}\tmtwo &  \rtom  & \sctxp{\tm\esub{\var}{\tmtwo}} 
		\\
		\textsc{Exp.}  & \tm\esub\var{\sctxp{\val}} &  \rtoe  & \sctxp{\tm\isub{\var}{\val}} 
		
	\end{array}$  
		&
		
		$\begin{array}{rr@{\ }l@{\ }l}
		\textsc{Contextual Closure}	
		\\
		\multicolumn{3}{l}{
			\AxiomC{$\tm \rootRew{a} \tm'$}
			\UnaryInfC{$\evctxp{\tm} \Rew{a} \evctxp{\tm'}$}
			\DisplayProof
			\ \ \ 
			a \!\in\! \set{\msym,\esym}
		}
		  
		  \end{array}$  
	\end{tabular}
\\[2pt]
	$\begin{array}{l@{\hspace{.15cm}}r@{\hspace{.1cm}}l@{\hspace{.1cm}}ll}

	\textsc{Notation} & \ \tovsc   & \defeq  & \tom \cup \toe 
	\end{array}$
\\[0pt]
\hline
\\[-12pt]
\textsc{normal forms}
\\[3pt]
	$\arraycolsep = 5pt
	\begin{array}{r@{\hspace{.5cm}}rllr@{\hspace{.5cm}}rll}
	\textsc{Inert terms}  &
	 \itm,\itmtwo & \grameq & \isctxp{\var\ntm} \mid \itm \fire \mid \itm \esub{\var}{\itmtwo}
	& 
	\textsc{ Inert Sub. ctxs} & \isctx & \grameq &  \ctxhole\mid \isctx\esub{\var}{\itm}

	\end{array}$ 
\\[4pt]
$\arraycolsep = 2pt
\begin{array}{r@{\hspace{.5cm}}rll}

\tovsc\textsc{-Normal forms} 
& \ntm,\ntmtwo &\grameq &\val \mid \itm \mid \fire \esub{\var}{\itm}
\end{array}$
\\[-5pt]
\end{tabular}
\end{tabular}
\caption{Value substitution calculus and normal forms.}
\label{fig:vsc}
\end{figure}

%
%
 
Intuitively, the VSC is a \cbv $\lambda$-calculus extended with $\letexp$-expressions, as is common for \cbv $\l$-calculi such as \citeauthor{DBLP:conf/lics/Moggi89}'s one [\citeyear{Moggi88tech,DBLP:conf/lics/Moggi89}]. 
We do however replace a $\letexp$-expression $\letin\var\tmtwo\tm$ with a more compact  \emph{explicit substitution} (ES for short) notation $\tm\esub{\var}{\tmtwo}$, which binds $\var$ in $\tm$ and that has precedence over abstraction and application (that is, $\la\var\tm\esub\vartwo\tmtwo$ stands for $\la\var(\tm\esub\vartwo\tmtwo)$ and $\tm\tmthree\esub\vartwo\tmtwo$ for $\tm(\tmthree\esub\vartwo\tmtwo)$). Moreover, our $\letexp$/ES does not fix an order of evaluation between $\tm$ and $\tmtwo$, in contrast to many papers in the literature (\eg \citet{DBLP:journals/toplas/SabryW97,DBLP:journals/iandc/LevyPT03}) where $\tmtwo$ is evaluated first.

The reduction rules of VSC are slightly unusual as they use \emph{contexts} both to allow one to reduce redexes located in sub-terms, which is standard, \emph{and} to define the redexes themselves, which is less standard---these kind of rules is 
called \emph{at a distance}. The rationale behind is that the rewriting rules are designed to mimic exactly cut-elimination on linear logic proof nets, via \citeauthor{DBLP:journals/tcs/Girard87}'s \citeyearpar{DBLP:journals/tcs/Girard87} \cbv translation $(A \Rightarrow B)^v = ! (A^v \multimap B^v)$ of intuitionistic logic into linear logic, see \citet{Accattoli-proofnets}.

\paragraph{Root rewriting rules}
In VSC, 
there are two main rewrite rules, the \emph{multiplicative} one $\tom$ and the \emph{exponential} one $\toe$ (the terminology comes from the connection between VSC and linear logic), and both work \emph{at a distance}: they use contexts even in the definition of their \emph{root} rules (that is, before the contextual closure). Their definition is based on \emph{substitution contexts} $\sctx$, which are lists of~ES. 
In \Cref{fig:vsc}, the root rule $\rtom$ (resp. $\rtoe$) is assumed to be capture-free, so no free
 variable of $\tmtwo$ (resp. $\tm$) is captured by the substitution context $\sctx$ (by possibly $\alpha$-renaming on-the-fly).

Examples: $(\la\var\vartwo)\esub\vartwo\tm\tmtwo \rtom \vartwo\esub\var\tmtwo\esub\vartwo\tm$ and $(\la\varthree\var\var)\esub\var{\vartwo\esub\vartwo\tm} \rtoe (\la\varthree\vartwo\vartwo)\esub\vartwo\tm$. An example with on-the-fly $\alpha$-renaming is $(\la\var\vartwo)\esub\vartwo\tm\vartwo \rtom \varthree\esub\var\vartwo\esub\varthree\tm$.

A key point is that $\beta$-redexes are decomposed via ES: the \emph{by-value} restriction is on ES-redexes, \emph{not} on $\beta$-redexes, because only values can be substituted.
The multiplicative rule
$\rtom$ fires a $\beta$-redex at a distance and generates an  ES even when the argument is not a value.
	The \cbv discipline is entirely encoded in the exponential rule $\toe$ (see \Cref{fig:vsc}): it can fire an  ES performing a substitution only when its argument is a \emph{value} (\ie a variable or an abstraction) up to a list of ES. This means that only values can be duplicated or erased.

\paragraph{Rewriting Rules} \cadr{We close the root rules by evaluation contexts $\evctx$, which allow reduction everywhere but under abstractions.}{We define weak reduction, noted $\tovsc$, to be the closure of the root rules by evaluation contexts which allow reduction everywhere but under lambdas.} In other papers about the VSC \cite{accattoli+paolini-vsc,accattoli+guerrieri-opencbv,DBLP:conf/lics/AccattoliCC21,DBLP:journals/pacmpl/AccattoliG22}, where other contextual closures are also considered, they are called \emph{weak} or \emph{open} contexts. \cadr{In this work, we only consider weak reduction, noted $\tovsc$.  Examples: }{Examples:}

\adr{\begin{center}
$\begin{array}{r@{\hspace{.1cm}}r@{\hspace{.1cm}}l@{\hspace{.5cm}} r@{\hspace{.1cm}}r@{\hspace{.1cm}}l}
		\tm \esub\var{(\la\vartwo\tmtwo)\esub\varthree\tmthree \tmfour}
		& \tom &\tm \esub\var{\tmtwo\esub\vartwo\tmfour\esub\varthree\tmthree}
		&
		\tm (\var\var) \esub\var{\vartwo\esub\varthree\tmtwo} & \toe & \tm (\vartwo\vartwo)\esub\varthree\tmtwo
		\\[2pt]
		((\var\var) \esub\var{\la\vartwo\varthree} \tm)\esub\varfour\tmtwo 
		& \toe & ((\la\vartwo\varthree) (\la\vartwo\varthree) \tm)\esub\varfour\tmtwo
		&
		\la\varthree (\var\var) \esub\var{\la\vartwo\varthree} & \not\toe&  \la\varthree(\la\vartwo\varthree) \la\vartwo\varthree
\end{array}$
\end{center}
}
\paragraph{Diamond}The $\tovsc$ reduction is non-deterministic, as for instance:
\begin{center}
$\delta (\delta \Id) \ \lRew{\expoabs}\  \vartwo\esub\vartwo\delta (\delta \Id) \ \tom \ \vartwo\esub\vartwo\delta((\var\var)\esub\var\Id).$
\end{center}
It is however more than confluent, it is diamond.
\begin{proposition}[Diamond, \cite{accattoli+paolini-vsc}]
\label{prop:vsc-diamond}
$\tovsc$ is diamond.
\end{proposition}

We now investigate normal forms for terms in the VSC, which admit a, quite complex, inductive description via \emph{inert terms}.

\paragraph{Inert Terms and Normal Forms} 
\cbv is about \emph{values}, and, if terms are closed, normal forms are abstractions. In going beyond the closed setting,  a finer and more general view is required. Normal forms (for $\tovsc$) are given by mutual induction with the notions of \emph{inert term} and \emph{inert substitution contexts}, as in \Cref{fig:vsc}. 
Inert substitution contexts are lists of ESs where the content of every ES is an inert, ensuring that none of the ESs can fire a $\toe$-redex (as inerts are exactly normal forms which are not of the shape $\sctxp\val$).

Examples: $\la\var\vartwo$ is a normal form 
as an abstraction, while $\vartwo(\la\var\var)$, $\var\vartwo$, and $(\varthree(\la\var\var))(\varthree\varthree) (\la\vartwo(\varthree\vartwo))$ are normal forms as inert terms. The grammars also allow to have ES containing inert terms around abstractions and applications: $(\la\var\vartwo)\esub\vartwo{\varthree\varthree}$ is a normal form and $\var\esub\var{\vartwo(\la\var\var)}\vartwo$ is an inert term. One of the key points of inert terms is that they have a \emph{free} head variable (in particular they are open). Inert terms are the \cbv equivalent of \cbn neutral terms. In  
\cite{DBLP:conf/icfp/GregoireL02}, inert terms are called \emph{accumulators}, and normal forms are  called \emph{values}.
Some papers about the VSC adopt a restricted version of the $\tovsc$ reduction, where variables are not values and thus cannot be substituted by the $\toe$-rule. The restriction induces slightly different notions of normal forms (often called \emph{fireballs}) and inert terms (still called \emph{inert terms}), but the difference does not change the main properties of the VSC, as discussed in \cite{DBLP:journals/pacmpl/AccattoliG22}.

\begin{proposition}[Normal forms, \cite{accattoli+paolini-vsc} ]
$\tm$ is $\tovsc$-normal if and only if $\tm$ is a $\ntm$-term as defined in \Cref{fig:vsc}.
\end{proposition}

\paragraph{$\Omega$-Equivalence} One of the features of the VSC is that it solves the issues of Plotkin's calculus with respect to \cbv $\Omega$-terms. Consider the $\Omega$-term $\Omega^L =  (\la\var\delta)(\vartwo\varthree)\delta$ which is contextually equivalent to $\Omega$ but normal for Plotkin. In the VSC, instead, it diverges:
\begin{center}
$ (\la\var\delta)(\vartwo\varthree)\delta \tom \delta \esub\var{\vartwo\varthree}\delta \tom \varfour\varfour\esub\varfour\delta \esub\var{\vartwo\varthree} \toeabs \delta\delta \esub\var{\vartwo\varthree} \tom \ldots$.
\end{center}
More generally, \emph{all} \cbv $\Omega$-terms diverge: the following characterization of \cbv $\Omega$ holds, analogous to the one for \cbn (\refthm{cbn-scrutability-characterization}) and obtained composing the diverging characterization of \cbv inscrutable terms due to \citet{accattoli+paolini-vsc} with an easy proposition (in \refapp{app-vsc}) showing that \cbv inscrutable terms and \cbv $\Omega$-terms coincide.

\begin{toappendix}
\begin{theorem}[VSC diverging characterization of $\Omega$-terms]
	\label{thm:cbv-scrutability-characterization}
$\tm$ is a \cbv $\Omega$-term if and only if $\tm$ is $\tovsc$ diverging.
\end{theorem}
\end{toappendix}
Despite the fact that the VSC makes $\Omega^L$ diverge as in \cbn, it does \emph{not} validate \cbn duplication nor \cbn erasure, as $(\la\var\vartwo)\Omega$ is $\tovsc$-divergent and $(\la\var \vartwo \var \var) (\varthree\varthree) \tom (\vartwo \var \var)\esub\var{\varthree\varthree} \not\tovsc \vartwo (\varthree\varthree)(\varthree\varthree)$.

\paragraph{Contextual Equivalence} Terms such as $\Omega^L$ are normal for Plotkin but divergent in the VSC, that is, the VSC and Plotkin's calculus have different notions of termination. One might then suspect that contextual equivalence in the VSC is not the same as in Plotkin's calculus. This is not the case, in fact, because the VSC behaves differently \emph{only on open terms}, while contextual equivalence is defined reducing \emph{closed terms} only. 

\begin{proposition}[\cite{DBLP:journals/pacmpl/AccattoliG22}]
Two $\l$-terms without ES are contextual equivalent in Plotkin's calculus if and only if they are contextual equivalent in the VSC.
\end{proposition}


\section{Equational Benchmarks and the Value Substitution Calculus}
\label{sect:benchmarks-vsc}
Here, we revisit the equational benchmarks of \refsect{benchmarks} in the VSC. We begin with the proof nets equivalences, as they are one of the \emph{raison d'\^etre} of the VSC and the key to re-understand them all.

\paragraph{Structural Equivalence} The translation of the VSC to proof nets maps some terms with ES to the same proof net. The induced identification of terms is expressed by structural equivalence \cite{accattoli+paolini-vsc,Accattoli-proofnets}. 
\begin{definition}[Structural equivalence $\streq$]
Structural equivalence $\streq$ is defined as the smallest compatible equivalence relation generated by union of the following root rules.
\begin{center}
$\begin{array}{rllrcc}
	(\tm\tmthree)\esub\var\tmtwo & \equivsone&\tm\esub\var\tmtwo\tmthree 	 & \mbox{if }\var \not \in \fv\tmthree
	\\
	(\tm\tmthree)\esub\var\tmtwo &\equivexsthree& 	\tm\tmthree\esub\var\tmtwo & \mbox{if }\var \not \in \fv\tm
	\\
	\tm\esub\var\tmtwo\esub\vartwo\tmthree &\equivass& \tm\esub\var{\tmtwo\esub\vartwo\tmthree} & \mbox{if }\vartwo \not \in \fv\tm
	\\
	\tm\esub\vartwo\tmthree\esub\var\tmtwo &\equivcom& 	\tm\esub\var\tmtwo\esub\vartwo\tmthree & \mbox{if }\var \not \in \fv\tmthree \mbox{ and }\vartwo \not \in \fv\tmtwo
\end{array}$
\end{center}
\end{definition}
These axioms preserve the number and type of the constructors in terms, they only rearrange the order. In particular, structurally equivalent terms have the same number of ES. Note that the axioms simply express the constructor-and-scope-preserving commutation of ES with applications and ES themselves (but not abstractions, as that would break \refprop{strong-bisimulation} below). 

Additionally, structural equivalence behaves very well with respect to evaluation: it commutes with reduction rules---and is therefore postponable---preserving the number and kind of steps. This is expressed by the following proposition. In the literature, what is below called \emph{strong commutation} is usually called \emph{strong bisimulation}. We prefer to change the terminology here to avoid confusion with \emph{nf-(bi)simulations}, as the concept is similar and yet different (no need to observe normal forms, and it preserves the number of steps).

\begin{toappendix}
\begin{proposition}[$\streq$ strongly commutes with $\tovsc$, \citet{accattoli+paolini-vsc}]
	\label{prop:strong-bisimulation}
	Let $a \in \set{\msym,\expoabs,\expovar}$. Structural equivalence $\streq$ strongly commutes with $\tovsc$:
	if  $\tm \streq\tmtwo$ and $ \tm \Rew{a}\tmp$ then $\tmtwo \Rew{a}\tmtwop$ and $\tmp\streq\tmtwop$.
\end{proposition}
\end{toappendix}

In fact, $\streq$-equivalence classes are an isomorphic representation of proof nets, as the proof net $P_\tm$ associated to $\tm$ does the same exact rewriting steps as $\tm$, that is, the translation from $\tm$ to $P_\tm$ also strongly commutes with evaluation (turning term steps into proof nets steps, and vice-versa), see \citet{Accattoli-proofnets}. Consequently, $\streq$-equivalent terms are \emph{indistinguishable} and should be equated by any sensible notion of program equivalence on pure terms (extensions with effects can invalidate some cases of structural equivalence, typically $\equivcom$, as we discuss below). In particular, structural equivalence is included in contextual equivalence, as shown by \citet{DBLP:journals/pacmpl/AccattoliG22}.

\paragraph{Revisiting the Benchmarks From Calculi} The equivalences of Moggi's can be re-understood via structural equivalence. The idea is that by applying $\tom$ to the two sides of an equivalence, we can express it via ES, and many of become cases of $\streq$. Consider $\equiv_{ass}$:
	\begin{center}
		$\begin{array}{ccccc}
	(\la\vartwo\tm)((\la\var\tmtwo)\tmthree) &\equiv_{ass}& (\la\var(\la\vartwo\tm)\tmtwo)\tmthree) & \mbox{with }\var\notin\la\vartwo\tm
	\\
\downarrow_{\mult} && \downarrow_{\mult}
	\\
(\la\vartwo\tm)\tmtwo\esub\var\tmthree &\equivexsthree& ((\la\vartwo\tm)\tmtwo)\esub\var\tmthree
\end{array}$
	\end{center}

Similarly, the equivalences $\equivexsthree$ and $\equivcom$ of \refsect{benchmarks} become the axioms with the corresponding label of $\streq$. Therefore, structural equivalence covers the proof nets equivalences. Moggi's equivalences $\equivlid$, $\equivlad$ and $\equivrad$, instead, are not covered. By applying $\tom$, we obtain the following reformulation for the VSC which corresponds to Moggi's original formulation with $\letexp$-expressions. If $\var\notin \fv\tmtwo$ for $\equivlad$ and $\var\notin \fv\val$ for $\equivrad$:
\begin{center}
\begin{tabular}{c@{\hspace{1.5cm}} c@{\hspace{1.5cm}} c}
		$\begin{array}{ccccc}
	(\la\var\var)\tm &\equivlid& \tm 
	\\
\downarrow_{\mult} && =
	\\
\var\esub\var\tm &\equivlid& \tm
\end{array}$
&
$\begin{array}{ccccc}
	(\la\var\var\tmtwo)\tm &\equivlad& \tm\tmtwo 
	\\
\downarrow_{\mult} && =
	\\
(\var\tmtwo)\esub\var\tm &\equivlad& \tm\tmtwo
\end{array}$
&
$\begin{array}{ccccc}
	(\la\var\val\var)\tm &\equivrad& \val\tm 
	\\
\downarrow_{\mult} && \downarrow_{\mult}
	\\
(\val\var)\esub\var\tm &\equivrad&  \val\tm
\end{array}$
\end{tabular}
	\end{center}
In presence of structural equivalence, the ES formulation of the application decomposition equivalences $\equivlad$ and $\equivrad$ is derivable:  $\equivlad$ follows from $\equivsone$ and $\equivlid$, $\equivrad$ follows from \cadr{$\equivsthree$}{$\equivexsthree$} and $\equivlid$. In fact, \cadr{by using the extended version $\equivexsthree$ of $\equivsthree$ }{by using the (unrestricted) proof nets equivalence $\equivexsthree$ }we can actually derive the extended version of $(\tmtwo\var)\esub\var\tm \equivexrad  \tmtwo\tm$ (if $\var\notin\fv\tmtwo$) of $\equivrad$. Therefore, the whole of Moggi's equivalences is captured by simply adding $\equivlid$ to structural equivalence.

\paragraph{A Family of Strong Commutations} It turns out that various sub-relations of structural equivalence also verify strong commutation (to allow one to verify the claims of this paragraph, the proof of \refprop{strong-bisimulation} is in Appendix D of the additional material on HotCRP, as the proof is omitted from \citet{accattoli+paolini-vsc} where the result first appeared). We here describe them by the root rules, while implicitly referring to the same closure used in the definition of $\streq$. For instance, $\equivcom$ by itself strongly commutes with $\tovsc$, as well as $\equivsone$ by itself, or $\equivexsthree\cup\equivass$, or some of the restricted version such as ${\equivsthree}\cup\equivass$\adr{, where $\equivsthree$ is the \emph{left-to-right} restriction of $\equivexsthree$}. In particular, $\equivsone\cup\equivsthree\cup\equivass$, which would be the restriction of $\streq$ to a non-commutative setting for effects, also strongly commutes with $\tovsc$. In the next section, we shall craft a nf-similarity for the VSC which is parametric with respect to these variants of structural equivalence.

\paragraph{Shuffling} \citeauthor{shufflingcalculus}'s \emph{shuffling calculus} \citeyearpar{shufflingcalculus} is yet another extension of Plotkin's calculus with the \cbv variants of two shuffling rules introduced in \cbn by \citet{regnier94}, and it is used in semantical studies about \cbv \cite{DBLP:journals/lmcs/GuerrieriPR17,bohmtree-cbv}. Its \cbv shuffling rules are also instances of structural equivalence if one applies $\tom$ steps as discussed above.
\section{Net Similarity for the Value Substitution Calculus}
\label{sect:net}
In this section, we finally define the nf-similarity for the VSC we are interested in, \emph{net similarity}, which extends $\leqncbv$ along two axes:
\begin{enumerate}
\item \emph{Changing the underlying calculus:} evaluation is now based on the VSC, where $\Omega$-terms have a diverging characterization, hence $\Omega$-equivalence $\equivom$ will be trivially included in the bisimilarity, and 
\item \emph{Changing the nf-bisimilarity}: allowing simulations to test terms modulo structural equivalence $\streq$, in order to avoid artificial distinctions of indistinguishable terms and to validate the commutation of $\letexp$s $\equivcom$. 
\end{enumerate}
The problematic addition of the left identity equivalence $\equivlid$ is discussed at the end of the section.

\paragraph{Changing the Underlying Calculus} Moving from Plotkin's \cbv calculus to the VSC, normal forms are harder to describe, as one can see from the grammar of normal forms in Figure \ref{fig:vsc}.

Since the case $\isctxp\var\ntm$ is quite unpleasant to manage in proofs, we go one step further and consider normal forms modulo $\equivsone$, picking the $\equivsone$-representative of each normal form where the context $\isctx$ has been pushed out of the unpleasant case for inert terms. This can be done harmlessly because $\equivsone$ by itself verifies strong \emph{commutation} with respect to $\tovsc$ and the described representant can be easily described at the big-step level. The following lemma gives a grammar for normal forms modulo $\equivsone$; we overload/redefine \emph{inert terms}, which from now on only refer to the new grammar.

\begin{lemma}
$\tm$ is $\tovsc$-normal iff there exists $\ntm$, given by the following grammar, such that $\tm \equivsone \ntm$.
\begin{center}
	$\begin{array}{r@{\hspace{.5cm}}rlll}
		\textsc{Applicative Inert terms}  &
		\itmapp,\itmapptwo & \grameq &  \var\fire\mid \itmapp\fire
		\\
		\textsc{Inert terms}  &
		\itm,\itmtwo & \grameq &  \itmapp\mid \itm\esub\var\itmtwo
		\\
		\tovsc\textsc{-normal forms modulo $\equivsone$}  &
 \ntm,\ntmtwo & \grameq &  \val \mid \itm\mid \fire\esub\var\itmtwo
	\end{array}
	$\end{center}
\end{lemma}
Note the notion of applicative inert terms, which are specific inert terms where no substitutions can be pushed outward by $\equivsone$. 
Example: the $\tovsc$-normal form $\tm =\var\esub\var{\vartwo\vartwo}\varthree$ does not belong to grammar above, but $\tm \equivsone (\var\varthree)\esub\var{\vartwo\vartwo} = \ntm$ and $\ntm$ is described by the grammar.

\paragraph{Big-Step Evaluation} We then need to express evaluation to $\tovsc$-normal form modulo $\equivsone$ as a big-step predicate. For that, we re-define the inert substitution contexts, which at first sight are defined as before, except that the notion of inert term now has changed.
\begin{center}
	$\begin{array}{r@{\hspace{.5cm}}rlll}
		\textsc{Inert Substitution Contexts} & \isctx & \grameq &  \ctxhole\mid \isctx\esub{\var}{\itm} 
	\end{array}
	$\end{center}

\begin{definition}[Big-step evaluation $\bsvsct k$]
Big-step $\VSC$ modulo $\equivsone$ evaluation $\tm \bsvsct k \ntm$ is given by:
\begin{center}

\begin{tabular}{cccccc}

	\infer[(\bsvsctax)]{\val \bsvsct 0 \val}{}
	
	&
	
	\infer[(\bsvsctapm)]{\tm\tmtwo \bsvsct {k+i+1} \isctxp\fire}{
		\tm \bsvsct k \isctxp{\la\var\tmthree}
		&
		{\tmthree\esub\var\tmtwo} \bsvsct i \fire
	}

	\\[6pt]
	\infer[(\bsvsctapvar)]{\tm\tmtwo \bsvsct {k+h} \isctxp{\var\fire}}{
		\tm \bsvsct k \isctxp\var
		&
		\tmtwo \bsvsct h \fire
	}
	&
	\infer[(\bsvsctese)]{\tm\esub\var{\tmtwo} \bsvsct {k+i+1} \isctxp\fire}{
		\tmtwo \bsvsct k \isctxp{\val}
		&
		{\tm\isub\var{\val}} \bsvsct i \fire
	}

	\\[6pt]
	
	\infer[(\bsvsctapi)]{\tm\tmtwo \bsvsct {k+h} \isctxp{\itmapp\fire}}{
		\tm \bsvsct k \isctxp\itmapp
		&
		\tmtwo \bsvsct h \fire
	}
	&
\infer[(\bsvsctesi)]{\tm\esub\var\tmtwo \bsvsct {k+h} \fire\esub\var\itm}{
	\tm \bsvsct k \fire
	&
	\tmtwo \bsvsct h \itm
}
\end{tabular}

\end{center}
\emph{Notation}: $\tm \bsvscts \fire$ abbreviates \emph{there exists a $k$ such that $\tm \bsvsct k \fire$}.
\end{definition}

 The given big-step system captures $\equivsone$ via the rules ($\bsvsctapvar$) and ($\bsvsctapi$): when applying an inert term / variable surrounded by an inert context $\isctx$ to a normal term $\ntm$, the context $\isctx$ is pushed out of the application, obtaining $\isctxp{\var\fire}$ and $\isctxp{\itmapp\fire}$ instead of $\isctxp{\var}\fire$ and $\isctxp{\itmapp}\fire$. As a result, the $\bsvsctsym$-normal forms are exactly those of the given grammar for $\tovsc$-normal forms modulo $\equivsone$.

	We prove this big-step system to be correct and complete with respect to small-step reduction. 	Importantly, substitutivity also smoothly adapts.
	\begin{toappendix}
	\begin{proposition}
		\label{l:ss-bs-equivalence_vsce}
		$\tm \bsvsct k \fire$ if and only if there exists a normal form $\firep$ such that $\tm \tovsc^k \firep \equivsone \fire$.
	\end{proposition}
	\end{toappendix}
	

	\begin{toappendix}
	\begin{proposition}[Substitutivity]
		\label{prop:substitutivity_vsce}
	\hfill
	\begin{enumerate}
	\item 
	\emph{Small-step}: if $\tm~\tovsc~\tmp$ then $\tm\isubst\val\var ~\tovsc~ \tmp\isubst\val\var$.
	\item 
	\emph{Big-step}: 	if $\tm\isubst\val\var \bsvsct k \ntm$ then $\exists$ $k'$ and $\ntmtwo$ such that $ \tm \bsvsct {k'} \ntmtwo$ and $\ntmtwo\isubst\val\var\bsvsct {k-k'} \ntm$.
	\end{enumerate}
\end{proposition}
\end{toappendix}

	\paragraph{Changing the Nf-Bisimilarity by Adding Structural Equivalence to Simulations, Parametrically} We are now also going to refine the definition of naive similarity by adding structural equivalence $\streq$ to the nf-bisimilarity. In fact, we are going to do something more general, in order to obtain a whole family of similarities. We abstract away structural equivalence $\streq$ as a more abstract notion of \emph{mirror equivalence} $\equivx$, defined by the properties of $\streq$ that are needed to prove that similarity modulo $\streq$ is compatible. Then, similarity is defined \emph{parametrically} in a mirror $\equivx$, and net similarity is obtained by taking $\streq$ as mirror. The terminology \emph{mirror} is meant to suggest that $\equivx$ can modify terms only in inessential ways.
	
\begin{definition}[Mirror]
An equivalence relation $\equivx$ is a \emph{mirror} for $\tovsc$ when:
\begin{enumerate}
\item \emph{Strong commutation}: if  $\tm \equivx\tmtwo$ and $ \tm \tovsc\tmp$ then $\tmtwo \tovsc\tmtwop$ and $\tmp\equivx\tmtwop$.

\item \emph{Substitutivity}: if $\tm\equivx\tmtwo$ then $\tm\isub\var\val \equivx \tmtwo\isub\var\val$ for all values $\val$.
\end{enumerate}
\end{definition}

\begin{definition}[Mirrored and \net similarities]
	Let $\relsym$ be relation and $\equivx$ be a mirror over VSC terms.	We say that $\relsym$ is a \emph{$\equivx$-mirrored (\nafex) simulation} if $\relsym\subseteq\relvscx$, where $\tm \relvscx\tmp$ holds whenever $\tm,\tmp$ satisfy one of the following clauses:
	\begin{center}
		$\begin{array}{r@{\hspace{.3cm}}r@{\hspace{.3cm}}l@{\hspace{.3cm}}l@{\hspace{.3cm}}lll}
		\textup{(\nafex 1)} & &&\tm\bsvsctdiv & \ie ~ \text{has no} \tovsc \text{-normal form.}
		\\
		\textup{(\nafex 2)} & \tm \bsvscts \var  &\text{and}& \tmp \bsvscts \var
		\\
		\textup{(\nafex 3)} & \tm \bsvscts \la\var\tmfirst &\text{and}& \tmp \bsvscts\la\var\tmpfirst
		& \text{with} ~ \tmfirst \rel \tmpfirst
		\\
		\textup{(\nafex 4)} & \tm \bsvscts \ntmONE \ntmTWO &\text{and}& \tmp \bsvscts \ntmtwo \equivx \ntmONEtwo \ntmTWOtwo
		& \text{with} ~ \ntmONE \rel \ntmONEtwo ~\text{and}~ \ntmTWO \rel \ntmTWOtwo
		\\
		\textup{(\nafex 5)} & \tm \bsvscts \ntmONE\esub\var\ntmTWO &\text{and}& \tmp \bsvscts \ntmtwo \equivx \ntmONEtwo\esub\var\ntmTWOtwo
		& \text{with} ~ \ntmONE \rel \ntmONEtwo ~\text{and}~ \ntmTWO \rel \ntmTWOtwo
	\end{array}
	$\end{center}
		$\equivx$-Mirrored (\nafex) similarity , written $ \leqvscx $, is defined the largest $\equivx$-mirrored simulation.
		
		\Net simulations and \net similarity $\leqnet$ are defined as the $\equivx$-mirrored simulations and similarities with structural equivalence $\streq$ as mirror $\equivx$. 
\end{definition}

\paragraph{Making Inert Terms Explicit in the Clauses} Cases (\nafex 4) and (\nafex 5) can be rewritten using the grammar of normal forms, which is useful for clarity in proofs. For (\nafex 4), it actually splits in two:
	\begin{center}
		$\begin{array}{r@{\hspace{.3cm}}r@{\hspace{.3cm}}l@{\hspace{.3cm}}l@{\hspace{.3cm}}lll}
		
		\text{(\nafex 4a)} & \tm \bsvscts \var \ntmONE &\textit{and}& \tmp \bsvscts \ntmtwo \equivx\var \ntmONEtwo 
		& \textit{with}~\ntmONE \rel \ntmONEtwo
		\\
		\text{(\nafex 4b)} & \tm \bsvscts \itmapp \ntmONE &\textit{and}& \tmp \bsvscts \ntmtwo \equivx\itmapptwo \ntmONEtwo 
		& \textit{with} ~ \itmapp \rel \itmapptwo ~\textit{and}~ \ntmONE \rel \ntmONEtwo
		\\
		\text{(\nafex 5)} & \tm \bsvscts \ntmONE\esub\var\itm &\textit{and}& \tmp \bsvscts \ntmtwo \equivx\ntmONEtwo\esub\var\itmtwo
	& \textit{with} ~ \itm \rel \itmtwo ~\textit{and}~ \ntmONE \rel \ntmONEtwo
	\end{array}
	$\end{center}

\paragraph{Compatibility} The compatibility proof for $\leqvscx$ follows the same structure of the one for $\leqncbv$ (in Appendix E of the additional material on HotCRP). At the evaluation level, we have already seen that substitutivity holds also for the weak reduction of the VSC and with the addition of the equivalence $\equivsone$ (\refprop{substitutivity_vsce}). At the level of the simulation, we need to refine the notion of Lassen closure, by adding rule $\mscequivx$ accounting for mirrors.
\begin{definition}[Mirrored Lassen closure]
Let the \emph{mirrored Lassen closure} $\mlasrelsym$ of $\relsym$ be:
	\begin{center}
\begin{tabular}{cccccc} 
\begin{tabular}{cccccc} 
	\infer[\msclift ]{\tmrone \mlasrel \tmrtwo} {\tmrone \rel \tmrtwo}
	&
	\infer[\mscvar]{\var \mlasrel \var}	{}
	&
	\infer[\mscabs ]{\la\var\tmrone \mlasrel \la\var\tmrtwo} {\tmrone \mlasrel \tmrtwo}
	&
		\infer[\mscapp ] {\tmrone\tmrthree  \mlasrel  \tmrtwo\tmrfour} {\tmrone  \mlasrel \tmrtwo & \tmrthree \mlasrel \tmrfour }  
		
\end{tabular}
\\[14pt]
\begin{tabular}{cccccc}

	\infer[\mscesub ]{\tmrone\esub\var{\tmrthree} \mlasrel \tmrtwo\esub\var{\tmrfour}{}} {\tmrone \mlasrel \tmrtwo & \tmrthree \mlasrel \tmrfour }
&
	\infer[\mscsub ]{\tmrone\isub\var{\valof\tmrthree} \mlasrel \tmrtwo\isub\var{\valof\tmrfour}{}} {\tmrone \mlasrel \tmrtwo & \valof\tmrthree \mlasrel \valof\tmrfour }	
	&
\infer[\mscequivx]{\tmrone\mlasrel\tmrtwo}{\tmronep \mlasrel\tmrtwop & \tmrtwop \equivx \tmrtwo}
\end{tabular}
\end{tabular}		
	\end{center}
\end{definition}
Then the reasoning for compatibility---and in particular the coherence properties---smoothly adapts, using the mirror properties for rule $\mscequivx$ in the proof that the closure preserves mirrored simulations. In particular, strong commutation of $\equivx$ implies that it preserves normal forms and steps, that is, the coherence properties. Summing up, we obtain our main result.
\begin{toappendix}
\begin{theorem}[Compatibility and soundness of $\leqvscx$ and $\leqnet$]
	\label{thm:nafex-included-leqc}
Let $\equivx$ be a mirror.
	\begin{enumerate}
	\item \emph{Redundancy of the mirrored Lassen closure}: $\leqvscx \,= \mlassenop \leqvscx$.
	\item \Nafex similarity $\leqvscx$ is compatible and included in the \cbv contextual preorder $\leqcv$.
	\item \Net similarity $\leqnet$ is compatible and included in the \cbv contextual preorder $\leqcv$.
	\end{enumerate}
\end{theorem}
\end{toappendix}

\paragraph{Fixpoints and Benchmarks.} For any mirror $\equivx$, and in particular for $\equivx\defeq Id$ and $\equivx\defeq \streq$, one can show that Turing's and Curry's \cbv fixpoint combinators are \nafex bisimilar. The proof relies on exactly the same relation that for naive bisimilarity (\refprop{naive-fix-points-equiv}). \adr{Plotkin's $\betav$ and VSC conversions are included in \nafex similarities. }Unlike naive and enf similarities, \nafex and net similarities validate $\Omega_v$-equivalence $\equivomv$. Net and $\equivx$-mirrored bisimilarities do not however validate $\eta_v$ equivalence, one has to change the case for abstractions to accommodate it, and it does not validate \cbn duplication, as for instance $(\vartwo\var\var)\esub\var{\varthree\Id}$ and $\vartwo(\varthree\Id)(\varthree\Id)$ are both $\tovsc$-normal but not $\streq$-equivalent.

%



\paragraph{Left Identity Is Not Validated By \Nafex} Analogously, \net similarity does not validate Moggi's $\equivlid$ rule, because $\var\esub\var{\vartwo\Id} \not \tovsc \vartwo\Id$ and  $\var\esub\var{\vartwo\Id} \not \streq \vartwo\Id$. Thus, \enf is not included in \net similarity. About adding $\equivlid$, it is easy to define a \nafexp\equivlid bisimulation, but the current  compatibility proof does not go through, as $\equivlid$ is not a mirror for $\tovsc$ (in particular, it does not strongly commute with $\tovsc$) and the proof technique is not able (for now) to handle $\equivlid$ terms, as it breaks coherence for normal forms, that is, the fact that if  $\ntm\,\mlasrelsym \tm$ then $\tm$ is normal (and the symmetric statement).

One could also add $\equivlid$ as a reduction step of the VSC, but then the reduction is no longer diamond, and the diamond property (or at least the invariance of the number of steps to normal form) is essential in the current proof technique. It is thus unclear how to extend \net similarity as to validate $\equivlid$. Next section introduces a program equivalence including $\eqnet$ and validating $\equivlid$.

\paragraph{Net Bisimilarity Cannot Be Extensional} Lassen introduced an extension of enf bisimilarity validating $\equivetav$, enf bisimilarity up to $\eta_v$, at the same time that he introduced enf bisimilarity \cite{LassenEnf}. As of now, that modification cannot be applied to net bisimilarity. We explain the problem which boils down to, again, the fact that net does not validate the left identity law $\equivlid$.
	
	Let us consider that there exists a nf-bisimilarity $\relsym$ based on the VSC (\ie $\tm\tovsc\tmp$ implies $\tm\rel\tmp$) which is a compatible equivalence relation, and which validates $\equivetav$. Then, in particular, we have that $\var\rel\la\vartwo\var\vartwo$, as $\equivetav\subseteq\relsym$. By compatibility (for $\ctx=\ctxhole\tm$), $\var\tm\rel(\la\vartwo\var\vartwo)\tm$ and by reduction and transitivity, $\var\tm\rel\var\vartwo\esub\vartwo\tm$. This means that at least $(\var\tm,\var\vartwo\esub\vartwo\tm)$ must be included in the relation $\relsym$, which is a subcase of $\equivrad$, which itself can be implied by the $\equivlid$ rule and structural equivalence $\streq$. As net bisimilarity does not validate Moggi's laws, and we do not currently know how to include them, net bisimilarity is unable to include $\equivetav$.



\ignore{
\paragraph{Fixed point combinators are \nafex bisimilar.} As Lassen did with \enf bisimilarity, we can prove the equivalence of call-by-value versions of Curry's and Turing's fixed point combinators:

\[ \curryfix = \la\var{\curryfixaux\curryfixaux}\text{, where } \curryfixaux = \la\varthree{\var\la\vartwo{\varthree\varthree\vartwo}}\]
\[ \turingfix = (\la\varthree{\la\var{\var\la\vartwo{\varthree\varthree\var\vartwo}}})(\la\varthree{\la\var{\var\la\vartwo{\varthree\varthree\var\vartwo}}}) \]

To prove that they are \nafe bisimilar we build a bisimulation containing $\{(\curryfix,\turingfix)\}$.

\[ \relsym \defeq \{(\curryfix,\turingfix), (\la\var\curryfixaux\curryfixaux,\la\var{\var\la\vartwo{\turingfix\var\vartwo}}),(\curryfixaux\curryfixaux,{\var\la\vartwo{\turingfix\var\vartwo}}), \]\[ 
(\var\la\vartwo{\curryfixaux\curryfixaux\vartwo},{\var\la\vartwo{\turingfix\var\vartwo}}),(\var,\var),(\la\vartwo{\curryfixaux\curryfixaux\vartwo},\la\vartwo{\turingfix\var\vartwo}),\]
\[({\curryfixaux\curryfixaux\vartwo},{\turingfix\var\vartwo}),((\var\la\vartwo{\curryfixaux\curryfixaux\vartwo})\vartwo,({\var\la\vartwo{\turingfix\var\vartwo}})\vartwo),(\vartwo,\vartwo)\} \]

$\relsym \subseteq \opnafep{\relsym}$ by construction (we start with $(\curryfix,\turingfix)$ and we add to $\relsym$ what is needed for each element to satisfy one \nafe case), and sym(R) is a \nafe simulation as well, hence $\curryfix \nafebisim \turingfix$.

This relation is similar to the one defined by Lassen since those terms have "pure lambda-calculus" normal forms (nothing more can be done at the end with VSC and explicit substitutions).
}

\section{From Operational to Denotational Semantics: the Type Preorder}
\label{sect:type-preorder}
In this section, we study a behavioral preorder, the \emph{type preorder} $\leqtype$, which is not defined as a nf-similarity, it is instead induced by a denotational model. Namely, \citeauthor{DBLP:conf/csl/Ehrhard12}'s \cbv relational model \citeyearpar{DBLP:conf/csl/Ehrhard12} presented as a system of multi types, also known as \emph{non-idempotent intersection types}. We shall prove that both Lassen's similarity $\leqenf$ and our net similarity $\leqnet$ are included in $\leqtype$. The aim is to show that, while $\leqenf$ and $\leqnet$ are incomparable, they can be combined in a cost-sensitive preorder (the contextual preorder combines them but it is not cost-sensitive). We introduce the bare minimum about \cbv multi types. For more, see \cite{Accattoli-Guerrieri-TypesFireballs,DBLP:journals/pacmpl/AccattoliG22}.


%
%

\begin{figure}
\begin{tabular}{c}
		$\begin{array}{ccccc}
		\textsc{Linear Types} & \ltype, \ltypetwo &\grameq&
		\mtype \multimap \mtypetwo
		\\
		\textsc{Multi Types} & \mtype, \mtypetwo &\grameq& \multitype{n}{\ltype} & n\geq 0
		\end{array}$
		
		\\[10pt]

		\begin{tabular}{ccc}
			\infer[\typingruleAx]{\var \hastype [\ltype] \types \var \hastype \ltype}{}
			
			&
			
			\infer[\typingruleAbs]{\typectx \types \la\var\tm \hastype \mtype \multimap \mtypetwo}{\typectx, \var \hastype \mtype \types \tm \hastype\mtypetwo}
			
			&
			
			\infer[\typingruleMany]{\biguplus_{i\in I} \typectx_i \types \val \hastype \biguplus_{i\in I} [\ltype_i]}{(\typectx_i \types \val \hastype \ltype_i)_{i\in I}  & I~ \text{finite} }

		\end{tabular}
		\\[8pt]
		\begin{tabular}{cc}
			\infer[\typingruleApp]{\typectx \uplus \typectxtwo \types \tm\tmtwo \hastype \mtypetwo}{ \typectx \types \tm \hastype [\mtype \multimap \mtypetwo] & \typectxtwo \types \tmtwo \hastype \mtype }
			&
			\infer[\typingruleES]{\typectx \uplus \typectxtwo \types \tm\esub\var\tmtwo \hastype \mtypetwo}{ \typectx, \var \hastype \mtype \types \tm \hastype \mtypetwo & \typectxtwo \types \tmtwo \hastype \mtype }
			
		\end{tabular}
\\[-5pt]
\end{tabular}
\caption{Call-by-Value Multi Type System for VSC.}
\label{fig:multi-types-vsc}
\end{figure}

%
%
 
\paragraph{Multi Types} \Cref{fig:multi-types-vsc} gives the definition of multi types $\mtype$ for the VSC, which  mutually depends on the definition of linear types $\ltype$. Multi types are defined as finite multi-sets $\multitype{n}{\ltype}$, which intuitively denote the intersection $\ltype_1 \cap \ldots \cap \ltype_n$, where the intersection $\cap$ is a commutative, associative and non-idempotent ($A \cap A \not = A$) operator, the neutral element of which is $\emptytype$, the empty multi set.
Note that there is no ground type, its role is played by the empty multi type $\emptytype$.

A typing judgment is of the shape $\typectx \types \tm \hastype T$ where $T$ is a linear or a multi type and $\typectx$ is a typing context, that is an assignment of multi types to a finite set of variables ($\typectx = \var_1 \hastype \mtype_1, \ldots, \var_n \hastype \mtype_n$). A typing derivation $\typeder \derives \typectx \types \tm \hastype \mtype$ is a tree built from the derivation rules defined in \Cref{fig:multi-types-vsc} which ends with the typing judgment $\typectx \types \tm \hastype \mtype$.

\paragraph{Typing Rules} Linear types only type values, via the rules $\typingruleAx$ and $\typingruleAbs$. To give a multi type to value $\val$, one has to use the $\typingruleMany$ rule, turning an indexed family of linear types for $\val$ into a multi type. Note that any value can be typed with the empty multi type $\emptytype$. 
The symbol $\uplus$ is the disjoint union operator on multi sets (which corresponds to our non-idempotent intersection on multi types).  

\paragraph{Characterization of Termination} A key property of multi types is that they characterize $\tovsc$ termination. The characterization is proved via subject reduction and expansion.

\begin{theorem}[Characterization of termination, \cite{DBLP:journals/pacmpl/AccattoliG22}]
\label{thm:mtypes-charac}
\hfill
\begin{enumerate}
\item 	\label{p:mtypes-charac-subject} \emph{Subject reduction and expansion}:	let $\tm \tovsc \tmtwo$. Then $\typectx \types \tm \hastype \mtype$ if and only if $\typectx \types \tmtwo \hastype \mtype$.

\item $\tm$ is $\tovsc$-terminating if and only if there exists $\typectx$ and $\mtype$ such that $\typectx \types \tm \hastype \mtype$.
\end{enumerate}
\end{theorem}
Since $\tovsc$-divergence characterizes \cbv $\Omega$-terms (\refthm{cbv-scrutability-characterization}), \emph{not being typable} with multi types characterizes it too, by the previous theorem.

\paragraph{Multi Types Induce a Model}
Multi types induce a model
by interpreting a term 
as the set of its type judgments. 
A possibly empty list of pairwise distinct variables $\vec{\var} = (\var_1, \dots, \var_n)$ is \emph{suitable for} $\tm$ if $\fv{\tm} \subseteq \{\var_1, \dots, \var_n\}$.
If $\vec{\var} = (\var_1, \dots, \var_n)$ is suitable for $\tm$, the \emph{semantics} $\sem{\tm}_{\vec{\var}}$ \emph{of} $\tm$ \emph{for} $\vec{\var}$ is given by:\adr{
\begin{center}$
	\sem{\tm}_{\vec{\var}} \defeq \{((\mtypetwo_1,\dots, \mtypetwo_n),\mtype) \mid 
	\exists 
	\, 
	\concl{\tderiv}{\var_1 \hastype \mtypetwo_1, \dots, \var_n \hastype \mtypetwo_n}{\tm}{\mtype} \}$
\end{center}}
This is exactly \citeauthor{DBLP:conf/csl/Ehrhard12}'s \cbv relational model \citeyearpar{DBLP:conf/csl/Ehrhard12}. Ehrhard considers it with respect to Plotkin's calculus. We do not prove that it is a model for the VSC, because there is no formal notion of VSC model. We do have, however, that subject reduction and expansion (\refthmp{mtypes-charac}{subject}) ensure that the interpretation $\sem{\tm}_{\vec{\var}}$ is \emph{invariant} by $\tovsc$, and compatibility of the induced equational theory is proved below. These properties are what the definitions of $\l$-models or categorical models are meant to ensure. Moreover, the characterization theorem (\refthm{mtypes-charac}) ensures that $\sem{\tm}_{\vec{\var}}$ is adequate.

\begin{corollary}[\cite{DBLP:journals/pacmpl/AccattoliG22}]
	\label{thm:invariance-and-adequacy}
	Let $\tm$ be a term in the \VSC with $\vec{\var}  = (\var_1, \dots, \var_n)$ suitable for it.
\begin{enumerate}
	\item \emph{Invariance}: if $\tm (\tovsc \cup \streq) \tmtwo$ then $\sem{\tm}_{\vec{\var}} = \sem{\tmtwo}_{\vec{\var}}$.
	\item \emph{Adequacy for} $\tovsc$: $\sem{\tm}_{\vec{\var}}$ is non-empty if and only if $\tm$ is $\tovsc$-terminating.

\end{enumerate}
\end{corollary}

\paragraph{The Type Preorder} Every model $M$ induces an equational theory defined as $\tm =_M \tmtwo$ if $\sem\tm_M = \sem\tmtwo_M$. For the multi types model, we consider the \emph{preorder} induced by $\interp\tm_{\vec\var} \subseteq \interp\tmp_{\vec\var}$.
\begin{definition}[Type preorder]
The type preorder $\tm \leqtype \tmp$ holds if $\typectx \types \tm \hastype \mtype$ implies $\typectx \types \tmp \hastype \mtype$.
\end{definition}
Rephrasing the definition with respect to interpretations, we have that $\tm \leqtype \tmp$ if $\interp\tm_{\vec\var} \subseteq \interp\tmp_{\vec\var}$ for every list of suitable variables $\vec\var$. By the adequacy of $\interp\tm_{\vec\var}$, it follows the adequacy of $\leqtype$. Compatibility is easily proved directly, for once, and soundness follows.

\begin{toappendix}
\begin{proposition}[Compatibility of $\leqtype$] \label{prop:type-preorder-is-compatible}
\hfill
\begin{enumerate}
\item \emph{Compatibility}: if $\tm\leqtype\tmp$ then $\ctxp\tm \leqtype \ctxp\tmp$.
\item \emph{Soundness}: if $\tm\leqtype\tmp$ then $\tm\leqcv\tmp$.
\end{enumerate}
\end{proposition}
\end{toappendix}

\paragraph{Enf and Net Are Included in Type}
Now, we show that $\leqenf$ and $\leqnet$ are both included in $\leqtype$. For that, we prove that, if $\tm\leqenf\tmp$ or $\tm\leqnet \tmp$, then any typing derivation for $\tm$ can be transformed in a typing derivation for $\tmp$ having the same final judgement. The next two propositions and the associated lemma are proved by induction on the following notion: the \emph{size} $\size\typeder$ of a type derivation $\typeder$, defined as the number of rule occurences in $\typeder$ except for rule $\typingruleMany$.

\begin{toappendix}
\begin{proposition}
	\label{prop:bisimulation-preserves-typeder}
\hfill
\begin{enumerate}
\item \emph{Net simulations and type derivations}:
	let $\relsym$ be a \net simulation. If $\tm \rel \tmp$ and $\typeder: \typectx \types \tm \hastype \mtype$ then there exists a derivation $\typederp: \typectx \types \tmp \hastype \mtype$.
	\item \emph{Net is included in Type}: if $\tm\leqnet\tmp$ then $\tm\leqtype\tmp$.
	\end{enumerate}
\end{proposition}
\end{toappendix}

To relate \enf similarity and typability, we need a lemma to deal with Lassen's stop-and-go. 
\begin{toappendix}
\begin{lemma}[Stop-and-go and type derivations]
	\label{l:smaller-derivations-stuck}
	Let $\typeder \derives \typectx \types {\levctxp{\var\val}} \hastype \mtype$ and $\varthree$ be fresh. Then there exist $\typeder_\levctx \derives  \typectx_\levctx, \varthree \hastype \mtypetwo \types  \levctxp{\varthree} \hastype \mtype$ with $\size{\typeder_\levctx}<\size\typeder$ and $\typeder_\val \derives \typectx_\val \types \val : \mtypetwo_1$ with $\size{\typeder_\val}<\size\typeder$.
\end{lemma}
\end{toappendix}

\begin{toappendix}
\begin{proposition}
	\label{l:enf-bisimulation-preserves-typeder}
	\hfill
\begin{enumerate}
\item \emph{Enf simulations and type derivations}:
	let $\relsym$ be an \enf simulation. If $\tm \rel \tmp$ and $\typeder \derives \typectx \types \tm \hastype \mtype$ then there exists a derivation $\typederp \derives \typectx \types \tmp \hastype \mtype$.
	\item \emph{Enf is included in Type}: if $\tm\leqenf\tmp$ then $\tm\leqtype\tmp$.
	\end{enumerate}
\end{proposition}
\end{toappendix}

\ignore{\paragraph{$\etav$ Reduction} Concerning $\etav$ equivalence, the multi type system we consider does not (fully) validate it. In \cbn but this is standard and it is usually fixable by adding a recursive equation on the ground type. To our knowledge, however, the question has not been studied in \cbv. 

Part of $\eta_v$ equivalence, that is $\eta_v$-reduction, is however validated by the type preorder. On the other side, $\eta_v$-expansion fails. Hence none of them is validated by the symmetric closure of the type preorder, namely type equivalence. We first show the result on $\eta_v$-reduction for variables.
Actually, $\eta_v$ equivalence on closed values is standard and validated by $\alpha$-equivalence and reduction steps below lambdas. Therefore we can describe exactly how $\eta_v$ is included in the type preorder.
\begin{toappendix}
\begin{proposition} Let $\var$ a variable and $\val$ a value. One then has:
	\label{prop:etav-for-leqtype}
	\begin{enumerate}
		\item \emph{(Variable $\eta_v$-equivalence)} $\la\vartwo\var\vartwo \leqtype \var$, but $\var \not \leqtype \la\vartwo\var\vartwo$,
		\item  \emph{(Value $\eta_v$-equivalence)} $\la\vartwo\val\vartwo \leqtype \val$ for any $\val$ and $\val \leqtype \la\vartwo\val\vartwo$ iff $\val$ is an abstraction.
	\end{enumerate}
\end{proposition}
\end{toappendix}}

\paragraph{$\eta_v$ Equivalence} 
By the fact that \enf and \net similarities are incomparable follows that they are strictly included in $\leqtype$. A further  gap between the type preorder $\leqtype$ and $\leqenf$ or $\leqnet$ is $\eta_v$ equivalence, which is included in $\leqtype$ but not in $\leqenf$ nor $\leqnet$.

\begin{toappendix}
	\begin{proposition}[$\eta_v$-equivalence is included in type equivalence]	\label{prop:etav-for-leqtypetwo}
	Let $\var$ a variable. Then $\la\vartwo\var\vartwo \leqtype \var$ and $\var \leqtype \la\vartwo\var\vartwo$.
	\end{proposition}
\end{toappendix}

\paragraph{Characterizing Type Equivalence} We conjecture that $\leqtype$ is exactly the sup of the \enf and \net similarities enriched with $\eta_v$ equivalence, that is, that generalizing $\leqnet$ as to validate $\equivlid$ and $\equivetav$ would match $\leqtype$. If the conjecture is false, finding a nf-similarity presentation of $\leqtype$---which corresponds to describe the equational theory of \cbv relational semantics---is anyway an interesting and challenging problem.

About full abstraction with respect to \cbv contextual equivalence $\eqcv$, it fails for $\equivtype$, as $\equivtype$ is cost-sensitive---it does not validate \cbn duplication---while $\eqcv$ is cost-insensitive.

\section{Conclusions}
Motivated by the fact that Lassen's enf bisimilarity $\eqenf$---the normal form bisimilarity of reference in \cbv---does not identify $\Omega$-terms and commuting $\letexp$s, we introduced \emph{net bisimilarity} $\eqnet$, which does identify them. It turns out, however, that $\eqnet$ does not validate Moggi's laws nor $\etav$, which are instead validated by $\eqenf$. Additionally, it is unclear how to extend either enf or net bisimilarity as to catch the other one.

Such a problematic duality led us to develop a sharp analysis of \cbv and of the principles that can be validated or not by normal form bisimulations. The analysis shows that the semantic landscape of \cbv is considerably richer and more sophisticated than the \cbn one. 

Concretely, our analysis contributed two further equivalences. First, a naive bisimilarity $\eqncbv$, that mainly provides a better understanding of Lassen's tricky definition of enf simulations. Second, the type equivalence $\equivtype$ induced by Ehrhard's multi types, which subsumes both enf and net bisimilarity, and includes $\etav$-equivalence, while retaining their cost-sensitive aspect. In practice, $\equivtype$ is not really usable for comparing programs, but it provides a sharp theoretical tool.

\paragraph{Future Work} Type equivalence suggests that it could be possible to find a normal form bisimilarity merging the enf and net ones. We are actively working on this challenging problem. A related question is finding an axiomatization of $\equivtype$, for which some sort of separation theorem should be developed.
We would also like to investigate how net bisimilarity $\eqnet$ relates to the topics connected to $\eqenf$, such as  game semantics \cite{DBLP:conf/lics/JaberM21}, extensions with effects \cite{DBLP:conf/esop/LagoG19,DBLP:conf/fossacs/BiernackiLP19,biernacki_et_al:LIPIcs:2020:12329}, and the $\pi$-calculus \cite{DBLP:journals/tcs/DurierHS22}.


\bibliographystyle{ACM-Reference-Format}
\bibliography{\macrospath/biblio_ICFP}

\newpage
\appendix

\section{Proofs from \refsect{naive} (Naive CbV Bisimilarity)}
\label{chapter:proof-compatibility-naive}
In this section, we develop the proof of compatibility for the (weak) naive similarity, following Lassen's variant of Howe method for nf-bisimulations.

\subsection{Proof of Equivalence of Small-Step and Big-Step Operational Semantics} We prove the big-step evaluation predicate sound and complete with respect to the small-step operational semantics (Proposition \ref{l:ss-bs-equivalence_weak}).

\begin{lemma}
	\label{l:aux-ss-bs-equivalence}
	If $\tm\tow\tmp$ and $\tmp\bsw k \ntm$, then $\tm\bsw {k+1} \ntm$.
\end{lemma}

\begin{proof}
	Straightforward proof by structural induction.
\end{proof}

\gettoappendix{l:ss-bs-equivalence_weak}

\begin{proof}
	Trivial using \reflemma{aux-ss-bs-equivalence}.
\end{proof}

\subsection{Lemmas about normal forms and $\lasrel$ and $\lasrelncbv$}
Actually, before proving \refprop{main-lemma_naive} for all terms, we somehow need to prove it only on normal forms. More precisely, we show that $\lasrelsym$ and $\lasrelncbvsym$ coincide on normal forms (\reflemma{lasrelncbv-normal-forms-lasrel-left-to-right} and \reflemma{lasrelncbv-normal-forms-lasrel-right-to-left}).

\begin{lemma}
	\label{l:lasrelncbv-normal-forms-lasrel-left-to-right}
	If $\ntm\lasrelncbv\ntmtwo$ then $\ntm\lasrel\ntmtwo$.
\end{lemma}
\begin{proof}
	By case analysis on the shape of $\ntm$.
\end{proof}

Notice that the next lemma already proves part of the conclusion of the first part of Proposition \ref{prop:ncbv-coherence}.

\begin{lemma}[Constrained Substitutivity of $\lasrelncbv$ on normal forms]
	\label{l:lasrelncbv-normal-forms-substitutive}
	If $\ntm \lasrelncbv \ntmtwo$, $\val \lasrelncbv \valtwo$ and $\ntm\isub\var{\val}$ and $\ntmtwo\isub\var{\valtwo}$ are $\tow$-normal then $\ntm\isub\var{\val} \lasrelncbv \ntmtwo\isub\var{\valtwo}$.
\end{lemma}

\begin{proof}
	By case analysis on the shape of $\ntm$. Cases:
	\begin{itemize}
		\item $\ntm = \var$ and $\ntmtwo = \var$ then $\ntm\isub\var{\val} = \val \lasrelncbv \valtwo = \ntmtwo\isub\var{\valtwo}$.
		
		\item $\ntm = \vartwo$ and $\ntmtwo = \vartwo$ then $\ntm\isub\var{\val} =  \vartwo \lasrelncbv \vartwo = \ntmtwo\isub\var{\valtwo}$.
		
		\item $\ntm = \la\vartwo\tm$ and $\ntmtwo = \la\vartwo\tmp$ with $\tm \lasrel \tmp$
		we have \[\infer{\tm\isub\var{\val} \lasrel \tmp\isub\var{\valtwo}}{\tm \lasrel \tmp & \val \lasrel \valtwo}\]
		hence by case (ncbv 3) $\ntm\isub\var{\val} = \la\vartwo{\tm\isub\var{\val}} \lasrelncbv  \la\vartwo{\tmp\isub\var{\valtwo}} = \ntmtwo\isub\var{\valtwo}$.

		\item $\ntm = \ntmONE\ntmTWO$ and $\ntmtwo = \ntmONEtwo\ntmTWOtwo$ with $\ntmONE\lasrel\ntmONEtwo$ and $\ntmTWO\lasrel\ntmTWOtwo$. Also, by \reflemma{lasrelncbv-normal-forms-lasrel-left-to-right}, $\val\lasrel\valtwo$.
		
			\[ \infer[\scsub]{\ntmONE\isub\var\val \lasrel \ntmONEtwo\isub\var\valtwo}{\ntmONE\lasrel\ntmONEtwo & \val \lasrel \valtwo} ~\text{and}~ \infer[\scsub]{\ntmTWO\isub\var\val \lasrel \ntmTWOtwo\isub\var\valtwo}{\ntmTWO \lasrel \ntmTWOtwo & \val \lasrel \valtwo}\]

		Hence, as $\ntm\isub\var\val$ and $\ntm\isub\var\valtwo$ are normal forms, this concludes the proof.

	\end{itemize}
\end{proof}

\begin{lemma}
	\label{l:lasrelncbv-normal-forms-lasrel-right-to-left}
	If $\relsym$ is a naive simulation.
	If $\ntm\lasrel\ntmtwo$ then $\ntm\lasrelncbv\ntmtwo$.
\end{lemma}

\begin{proof}
	By induction on the derivation $\ntm \lasrel \ntmtwo$. Cases of the last rule in the derivation of $\ntm\lasrel\ntmtwo$:
	\begin{itemize}
		\item \emph {$\scvar$}\[ \infer[\scvar]{\var \lasrel \var}{} \]
		then $\var \lasrelncbv \var$ by definition of naive.
		\item \emph {$\scabs$} \[ \infer[\scabs]{\ntm = \la\var\tm \lasrel \la\var\tmp = \ntmtwo}{\tm \lasrel \tmp} \]
		then $\ntm \lasrelncbv \ntmtwo$ by definition of naive with $\tm \lasrel \tmp$.
		\item \emph {$\sclift$} \[ \infer[\scabs]{\ntm \lasrel \ntmtwo}{\ntm \rel \ntmtwo} \]
		then since $\relsym$ is a naive simulation $\ntm \relncbv \ntmtwo$, hence by monotonicity of $\ncbvfp\cdot$, $\ntm \lasrelncbv \ntmtwo$.
		\item \emph {$\scapp$} \[ \infer[\scapp]{\ntm = \ntmONE\ntmTWO \lasrel \ntmONEtwo\ntmTWOtwo = \ntmtwo}{\ntmONE \lasrel \ntmONEtwo & \ntmTWO \lasrel \ntmTWOtwo} \]then $\ntm \lasrelncbv \ntmtwo$ by definition of naive with $\ntmONE \lasrel \ntmONEtwo$ and  $\ntmTWO \lasrel \ntmTWOtwo$.
		\item \emph {$\scsub$} \[ \infer[\scsub]{\ntm = \ntmONE\isub\var\val \lasrel \ntmTWO\isub\var\valtwo = \ntmtwo}{\ntmONE \lasrel \ntmTWO & \val \lasrel \valtwo} \]
		by \ih we have $\ntmONE \lasrelncbv \ntmTWO$ and $\val \lasrelncbv \valtwo$. By \reflemma{lasrelncbv-normal-forms-substitutive}, $\ntmONE\isub\var{\val} \lasrelncbv \ntmTWO\isub\var{\valtwo}$.\qedhere
	\end{itemize}
\end{proof}

\subsection{Coherence of simulation, reduction and substitution}
The main difficulty in the proof of \refprop{main-lemma_naive} is the case of the ($\scsub$) rule. It is dealt with the following lemma, that states that substitution behaves nicely with the simulation and the reduction of the calculus.

\gettoappendix{prop:ncbv-coherence}

\begin{proof}
	\begin{enumerate}
		\item We only prove that $\ntmtwo\isub\var\valtwo$ is a normal form, the rest of the conclusion follows by \ref{l:lasrelncbv-normal-forms-substitutive}. By induction on $\ntm$.
		\begin{itemize}
			\item $\ntm=\vartwo$, then $\ntmtwo=\vartwo$ (since $\ntm\lasrelncbv\ntmtwo$), which is a normal form.
			\item $\ntm=\var$, then $\ntmtwo=\var$ (since $\ntm\lasrelncbv\ntmtwo$), hence the result.
			\item $\ntm=\la\vartwo\tm$, then $\ntmtwo=\la\vartwo\tmp$ with $\tm\lasrel\tmp$ (since $\ntm\lasrelncbv\ntmtwo$), we conclude since $\la\vartwo\tmp\isub\var\valtwo$ is a normal form.
			\item $\ntm = \ntmONE\ntmTWO$, then $\ntmtwo=\ntmONEtwo\ntmTWOtwo$ with $\ntmONE\lasrel\ntmONEtwo$ and $\ntmTWO\lasrel\ntmTWOtwo$. By \reflemma{lasrelncbv-normal-forms-lasrel-right-to-left}, $\ntmONE\lasrelncbv\ntmONEtwo$ and $\ntmTWO\lasrelncbv\ntmTWOtwo$. Hence we can apply the \ih and get that $\ntmONEtwo\isub\var\valtwo$ and $\ntmTWOtwo\isub\var\valtwo$ are normal forms. In fact, by a quick analysis, $\ntmONEtwo\ntmTWOtwo\isub\var\valtwo$ is a normal form (otherwise $\ntm\isub\var\val$ would not be normal).
		\end{itemize}
		\item By induction on $\ntm$. Note that the only possible shape for $\ntm$ is $\ntmONE\ntmTWO$, otherwise the substitution of a value does not imply a reduction step.
		
		\begin{itemize}
			\item $\ntmONE\isub\var\val$ is normal and $\ntmTWO\isub\var\val$ is normal. Then $\ntmONE\isub\var\val = \la\vartwo\tmrone$ and $\ntmTWO\isub\var\val = \val_1$.
			By Point 1, $\ntmONEtwo = \la\vartwo\tmrtwo$ and $\ntmTWOtwo=\val_2$ such that $\tmrone\lasrel\tmrtwop$ and $\val_1 \lasrel \val_2$ (by \ref{l:lasrelncbv-normal-forms-lasrel-left-to-right}). Hence, $\ntmtwo\isub\var\val \tow \tmrtwo\isub\vartwo{\val_2}$ and we conclude using the following derivation:
			\[\infer{\tmrone\isub\vartwo{\val_1} \lasrel  \tmrtwo\isub\vartwo{\val_2}}{\tmrone\lasrel\tmrtwo & \val_1 \lasrel \val_2}\]
			
			\item $\ntmONE\isub\var\val\tow\tm_1$ or $\ntmTWO\isub\var\val\tow\tm_1$.
			
			By \ih, $\ntmONEtwo\isub\var\valtwo\tow\tmtwo_1$ or $\ntmTWOtwo\isub\var\valtwo\tow\tmtwo_1$ such that $\tm_1\lasrel\tmtwo_1$, hence we find $\tmtwo$ easily in both cases, such that $\ntmtwo\isub\var\val \tow \tmtwo$ and $\tm\lasrel\tmtwo$.
		\end{itemize}
	\end{enumerate}
\end{proof}

\subsection{Lassen's closure preserves naive simulations}

Finally, we prove that the Lassen's closure of a naive simulation is a simulation.

\gettoappendix{prop:main-lemma_naive}
\begin{proof}
	


By case analysis on the last rule of the derivation $\tmrone\lasrel \tmrtwo$.
\begin{enumerate}
	\item \emph{Lifting}:
	\[ \infer[(\sclift) ]{\tmrone \lasrel \tmrtwo} {\tmrone \rel \tmrtwo}\text{ and }\tmrone\bsw k \ntm\]
	Since $\relsym$ is a naive simulation, we have $\tmrone\relncbv\tmrtwo$ and $\tmrtwo \bsws \ntmtwo$ for some $\ntmtwo$ such that $\ntm\relncbv\ntmtwo$. Hence  $\ntm\lasrelncbv\ntmtwo$ by monotonicity of $\ncbvfp\cdot$.
	
	\item \emph{Variables}:
	\[\infer[(\scvar) ]{\var \lasrel \var}	{} \text{ and } \var\bsw 0 \var\]
	
	hence the result ($\var\bsw 0 \var$) and $\var\lasrelncbv\var$ by definition of $\lasrelncbvsym$.
	
	\item \emph{Abstraction}:
	\[\infer[(\scabs) ]{\la\var\tmrone \lasrel \la\var\tmrtwo} {\tmrone \lasrel \tmrtwo} \text{ and } \la\var\tmrone \bsw 0 \la\var\tmrone \]
	
	hence the result ($\la\var\tmrtwo \bsw 0 \la\var\tmrtwo$) and $\la\var\tmrone \lasrelncbv \la\var\tmrtwo$ by definition of $\lasrelncbvsym$.
	\item \emph{Application}:
	\[ \infer[(sc.app) ] 
	{\tmrone\tmrthree  \lasrel  \tmrtwo\tmrfour} {\tmrone  \lasrel \tmrtwo & \tmrthree \lasrel \tmrfour } \text{ and }\tmrone\tmrthree \bsw k \ntm \]
	
	then, by case analysis on the last rule of the big-step derivation:
	
	\begin{itemize}
		\item \emph{Applied normal form:}
		
	\begin{center}
			$\infer{\tmrone\tmrthree\bsw {k+h} \ntm = \ntmONE\ntmTWO}{ \tmrone\bsw {k} \ntmONE & \tmrthree\bsw {h} \ntmTWO}$
	\end{center}

		by inductive hypothesis ($d$ strictly decreasing, first component not increasing) we obtain $\tmrtwo \bsws \ntmONEtwo$ and $\tmrfour \bsws \ntmTWOtwo$ with $\ntmONE\lasrelncbv\ntmONEtwo,~\ntmTWO \lasrelncbv \ntmTWOtwo$. Then we need two facts to conclude:
		\begin{itemize}
			\item $\ntmONEtwo\ntmTWOtwo$ is the normal form of $\tmrtwo\tmrfour$:
			Suppose it is not a normal form, \ie $\ntmONEtwo=\la\var\tmthree$ and $\ntmTWOtwo=\val$. By $\ntmONE\lasrelncbv\ntmONEtwo,~\ntmTWO \lasrelncbv \ntmTWOtwo$ and the definition of naive simulations, $\ntmONE$ must be an abstraction and $\ntmTWO$ must be a value, contradicting the fact that $\ntmONE\ntmTWO$ is a normal form. Hence:
			
			$\infer{\tmrtwo\tmrfour\bsw {k'+h'} \ntmtwo = \ntmONEtwo\ntmTWOtwo}{ \tmrtwo\bsw {k'} \ntmONEtwo & \tmrfour\bsw {h'} \ntmTWOtwo}$
			
			\item $\ntmONE\ntmTWO\lasrelncbv\ntmONEtwo\ntmTWOtwo$: which is clear from the first point and since, by \reflemma{lasrelncbv-normal-forms-lasrel-left-to-right}, $\ntmONE\lasrel\ntmONEtwo,~\ntmTWO \lasrel \ntmTWOtwo$. 
			
		\end{itemize}
		\item \emph{$\beta_v$ step}:
		\[\infer{\tmrone\tmrthree \bsw {k+h+i+1} \ntm}{
			\tmrone \bsw k {\la\var\tmronep}
			&
			\tmrthree \bsw h \val
			&
			{\tmronep\isub\var\val} \bsw i \ntm
		}\]

		then by inductive hypothesis ($d$ strictly decreasing, first component non increasing) on $\tmrone$ and $\tmrthree$ we get
		$\tmrtwo\bsws \la\var\tmrtwop$ with $\la\var\tmronep \lasrelncbv \la\var\tmrtwop$ and $\tmrfour\bsws \valtwo$ with $\val \lasrelncbv \valtwo$. In particular, by \reflemma{lasrelncbv-normal-forms-lasrel-left-to-right}, $\val\lasrel\valtwo$.

		Then:
		\[\infer{\tmronep\isub\var\val \lasrel \tmrtwop\isub\var\valtwo}{\tmronep\lasrel\tmrtwop&\val\lasrel\valtwo}\]
		
		since $\tmronep\isub\var\tmrthree \bsw i \ntm$ with $i < k+h+i+1$ we can apply the inductive hypothesis on the first component for $\tmronep\isub\var\val$ obtaining $\tmrtwop\isub\var\valtwo \bsws  \ntmtwo$ for some $\ntmtwo$ such that $\ntm\lasrelncbv\ntmtwo$. 
		
		Last, note that $\tmrtwo\tmrfour\bsws\ntmtwo$ by 
		\[\infer{\tmrtwo\tmrfour \bsws \ntmtwo}{
			\tmrtwo \bsws \la\var\tmrtwop
			&
			\tmrfour \bsws \valtwo
			&
			\tmrtwop\isub\var\valtwo \bsws \ntmtwo
		}\]
	\end{itemize}
	
%

	\ignore{
	\begin{enumerate}
		\item \emph{Applied inert}: 
		\[\infer{\tmrone\tmrthree \bsw {k+h} \itm\ntm}{
			\tmrone \bsw k \itm
			&
			\tmrthree \bsw h \ntm
		}\]
		
		by inductive hypothesis ($d$ strictly decreasing, first component not increasing) we obtain $\tmrtwo \bsws \itmtwo$ and $\tmrfour \bsws \ntmtwo$ with $\itm\lasrelncbv\itmtwo,~\ntm \lasrelncbv \ntmtwo$. In particular, by \reflemma{lasrelncbv-normal-forms-lasrel-left-to-right}, $\itm\lasrel\itmtwo,~\ntm \lasrel \ntmtwo$. Then:
		\[\infer{\tmrtwo\tmrfour \bsws \itmtwo\ntmtwo}{
			\tmrtwo \bsws  \itmtwo
			&
			\tmrfour \bsws \ntmtwo
		}\]
		and $\itm\ntm \lasrelncbv \itmtwo\ntmtwo$ by definition of $\lasrelncbv$ and $\itm\lasrel\itmtwo,~\ntm \lasrel \ntmtwo$.
		

		\item \emph{$m$ step}:
		\[\infer{\tmrone\tmrthree \bsw {k+i+1} \isctxp\ntm}{
			\tmrone \bsw k \isctxp{\la\var\tmronep}
			&
			{\tmronep\esub\var\tmrthree} \bsw i \ntm
		}\]

		then by inductive hypothesis ($d$ strictly decreasing, first component non increasing) on $\tmrone$ we get
		$\tmrtwo\bsws \ntm_\tmrtwo$ with $\isctxp{\la\var\tmronep} \lasrelncbv \ntm_\tmrtwo$ ($\ntm_\tmrtwo = \isctxtwop{\la\var\tmrtwop}$ by  \reflemma{abstraction-inerts-stable-lasrelncbv}) \ie $\isctxp{\la\var\tmronep} \lasrelncbv \isctxtwop{\la\var\tmrtwop}$.

		By \reflemma{lasrelncbv-values-isctx-decomposition} we get $\isctxp{\var} \lasrelncbv \isctxtwop{\var}$ and $\la\var\tmronep\lasrelncbv \la\var\tmrtwop$ then $\tmronep\lasrel\tmrtwop$ by case (nai 3)).
		
		Then:
		\[\infer{\tmronep\esub\var\tmrthree \lasrel \tmrtwop\esub\var\tmrfour}{\tmronep\lasrel\tmrtwop&\tmrthree\lasrel\tmrfour}\]
		
		since $\tmronep\esub\var\tmrthree \bsw i \ntm$ with $i < k+i+1$ we can apply the inductive hypothesis on the first component for $\tmronep\esub\var\tmrthree$ obtaining $\tmrtwop\esub\var\tmrfour \bsws  \ntmtwo$ for some $\ntmtwo$ such that $\ntm\lasrelncbv\ntmtwo$. Since $\isctxp{\var} \lasrelncbv \isctxtwop{\var}$ and $\ntm\lasrelncbv\ntmtwo$, by \reflemma{lasrelncbv-normal-forms-isctx-decomposition}, we get $\isctxp{\ntm} \lasrelncbv \isctxtwop{\ntmtwo}$.
		Last, note that $\tmrtwo\tmrfour\bsws\isctxtwop\ntmtwo$ by 
		\[\infer{\tmrtwo\tmrfour \bsws \isctxtwop\ntmtwo}{
			\tmrtwo \bsws \isctxtwop{\la\var\tmrtwop}
			&
			\tmrtwop\esub\var\tmrfour \bsws \ntmtwo
		}\]
		
%
	\end{enumerate}
	
	\item \emph{Explicit Substitution}: 
	\[ \infer[(sc.esubst) ]{\tmrone\esub\var{\tmrthree} \lasrel \tmrtwo\esub\var{\tmrfour}} {\tmrone \lasrel \tmrtwo & \tmrthree \lasrel \tmrfour }\text{ and }\tmrone\esub\var{\tmrthree} \bsw k \ntm \]
	
	\begin{enumerate}
		\item \emph{Substitution of an inert}:
		\[ \infer{\tmrone\esub\var\tmrthree \bsw {k+h} \ntm\esub\var\itm}{
			\tmrone \bsw k \ntm
			&
			\tmrthree \bsw h \itm
		} \]
		
		by inductive hypothesis ($d$ strictly decreasing, first component not increasing) we obtain $\tmrtwo \bsws \ntmtwo$ and $\tmrfour \bsws \itmtwo$ with $\itm\lasrelncbv\itmtwo,~\ntm \lasrelncbv \ntmtwo$. In particular, by \reflemma{lasrelncbv-normal-forms-lasrel-left-to-right}, $\itm\lasrel\itmtwo,~\ntm \lasrel \ntmtwo$. Then:
		\[\infer{\tmrtwo\esub\var\tmrfour \bsws \ntmtwo\esub\var\itmtwo}{
			\tmrtwo \bsws  \ntmtwo
			&
			\tmrfour \bsws \itmtwo
		}\]
		and $\ntm\esub\var\itm \lasrelncbv \ntmtwo\esub\var\itmtwo$ by definition of $\lasrelncbvsym$.
		
		\item \emph{e step}: 		\[ \infer{\tmrone\esub\var{\tmrthree} \bsw {k+i+1} \isctxp\ntm}{
			\tmrthree \bsw k \isctxp{\la\vartwo\tmrthreep}
			&
			{\tmrone\isub\var{\la\vartwo\tmrthreep}} \bsw i \ntm
		} \]
		
		then by inductive hypothesis ($d$ strictly decreasing, first component non increasing) on $\tmrone$ and $\tmrthree$ we get
		$\tmrfour\bsws \ntm_\tmrfour$ with $\isctxp{\la\vartwo\tmrthreep} \lasrelncbv \ntm_\tmrfour$($\ntm_\tmrfour = \isctxtwop{\la\vartwo\tmrfourp}$ by  \reflemma{abstraction-inerts-stable-lasrelncbv}) \ie $\isctxp{\la\vartwo\tmrthreep} \lasrelncbv \isctxtwop{\la\vartwo\tmrfourp}$.
		
		By \reflemma{lasrelncbv-values-isctx-decomposition} we get $\isctxp{\var} \lasrelncbv \isctxtwop{\var}$ and $\la\vartwo\tmrthreep\lasrelncbv\la\vartwo\tmrfourp$ and in particular\\ $\la\vartwo\tmrthreep\lasrel\la\vartwo\tmrfourp$ by \ref{l:lasrelncbv-normal-forms-lasrel-left-to-right}.
		
		Then:
		\[\infer{\tmrone\isub\var{\la\vartwo\tmrthreep} \lasrel \tmrtwo\isub\var{\la\vartwo\tmrfourp}}{\tmrone\lasrel\tmrtwo&\la\vartwo\tmrthreep\lasrel\la\vartwo\tmrfourp}\]
		
		since $\tmrone\isub\var{\la\vartwo\tmrthreep} \bsw i \ntm$ with $i < k+i+1$ we can apply the inductive hypothesis on the first component for $\tmrone\isub\var{\la\vartwo\tmrthreep}$ obtaining $\tmrtwo\isub\var{\la\vartwo\tmrfourp} \bsws  \ntmtwo$ for some $\ntmtwo$ such that $\ntm\lasrelncbv\ntmtwo$. 
		Since $\isctxp{\var} \lasrelncbv \isctxtwop{\var}$ and $\ntm\lasrelncbv\ntmtwo$, by \reflemma{lasrelncbv-normal-forms-isctx-decomposition}, we get $\isctxp{\ntm} \lasrelncbv \isctxtwop{\ntmtwo}$.
		
		Last, note that $\tmrtwo\tmrfour\bsws\isctxtwop\ntmtwo$ by 
		\[\infer{\tmrtwo\tmrfour \bsws \isctxtwop\ntmtwo}{
			\tmrfour \bsws \isctxtwop{\la\vartwo\tmrfourp}
			&
			\tmrtwo\isub\var{\la\vartwo\tmrfourp} \bsws \ntmtwo
		}\]
		
	\end{enumerate}}

	\item \emph{Meta-level Substitution}: 
	\[ \infer[(sc.subst) ]{\tmrone\isub\var{\val} \lasrel \tmrtwo\isub\var{\valtwo}} {\tmrone \lasrel \tmrtwo & \val \lasrel \valtwo }\text{ and }\tmrone\isub\var{\val} \bsw k \ntm \]
	then by applying Big-step substitutivity (\reflemma{splitting_weak}), we obtain $\tmrone \bsw {k_1} \ntm_\tmrone$ and $\ntm_\tmrone\isub\var{\val} \bsw {k_2} \ntm$ with $k=k_1+k_2$. Hence by inductive hypothesis ($d$ strictly decreasing, first component non increasing) $\tmrtwo\bsws\ntm_\tmrtwo$ and $\ntm_\tmrone \lasrelncbv \ntm_\tmrtwo$. In particular, by \reflemma{lasrelncbv-normal-forms-lasrel-left-to-right}, $\ntm_\tmrone \lasrel \ntm_\tmrtwo$.
	We then have: 
	\[\infer{\ntm_\tmrone\isub\var{\val} \lasrel \ntm_\tmrtwo\isub\var{\valtwo}} {\ntm_\tmrone \lasrel \ntm_\tmrtwo & \val \lasrel \valtwo }\]
	
	
	Two cases.
	\begin{enumerate}

		\item \emph{$\tmrone$ is not normal}, that is, $k_1>0$ and $k_2<k$. Then by applying the induction hypothesis to $k_2$ (first component)  and $\ntm_\tmrone\isub\var{\val}$ we obtain $\ntm_\tmrtwo\isub\var{\valtwo} \bsws \ntmtwo$ with $\ntm\lasrelncbv\ntmtwo$. We conclude using substitutivity of $\tow$ (\ref{l:stability_weak}) that $\tmrtwo\isub\var{\valtwo} \tow^* \ntm_\tmrtwo\isub\var{\valtwo} \tow^* \ntmtwo$, hence via the equivalence between big and small steps (\ref{l:ss-bs-equivalence_weak}), $\tmrtwo\isub\var{\valtwo} \bsws \ntmtwo$.

		\item \emph{$\tmrone$ is normal}, that is, $k_1=0$ and $k_2=k$. Then $\tmrone=\ntm_\tmrone$. Two sub-cases:
		\begin{itemize}
			\item \emph{$\tmrone\isub\var{\val} = \ntm_\tmrone\isub\var{\val}$ is also normal}

			Since we know that $\ntm_\tmrone \lasrelncbv \ntm_\tmrtwo$ and $\val \lasrel \valtwo $, we can apply Point 1 of \refprop{ncbv-coherence} and obtain that $\ntm_\tmrtwo\isub\var{\valtwo}$ is $\tow$-normal and  $\tmrone\isub\var{\val} \lasrelncbv \ntm_\tmrtwo\isub\var{\valtwo}$. It is only left to show that $\tmrtwo\isub\var{\valtwo} \bsws \ntm_\tmrtwo\isub\var{\valtwo}$, which follows from $\tmrtwo\bsws \ntm_\tmrtwo$, substitutivity of $\tow$ (\reflemma{stability_weak}) and the fact that $\ntm_\tmrtwo\isub\var{\valtwo}$ is $\tow$-normal (and via \reflemma{ss-bs-equivalence_weak}).
			
			\item   \emph{$\tmrone\isub\var{\val} = \ntm_\tmrone\isub\var{\val}$ is not normal}
			
			hence $\ntm_\tmrone\isub\var{\val} \tow \tmronep \tow ^ {k-1} \ntm$ (the reduction is diamond, all reductions are of the same length, we pick any first step possible). Then by Point 2 of \refprop{ncbv-coherence} with $\ntm_\tmrone \lasrelncbv \ntm_\tmrtwo$, $\ntm_\tmrtwo\isub\var{\valtwo} \tow \tmrtwop$ with $\tmronep \lasrel \tmrtwop$.
			
			We can apply the inductive hypothesis to $\tmronep$ (first component is decreasing, as $k-1<k$) and we obtain $\tmrtwop \bsws \ntmtwo$ with $\ntm\lasrelncbv\ntmtwo$.
			The statement is then proved, since (using \reflemma{stability_weak})
			$$\tmrtwo\isub\var{\valtwo} \tow^* \ntm_\tmrtwo\isub\var{\valtwo} \tow \tmrtwop \tow^* \ntmtwo$$ that is, $\tmrtwo\isub\var{\valtwo} \bsws \ntmtwo$ by \reflemma{ss-bs-equivalence_weak}.\qedhere
			
		\end{itemize}
	\end{enumerate}
	
\end{enumerate}

\end{proof}

\section{Proofs from \refsect{enf} (Lassen's eager normal form simulation)}
In this sections, we give the proof of \refprop{enf-validation-of-equivalences} concerning shuffling equivalences, detailing the $\equivsone$ case. The case for $\equivsthree$, the restricted left-to-right version of $\equivexsthree$, is similar.

\begin{lemma}
	\label{l:equivsone-is-included-in-enf}
	$\tm \equivsone \tmp$ then $\tm \enfbisim \tmp$.
\end{lemma}
\begin{proof}
	We prove that $\relsym = Id \cup \{(\tm,\tmp) \mid \tm \equivsone \tmp \}$ is an \enf bisimulation. First note that $Id \subseteq \openfp{Id}$ and $sym(Id) \subseteq \openfp{sym(Id)}$. We show that $\{(\tm,\tmp) \mid \tm \equivsone \tmp \} \subseteq \relenf$ (and the same reasoning shows that $sym(\relsym)$ is also an enf simulation).
	
	Let $(\tmrone,\tmrtwo)=(((\la\var\tm)\tmtwo)\tmthree,(\la\var\tm\tmthree)\tmtwo) \in\relsym$.
	\begin{itemize}
		\item If $\tmtwo \bswleftdiv$, then $\tmrone$ and $\tmrtwo$ diverge, hence $(\tmrone,\tmrtwo) \in \relenf$ by case (enf 1).
		\item If $\tmtwo \bswleft k \val$, then $\tmrone \tolw^{k+1} \tm\isub\var\val\tmthree$  and $\tmrtwo \tolw^{k+1} (\tm\tmthree)\isub\var\val$. We conclude $(\tmrone,\tmrtwo) \in \relenf$ because they both reduce to $\tm\isub\var\val\tmthree = (\tm\tmthree)\isub\var\val$ -- indeed $\var\not\in\fv\tmthree$) hence have the same normal form ($Id$ part of the relation $\relsym$).
		\item If $\tmtwo \bswleft k \levctxp {\vartwo\val}$, then $\tmrone \bswleft k ((\la\var\tm)\levctxp {\vartwo\val})\tmthree$ and $\tmrtwo \bswleft k (\la\var\tm\tmthree)\levctxp {\vartwo\val}$. By case (enf 4), $(\tmrone,\tmrtwo) \in \relenf$ since $\val \rel \val$ and $((\la\var\tm)\levctxp {\varthree})\tmthree \rel (\la\var\tm\tmthree)\levctxp {\varthree}$ (because $\{(\tm,\tmp) \mid \tm \equivsone \tmp \} \subseteq \relsym$).
	\end{itemize}	
	Hence the result by coinduction.
\end{proof}

%

\gettoappendix{prop:enf-validation-of-equivalences}

\begin{proof} Moggi's equivalences proofs are straightforward, and already included in Lassen's original paper \cite{LassenEnf}. 
	
	We deduce the result for the shuffling equivalences by an easy coinductive argument described in \reflemma{equivsone-is-included-in-enf} for $\equivsone$. For $\equivsthree$, the argument is similar.
	
	Counterexamples for the other equivalences are easy to come up with.
\end{proof}

\section{Proofs from \refsect{vsc} (The Value Substitution Calculus)}
\label{app:app-vsc}
In this Appendix, we give the definition of \cbn and \cbv inscrutable terms, prove that \cbv inscrutable terms coincide with \cbv $\Omega$ terms, to then connect with the diverging characterization of \cbv inscrutable terms due to with \cbv $\Omega$-terms.

\paragraph{Definition of \cbn and \cbv Inscrutable Terms}
\begin{definition}[Testing contexts and \cbn/\cbv (in)scrutability]
\label{def:cbn-scrutability}
Testing context are defined by:
\begin{center}
\textsc{Testing contexts}  \ \ \ $\tctx \grameq \ctxhole \mid \tctx \tm \mid (\la\var\tctx)\tm$
\end{center}
A $\l$-term $\tm$ is \emph{\cbn scrutable} if there is a testing context $\tctx$ and a value $\val$ such that $\tctxp\tm \tob^* \val$, that is, a variable or an abstraction, otherwise $\tm$ is \emph{\cbn inscrutable}.

A $\l$-term $\tm$ is \emph{\cbv scrutable} if there is a testing context $\tctx$ and a value $\val$ such that $\tctxp\tm \tobv^* \val$, otherwise $\tm$ is \emph{\cbv inscrutable}.
\end{definition}

\paragraph{\cbv Inscrutable Terms Coincide with \cbv $\Omega$-terms} It is obtained via the following proposition, based on an auxiliary lemma.

\begin{lemma}
\label{l:Omega-testing-diverging}
Let $\tctx$ be a testing context. Then $\tctxp\Omega$ is $\tovsc$-diverging.
\end{lemma}

\begin{proof}
By induction on $\tctx$. Cases:
\begin{itemize}
\item \emph{Empty}, that is $\tctx = \ctxhole$. Trivial.
\item \emph{Application}, that is $\tctx = \tctxtwo \tm$. It follows by the \ih and the fact that $\ctxhole\tm$ is an evaluation context.
\item \emph{Applied abstraction}, that is $\tctx = (\la\var\tctxtwo) \tm$. Then $\tctxp\Omega = (\la\var\tctxtwop\Omega) \tm \tom \tctxtwop\Omega\esub\var\tm$. By \ih, $\tctxtwop\Omega$ is $\tovsc$-diverging and $\ctxhole\esub\var\tm$ is an evaluation context, thus  $\tctxtwop\Omega\esub\var\tm$, and so $\tctxp\Omega$, is $\tovsc$-diverging.\qedhere
\end{itemize}
\end{proof}

\begin{proposition}[\cbv inscrutable terms = \cbv $\Omega$-terms]
\label{prop:cbv-inscr-equal-cbv-omega}
A term $\tm$ is \cbv inscrutable if and only if it is a \cbv $\Omega$-term.
\end{proposition}

\begin{proof}
\hfill
\begin{itemize}
\item \emph{Direction $\Rightarrow$.}
In \cite{DBLP:journals/pacmpl/AccattoliG22}, it is proved that all \cbv inscrutable terms are \cbv contextually equivalent (therein it is Proposition 6.3 and Corollary 6.4, page 17). 

\item \emph{Direction $\Leftarrow$.} In \cite{DBLP:journals/pacmpl/AccattoliG22}, it is proved that \cbv scrutability in Plotkin's calculus and in the VSC coincide (therein it is Theorem 5.5, page 14), that is one can replace $\tobv$ with $\tovsc$ in \refdef{cbn-scrutability} without changing the definition. Moreover, Plotkin's calculus and the VSC have also the same contextual equivalence. Thus, we here use the VSC.

Suppose by contradiction that $\tm$ is not \cbv inscrutable, that is, that there is a testing context $\tctx$ such that $\tctxp\tm \tovsc^*\val$. By \reflemma{Omega-testing-diverging}, $\tctxp\Omega$ is $\tovsc$-diverging. Since $\tm$ and $\Omega$ are contextually equivalent (and contextual equivalence is defined using closing contexts), $\tctxp\tm$ has to be an open term. Let $\fv{\tctxp\tm} =\set{\var_1, \ldots, \var_n}$ and consider the closed testing context $\tctxtwo \defeq \la{\var_n}(\ldots(\la{\var_1}\tctx)\Id\ldots)\Id$. By stability of $\tovsc$ reduction under substitution of values (\refprop{substitutivity_vsce}), we have:
\begin{center}$
\tctxtwop\tm = \la{\var_n}(\ldots(\la{\var_1}\tctxp\tm)\Id\ldots)\Id \tovsc^* \tctxp\tm\isub{\var_1}\Id\ldots \isub{\var_n}\Id \tovsc^* \val\isub{\var_1}\Id\ldots \isub{\var_n}\Id
$\end{center}
which is a closed value, while $\tctxtwop\Omega$ diverges, again by \reflemma{Omega-testing-diverging}. Thus, $\tm$ and $\Omega$ are not contextually equivalent, absurd.\qedhere
\end{itemize}
\end{proof}

\paragraph{Diverging Characterization of $\Omega$-Terms} We first recall the diverging characterization of \cbv inscrutable terms due to \cite{accattoli+paolini-vsc}.

\begin{theorem}[VSC diverging characterization of \cbv inscrutability, \cite{accattoli+paolini-vsc}]
	\label{thm:cbv-scrutability-characterization-bis}
A term $\tm$ is \cbv inscrutable if and only if $\tm$ is $\tovsc$ diverging.
\end{theorem}

\gettoappendix{thm:cbv-scrutability-characterization}
\begin{proof}
It follows from \refthm{cbv-scrutability-characterization-bis} and \refprop{cbv-inscr-equal-cbv-omega}.
\end{proof}

\section{Proofs from \refsect{benchmarks-vsc} (Equational Benchmarks and the Value Substitution Calculus)}
In this section, we recall the proof that structural equivalence strongly commutes with $\tovsc$. The proof is identical to the one given in \cite{accattoli+paolini-vsc}, it is here presented for the interested reader to be able to see which fragments of structural equivalence independently strongly commutes.

\gettoappendix{prop:strong-bisimulation}
\newcommand{\eqz}{\equiv_0}
\begin{proof}
Define $\eqz$ as the context closure of $\equivsone\cup \equivexsthree\cup \equivass \cup \equivcom$.
 We have $\streq=\eqz^*$. We prove that:
\begin{equation}
\mbox{if $t_0\eqz t_1 \Rew{a} s_1$ then there exists $w$ such that $t_0\Rew{a} w \streq s_1$}
\label{eq:bis}
\end{equation}
The statement then follows by induction on the reflexive and transitive closure of $\eqz$. Let us show that: the reflexive case is trivial and if $t_0\eqz t_0'\eqz^k t_1\Rew{a} s_1$ then by \ih\ exists $w$ such that  $t_0'\Rew{a} w \streq s_1$ and by (\ref{eq:bis}) there exists $w'$ such that  $t_0\Rew{a} w'\streq w\streq s_1$.\\

The proof of (\ref{eq:bis}) is by induction on $\eqz$. Actually, before to proceed with the proof one should first prove the following two easy substitutivity properties:
\begin{enumerate}
  \item \label{l:eqo-stability-two} If $t \eqz t'$ then   $t\isub{x}{u} \eqz t'\isub{x}{u}$.
  \item \label{l:eqo-stability-one} If $u \eqz u'$ then   $t\isub{x}{u} \streq t\isub{x}{u'}$.
  \end{enumerate}      
Used in the inductive cases for the ES. We omit their proofs, which are straightforward inductions.

\begin{itemize}
\item Base cases:
\begin{itemize}
 \item \emph{Commutativity}: let $t_0= t\esub\vartwo{u}\esub\var{s}\equivcom t\esub\var{s}\esub\vartwo{u}=t_1$ with $x\notin\fv{u}$ and $y\notin\fv{s}$. If $t_1\Rew{a} s_1$ because:
 \begin{itemize}
  \item $t\Rew{a}  t'$ then $t_0=t\esub\vartwo{u}\esub\var{s}\Rew{a}  t'\esub\vartwo{u}\esub\var{s}\equivcom t'\esub\var{s}\esub\vartwo{u}=s_1$.
  \item $u\Rew{a}  u'$ or $s\Rew{a}  s'$ then it is similar to the previous case.
  \item $s=\sctxp\val$ and $t\esub\var{\sctxp\val}\esub\vartwo{u}\toe \sctxp{t\isub\var\val}\esub\vartwo{u}=s_1$. Then:
  \[\begin{array}{llllll}
  t_0&\toe & \sctxp{t\esub\vartwo{u}\isub\var\val}\\
  &=& \sctxp{t\isub\var\val \esub\vartwo{u}} \\
    &\equivcom& \sctxp{t\isub\var\val} \esub\vartwo{u}&=&s_1
  \end{array}\]
  \item The case where $u=\sctxp\val$ and $t\esub\var{s}\esub\vartwo{\sctxp\val}\toe \sctxp{t\esub\var{s}\isub\vartwo\val}=s_1$is similar to the previous one.  \end{itemize}

\item \emph{Sigma 1}: let $t_0=t\esub\var{s} \tmtwo   \equivsone  (\tm\tmtwo)\esub\var{s}=t_1$ with $x\notin\fv{u}$. If $t_1\Rew{a} s_1$ because:

 \begin{itemize}
  \item $t\Rew{a}  t'$ then $t_0=t\esub\var{s} \tmtwo \Rew{a}  t'\esub\var{s} \tmtwo \equivsone (t' \tmtwo )\esub\var{s}=s_1$.
  \item $s\Rew{a}  s'$ or $u\Rew{a}  u'$ then it is similar to the previous case.
  \item $s=\sctxp\val$ and $t_1=(\tm\tmtwo)\esub\var{\sctxp\val}\toe \sctxp{(\tm\tmtwo)\isub\var\val}=s_1$.
Then:
  \[\begin{array}{llllll}
  t_0&=& t\esub\var{\sctxp\val} \tmtwo \\
  &\toe& \sctxp{t\isub\var\val} \tmtwo \\
    &\equivsone& \sctxp{t\isub\var\val  \tmtwo }\\
    &=& \sctxp{(\tm\tmtwo)\isub\var\val}&=&s_1
  \end{array}\]

  \item $t= \l y. t'$ and $t_1=((\l y. t') \tmtwo )\esub\var{s} \tom t'\esub\vartwo{u}\esub\var{s}$. Then:
  \[\begin{array}{llllll}
  t_0&=& (\l y. t')\esub\var{s} \tmtwo \\
  &\tom& t'\esub\vartwo{u}\esub\var{s}&=&s_1
  \end{array}\]
  \end{itemize}

Note that here it is reflexivity of $\streq$ which is used.

 \item The case symmetric to the previous one, \ie\ $t_0= (\tm\tmtwo)\esub\var{s}   \equivsone  t\esub\var{s} \tmtwo=t_1$ with $x\notin\fv{u}$, is proved analogously. It shall be so for all following cases, so we simply omit the symmetric cases.

\item \emph{Extended sigma 3}: let $t_0=\tm\tmtwo\esub\var{s}  \equivexsthree  (\tm\tmtwo)\esub\var{s}=t_1$ with $x\notin\fv{t}$. If $t_1\Rew{a} s_1$ because:

 \begin{itemize}
  \item $t\Rew{a}  t'$ then $t_0=\tm\tmtwo\esub\var{s}\Rew{a}  t' \tmtwo \esub\var{s} \equivexsthree (t' \tmtwo )\esub\var{s}=s_1$.
  \item $s\Rew{a}  s'$ or $u\Rew{a}  u'$ then it is similar to the previous case.
  \item $s=\sctxp\val$ and $t_1=(\tm\tmtwo)\esub\var{\sctxp\val}\toe \sctxp{(\tm\tmtwo)\isub\var\val}=s_1$. Then:
  \[\begin{array}{llllll}
  t_0&=& \tm\tmtwo\esub\var{\sctxp\val}\\
  &\toe & \tm\, \sctxp{\tmtwo\isub\var\val}\\
    &\equivexsthree & \sctxp{\tm\tmtwo\isub\var\val}\\
    &=& \sctxp{(\tm\tmtwo)\isub\var\val}&=&s_1
  \end{array}\]

  \item $t= \l y. t'$ and $t_1=((\l y. t') \tmtwo )\esub\var{s} \tom t'\esub\vartwo{u}\esub\var{s}$. Then:
  \[\begin{array}{llllll}
  t_0&=& (\l y. t') \tmtwo \esub\var{s}\\
  &\tom& t'\esub\vartwo{\tmtwo\esub\var{s}}\\
  &\equivass& t'\esub\vartwo{u}\esub\var{s}&=&s_1
  \end{array}\]
  \end{itemize}

\item \emph{Associativity of ES}: let $t_0=t\esub\vartwo{\tmtwo\esub\var{s}}  \equivass  t\esub\vartwo{u}\esub\var{s}=t_1$ with $x\notin\fv{t}$. If $t_1\Rew{a} s_1$ because:

\begin{itemize}
\item $t\Rew{a} t'$ then $t_0\Rew{a} t'\esub\vartwo{\tmtwo\esub\var{s}}\equivass t'\esub\var\tmtwo\esub\var{s}=s_1$.
\item $u\Rew{a} u'$ or $s\Rew{a} s'$ it is analogous to the previous case.
\item $s=\sctxp\val$ and $t_1 \toe \sctxp{t\esub\vartwo{u}\isub\var\val}=s_1$. Then
  \[\begin{array}{llllll}
  t_0&=& t\esub\vartwo{u\esub\var{\sctxp\val}}\\
  &\toe& t\esub\vartwo{\sctxp{u\isub\var\val}}\\
  &\equivass& \sctxp{t\esub\vartwo{u\isub\var\val}}\\
  &=& \sctxp{t\esub\vartwo{u}\isub\var\val}&=&s_1
  \end{array}\]

\item $u=\sctxp\val'$ and $t_1=\sctxp{t\esub\var{\sctxtwop\val}} \toe \sctxp{\sctxtwop{t\isub\var\val}}$. Then $t_0=t\esub\var{\sctxp{\sctxtwop\val}}\toe \sctxp{\sctxtwop{t\isub\var{\val}}}=s_1$. Note that here it is reflexivity of $\streq$ which is used.
\end{itemize}
 \end{itemize}

\item Inductive cases. We only show the interesting ones:
\begin{itemize}
\item Application: the only case where the reduction interact with the contextual closure is $t_0=\sctxp{\la\var\tm} \tmtwo  \eqz \sctxp{\la\var\tm'} \tmtwo  = t_1 \Rew{a} \sctxp{t'\esub\var\tmtwo}=s_1$. Then $t_0\Rew{a} \sctxp{t\esub\var\tmtwo} \eqz \sctxp{t'\esub\var\tmtwo}=s_1$. The variants  $t_0=\sctxp{\la\var\tm} \tmtwo  \eqz \sctxp{\la\var\tm} \tmtwo ' = t_1 \Rew{a} \sctxp{t\esub\var{\tmtwo'}}=s_1$ and $t_0=\sctxp{\la\var\tm} \tmtwo  \eqz \sctxtwop{\la\var\tm} \tmtwo  = t_1 \Rew{a} \sctxtwop{t\esub\var\tmtwo}=s_1$ are analogous. All other inductive cases for application are straightforward.

\item Explicit substitution. We only show the interesting cases. 
\begin{itemize}
\item $t_0=t\esub\var{\sctxp\val} \eqz t'\esub\var{\sctxp\val} = t_1 \Rew{a} \sctxp{t'\isub\var\val}=s_1$. Then by the first substitutivity property we obtain $t_0\Rew{a} \sctxp{t\isub\var\val} \eqz \sctxp{t'\isub\var\val}$.
\item $t_0=t\esub\var{\sctxp\val} \eqz t\esub\var{\sctxp\valtwo} = t_1 \Rew{a} \sctxp{t\isub\var\valtwo}=s_1$. Then by the second substitutivity property we obtain $t_0\Rew{a} \sctxp{t\isub\var\val}\streq \sctxp{t\isub\var\valtwo}$.\qedhere
\end{itemize}
\end{itemize}
\end{itemize}
\end{proof}

\section{Proof of Compatibility for \nafex and \net Similarity}
\label{chapter:proof-compatibility-nafex}
The proof follows the same structure as in the case of naive simulations. We prove the general statement for $\equivx$ a mirror.

\subsection{Proof of Equivalence of Small-Step and Big-Step Operational Semantics}
In this subsection, we give the details for the proof of completeness for the big step system we introduce for the Value Substitution Calculus using the diamond property.

\subsubsection{Preliminaries}
Some generalities about subreductions and their properties with a calculus that has the diamond property (or the Random Descent property).

\paragraph{General Properties of Random Descent}
The diamond property implies the Random Descent (RD) property, that is the following:
\begin{definition}[RD property, Newman]A relation $\to$ has the 
	\emph{Random Descent (RD)} property if for each element $\tm$, all maximal sequences from $\tm$ have the same length and---if it is finite---they all end in the same element.
\end{definition}

\begin{lemma}[Completeness of subreductions]\label{lem:RD_completeness} Let $\to$ be a reduction which satisfies the RD property, and  $\tos \subseteq \to$.
	If  $\tos$ and  $\to$ have the same normal forms, and $\tmn$ is $\to$-normal, then: 
	\[\tm \to^k \tmn  \iff \tm \tos^k \tmn \]
	
		%
		%
	
\end{lemma}
\begin{proof}By induction on $k$.
	\begin{itemize}
		\item $k=0$. Trivial.
		\item $k\geq 1$. Assume $\tm \to \tm' \to^{k-1} \tmn $. By   assumption, $\tm$ is not a $\tos$-$ \nf $. Hence it exists $\tm''$ such that $\tm \tos \tm''$.
		Since $\tos \subseteq \to$, then $\tm \to \tm''$ and so by RD property $\tm'' \to^{k-1} \tmn$. By \ih,  $\tm'' \tos^{k-1},   \tmn$, hence $ \tm \tos^k \tmn  $.\qedhere
	\end{itemize}
	
\end{proof}

\subsubsection{Big-Steps/Small-Steps for the Value Substitution Calculus }
\label{proof:proof-completeness-big-step-vsc-practical}
\paragraph{A constrained reduction to model the big steps semantics}
In the Value Substitution Calculus,we define a subreduction $\tos \subseteq \tovsc$ is defined as follows. 
\begin{center}
	$\begin{array}{r@{\hspace{.5cm}}rlll}
		%
		\textsc{Substitution-Restricted Evaluation Contexts} & \ievctx & \grameq &  \ctxhole\mid \tm\ievctx\mid \ievctx\tm \mid \ievctx\esub{\var}{\itm} \mid \tm\esub{\var}{\ievctx}\\
	\end{array}
	$\end{center}

\begin{center}
	$\begin{array}{c@{\hspace{.5cm}}rcc}
		\textsc{Rule at Top Level} & \textsc{Contextual closure} \\
		\isctxp{\la\var\tm}\tmtwo \rtos \isctxp{\tm\esub\var\tmtwo} &
		\ievctxp \tm \tos \ievctxp \tmtwo \textrm{~~~if } \tm \rtos \tmtwo
		\\
		\tm\esub\var{\isctxp{\val}} \rtos \isctxp{\tm\isub\var\val} &
		\ievctxp \tm \tos \ievctxp \tmtwo \textrm{~~~if } \tm \rtos \tmtwo
	\end{array}
	$\end{center}

\begin{remark}
	This subreduction is written with the big step system in mind to ease the proof of \reflemma{ss-bs-equivalenceaux1_vsc} and still be a complete subreduction.
\end{remark}

\begin{lemma}
	\label{l:sctx-normal-is-isctx}
	If $\sctxp\tm$ is normal then $\sctx = \isctx$.
\end{lemma}
\begin{proof}
	By induction on $\sctx$.
	\begin{itemize}
		\item $\sctx= \ctxhole$, the result is immediate.
		\item $\sctx = \sctxONE\esub\var\tmtwo$. $\sctxONEp\tm$ is also normal so by induction $\sctxONE = \isctxONE$. If $\tmtwo$ is not an inert, either $\tmtwo$ reduces or $\tmtwo$ is a value and the whole term reduces which contradicts the hypothesis that $\sctxp\tm$ is normal. Hence $\tmtwo = \itm$ and so $\sctx = \isctxONE\esub\var\itm$ \ie $\sctx = \isctx$.\qedhere
	\end{itemize}
\end{proof}

\begin{lemma}
	\label{l:tos-normal-is-tovsct-normal}
	If $\tm$ is $\tos$-normal, then $\tm$ is $\tovsc$-normal.
\end{lemma}
\begin{proof}
	By induction on the structure of $\tm$.
	\begin{itemize}
		\item $\tm =\var$ or $\tm = \la\var\tmp$, the result is immediate
		\item $\tm = \tmrone\tmrtwo$, then $\tmrone$ and $\tmrtwo$ are $\tos$-normal, and by induction are $\tovsc$-normal. There is only one possibility for $\tm$ to $\tovsc$-reduce.
		
		Suppose $\tmrone\tmrtwo \tom \tmp$ then $\tmrone =\sctxp{\la\var\tmronep}$.
		Since $\tmrone$ is a $\tovsc$ normal form, by \reflemma{sctx-normal-is-isctx} we have that $\tmrone = \isctxp{\la\var\tmronep'}$, which contradicts the assumption that $\tm$ is a $\tos$-normal form.
		\item $\tm = \tmrone\esub\var\tmrtwo$, , then $\tmrone$ and $\tmrtwo$ are $\tos$-normal, and by induction are $\tovsc$-normal. There is only one possibility for $\tm$ to $\tovsc$-reduce.
		
		Suppose $\tmrone\esub\var\tmrtwo \toeabs \tmp$ then $\tmrtwo =\sctxp{\val}$.
		Since $\tmrtwo$ is a $\tovsc$ normal form, by \reflemma{sctx-normal-is-isctx} we have $\tmrtwo = \isctxp{\valtwo}$, which contradicts the assumption that $\tm$ is a $\tos$-normal form.\qedhere
	\end{itemize}
	
\end{proof}

\begin{corollary}
	\label{cor:tovsct-and-tos-normal-forms}
	$\tovsc$ and $\tos$ have the same normal forms.
\end{corollary}

\begin{corollary}[Completeness of $\tos$] \label{prop:S_completeness}$ \tm \tovsc^k \tmn  \iff \tm \tos^k \tmn $ 
\end{corollary}
\begin{proof} 
Direct corollary of \Cref{lem:RD_completeness} and Corollary \ref{cor:tovsct-and-tos-normal-forms}.
	\end{proof}
	
	\paragraph{Big Steps/Small Steps via $\tos$}
	
	\begin{lemma}
		\label{l:ss-bs-equivalenceaux1_vsc}
		$\tm \tos \tmp$ and  $\tmp \bsvsct k \ntm \Rightarrow \tm \bsvsct {k+1} \ntm$
	\end{lemma}
	
	\begin{proof}
		Straightforward proof by structural induction. Note that substitutivity is never  required.
	\end{proof}

\gettoappendix{l:ss-bs-equivalence_vsce}

\begin{proof}
	$(\Rightarrow)$ by induction on the $(\tm \bsvsct k \ntm)$ derivation.
	
	$(\Leftarrow)$ by induction on $k$ using the fact that $\tm \tovsc^k \tmn$ implies $\tm \tos^k \tmn$  (by \Cref{prop:S_completeness}), and the \reflemma{ss-bs-equivalenceaux1_vsc}.
\end{proof}

\subsection{Lemmas concerning $\equivx$}

\begin{proposition}[$\opnafex\cdot\equivx\subseteq\opnafex$]
	\label{prop:relnafex-equivx-subseteq-relnafex}
	Let $\relsym$ be a relation on terms.
	If $\tm \relnafex \tmtwo$ and there exists $\tmtwop$ such that $\tmtwo \equivx \tmtwop$, then $\tm \relnafex \tmtwop$.
\end{proposition}

\begin{proof}
	By case analysis on $\tm \relnafex \tmtwo$.
\end{proof}

\begin{proposition}[$\equivx$ is a strong commutation wrto $\tovsc$]
	\label{prop:equivx-is-a-strong-bisimulation}
	If $\tm\equivx\tmtwo$ and $\tm \tovsc \tmp$ then $\tmtwo \tovsc \tmtwop$ and $\tmp\equivx\tmtwop$.
\end{proposition}

\begin{proposition}[$\equivx$ preserves normal forms]
	\label{prop:equivx-preserves-normal-forms}
	$\forall \tm,\tmtwo,~ \tm\equivx\tmtwo$ and $\tm$ is a normal form implies $\tmtwo$ is a normal form.
\end{proposition}

\begin{proof}
	This is derived from the fact that $\equivx$ is a strong commutation wrto $\tovsc$.
\end{proof}

\begin{proposition}[$\equivx$ is substitutive]
	\label{prop:equivx-is-substitutive}
	$\forall \tm,\tmtwo,\val,~ \tm\equivx\tmtwo$ implies $\tm\isub\var\val\equivx\tmtwo\isub\var\val$
\end{proposition}
\subsection{Compatibility proof}

A useful tool in the proof is substitutivity, with respect to small-step and big-step semantics, that is Proposition \ref{prop:substitutivity_vsce}.

\begin{proposition}[Substitutivity of $\tovsc$]
	\label{l:stability_vsce}
	$\tm\tovsc\tmp \Rightarrow \tm\isubst\val\var \tovsc \tmp\isubst\val\var$
\end{proposition}

\begin{proof}
	By induction on $\tm\tovsc\tmp$ (induction on contexts), using the fact that a value where a variable is substituted by a value is still a value.
\end{proof}

\begin{lemma}[Substitutivity of $\bsvscts$]
	\label{l:splitting_vsce}
	Forall $\tm,\val$,
	
	$\tm\isubst\val\var \bsvsct k \ntm \implies 
	\exists k',\ntmtwo$ s.t. $ \tm \bsvsct {k'} \ntmtwo$ and $\ntmtwo\isubst\val\var\bsvsct {k-k'} \ntm$

\end{lemma}

\begin{proof}
	Suppose $\tm\isubst\val\var \bsvsct k \ntm$, then $ \tm \bsvscts \ntmtwo$ because if it diverges then the first diverges as well by the stability of reduction by substitution {(\reflemma{stability_vsce})}.
	Let $k', \tm \bsvsct {k'} \ntmtwo$
	then $ \tm\isubst\val\var \tovsc^{k'}\equivsone \ntmtwo\isubst\val\var$ hence $\ntmtwo\isubst\val\var\bsvsct {k-k'} \ntm$ because the reduction $\tovsc$ is diamond (hence all reduction sequences have the same length).
\end{proof}

\subsection{Equivalence of $\mlasrelnafex$ and $\mlasrel$ on normal forms.}
As in the proof of compatibility for naive similarity, we need to show the main proof first on normal forms, then we will generalize to any term.

\begin{lemma}
	\label{l:lasrelnafex-normal-forms-lasrel-left-to-right}
	If $\fire\mlasrelnafex\firetwo$ then $\fire\mlasrel\firetwo$.
\end{lemma}
\begin{proof}
	By case analysis on $\fire =\valt \mid \var\fire\mid\itmapp\fire \mid \fire\esub\var\itm$ and using the ($\scequivx$) rule.
\end{proof}

\begin{lemma}[Constrained Substitutivity of $\mlasrelnafex$ on normal forms]
	\label{l:lasrelnafex-normal-forms-substitutive}
	If $\ntmONE \mlasrelnafex \ntmTWO$, $\valof\tmrthree \mlasrelnafex \valof\tmrfour$ and $\ntmONE\isub\var{\valof\tmrthree}$ and $\ntmTWO\isub\var{\valof\tmrfour}$ are $\tovsce$-normal then $\ntmONE\isub\var{\valof\tmrthree} \mlasrelnafex \ntmTWO\isub\var{\valof\tmrfour}$.
\end{lemma}

\begin{proof}
	By case analysis on $\ntmONE$. Cases:
	\begin{itemize}
		\item $\ntmONE = \var$ and $\ntmTWO = \var$ then $\ntmONE\isub\var{\valof\tmrthree} = \valof\tmrthree \mlasrelnafex \valof\tmrfour = \ntmTWO\isub\var{\valof\tmrfour}$.
		
		\item $\ntmONE = \vartwo$ and $\ntmTWO = \vartwo$ then $\ntmONE\isub\var{\valof\tmrthree} =  \vartwo \mlasrelnafex \vartwo = \ntmTWO\isub\var{\valof\tmrfour}$.
		
		\item $\ntmONE = \la\vartwo\tm$ and $\ntmTWO = \la\vartwo\tmp$ with $\tm \mlasrel \tmp$
		we have \[\infer{\tm\isub\var{\valof\tmrthree} \mlasrel \tmp\isub\var{\valof\tmrfour}}{\tm \mlasrel \tmp & \valof\tmthree \mlasrel \valof\tmfour}\]
		hence by case (\nafex 3) $\ntmONE\isub\var{\valof\tmrthree} = \la\vartwo{\tm\isub\var{\valof\tmrthree}} \mlasrelnafex  \la\vartwo{\tmp\isub\var{\valof\tmrfour}} = \ntmTWO\isub\var{\valof\tmrfour}$.

		\item $\ntmONE = \ntmONEtwo\ntmONEthree$ and $\ntmTWO \equivx \ntmTWOtwo\ntmTWOthree$ with $\ntmONEtwo \mlasrel \ntmTWOtwo$ and $\ntmONEthree \mlasrel \ntmTWOthree$. 
		From $\ntmTWO \equivx \ntmTWOtwo\ntmTWOthree$, we deduct by substitutivity of $\equivx$ (\refprop{equivx-is-substitutive}) that $\ntmTWOtwo\isub\var{\valof\tmrfour}\ntmTWOthree\isub\var{\valof\tmrfour} \equivx \ntmTWO\isub\var{\valof\tmrfour}$.
		
		Since $\ntmONE\isub\var{\valof\tmrthree}$ and $\ntmTWO\isub\var{\valof\tmrfour}$ are $\tovsce$-normal, then by \refprop{equivx-preserves-normal-forms}, $\ntmONEtwo\isub\var{\valof\tmrthree}$, $\ntmTWOtwo\isub\var{\valof\tmrfour}$, $\ntmONEthree\isub\var{\valof\tmrthree}$ and $\ntmTWOthree\isub\var{\valof\tmrfour}$ all are $\tovsce$-normal as well and $\ntmONEtwo\isub\var{\valof\tmrthree}$, $\ntmTWOtwo\isub\var{\valof\tmrfour}$ are not almost-abstractions (\ie  $\not =\isctxp{\la\vartwo\tm}$ for any $\isctx$).
		
		To conclude that $\ntmONE\isub\var\val \mlasrelnafex \ntmTWO\isub\var\val$, what is only remaining is that\\ $\ntmONEtwo\isub\var\val \mlasrel \ntmTWOtwo\isub\var\valtwo$ and  $\ntmONEthree\isub\var\val \mlasrel \ntmTWOthree\isub\var\valtwo$.
		
		We derive easily these facts: ($\val\mlasrelnafex\valtwo$ implies $\val\mlasrel\valtwo$ by \reflemma{lasrelnafex-normal-forms-lasrel-left-to-right})
		\[ \infer[\scsub]{\ntmONEtwo\isub\var\val \mlasrel \ntmTWOtwo\isub\var\valtwo}{\ntmONEtwo \mlasrel \ntmTWOtwo & \val \mlasrel \valtwo} ~\text{and}~ \infer[\scsub]{\ntmONEthree\isub\var\val \mlasrel \ntmTWOthree\isub\var\valtwo}{\ntmONEthree \mlasrel \ntmTWOthree & \val \mlasrel \valtwo}\]

		\item $\ntmONE = \ntmONEtwo\esub\vartwo\itmONEtwo$ and $\ntmTWO {\equivx} \ntmTWOtwo\esub\vartwo\itmTWOtwo$ with $\itmONEtwo \mlasrel \itmTWOtwo$ and $\ntmONEtwo \mlasrel \ntmTWOtwo$. 
		
		The hypothesis that $\ntmONE\isub\var\val$ is normal is equivalent to $\itmONEtwo\isub\var\val$ is an inert and $\ntmONEtwo\isub\var\val$ is normal.
		
		From $\ntmTWO \equivx \ntmTWOtwo\esub\vartwo\itmTWOtwo$, we deduct {by substitutivity of $\equivx$} (\refprop{equivx-is-substitutive}) that $\ntmTWO\isub\var\valtwo \equivx (\ntmTWOtwo\isub\var\valtwo)\esub\vartwo{\itmTWOtwo\isub\var\valtwo}$.
		Since {$\equivx$ preserves normal forms}  (\refprop{equivx-preserves-normal-forms}), the hypothesis that $\ntmTWO\isub\var\valtwo$ is normal is equivalent to $\itmTWOtwo\isub\var\valtwo$ is an inert and $\ntmTWOtwo\isub\var\valtwo$ is normal.

		To conclude that $\ntmONE\isub\var\val \mlasrelnafex \ntmTWO\isub\var\val$, what is only remaining is that\\ $\itmONEtwo\isub\var\val \mlasrel \itmTWOtwo\isub\var\valtwo$ and  $\ntmONEtwo\isub\var\val \mlasrel \ntmTWOtwo\isub\var\valtwo$.
		
		We derive easily these facts: ($\val\mlasrelnafex\valtwo$ implies $\val\mlasrel\valtwo$ by \reflemma{lasrelnafex-normal-forms-lasrel-left-to-right})
		\[ \infer[\scsub]{\itmONEtwo\isub\var\val \mlasrel \itmTWOtwo\isub\var\valtwo}{\itmONEtwo \mlasrel \itmTWOtwo & \val \mlasrel \valtwo} ~\text{and}~ \infer[\scsub]{\ntmONEtwo\isub\var\val \mlasrel \ntmTWOtwo\isub\var\valtwo}{\ntmONEtwo \mlasrel \ntmTWOtwo & \val \mlasrel \valtwo}\]
		
	\end{itemize}
\end{proof}

\begin{lemma}
	\label{l:lasrelnafex-normal-forms-lasrel-right-to-left}
	If $\relsym$ is a \nafex simulation.
	If $\ntm\mlasrel\ntmtwo$ then $\ntm\mlasrelnafex\ntmtwo$.
\end{lemma}

\begin{proof}
	By induction on the derivation $\ntm \mlasrel \ntmtwo$. Cases of the last rule in the derivation of $\ntm\mlasrel\ntmtwo$:
	\begin{itemize}
		\item \emph {$\scvar$}\[ \infer[\scvar]{\var \mlasrel \var}{} \]
		then $\var \mlasrelnafex \var$ by case (\nafex 2).
		\item \emph {$\scabs$} \[ \infer[\scabs]{\ntm = \la\var\tm \mlasrel \la\var\tmp = \ntmtwo}{\tm \mlasrel \tmp} \]
		then $\ntm \mlasrelnafex \ntmtwo$ by case (\nafex 3) with $\tm \mlasrel \tmp$.
		\item \emph {$\sclift$} \[ \infer[\scabs]{\ntm \mlasrel \ntmtwo}{\ntm \rel \ntmtwo} \]
		then $\ntm \relnafex \ntmtwo$ since $\relsym$ is a \nafex simulation.
		By monotonicity of $\opnafex$, $\ntm \mlasrelnafex \ntmtwo$.
		
		\item \emph {$\scapp$} 
		\[ \infer[\scapp]{\ntm = \ntmONE\ntmTWO \mlasrel \ntmONEtwo\ntmTWOtwo = \ntmtwo}{\ntmONE \mlasrel \ntmONEtwo & \ntmTWO \mlasrel \ntmTWOtwo} \]

		then $\ntm \mlasrelnafex \ntmtwo$ by case (\nafex 4) with $\ntmONE \mlasrel \ntmONEtwo$ and $\ntmTWO \mlasrel \ntmTWOtwo$.

		\item \emph {$\scesub$} \[ \infer[\scesub]{\ntm = \ntmONE\esub\var\itmONE \mlasrel \ntmTWO\esub\var\itmTWO = \ntmtwo}{\ntmONE \mlasrel \ntmTWO & \itmONE \mlasrel \itmTWO} \] then $\ntm \mlasrelnafex \ntmtwo$
		by case (\nafex 5) with $\itmONE \mlasrel \itmTWO$ and $\ntmONE \mlasrel \ntmTWO$.
		\item \emph {$\scsub$} \[ \infer[\scsub]{\ntm = \ntmONE\isub\var\val \mlasrel \ntmTWO\isub\var\valtwo = \ntmtwo}{\ntmONE \mlasrel \ntmTWO & \val \mlasrel \valtwo} \]
		by \ih we have $\ntmONE \mlasrelnafex \ntmTWO$ and $\valof\tmrthree \mlasrelnafex \valof\tmrfour$. By \reflemma{lasrelnafex-normal-forms-substitutive}, $\ntmONE\isub\var{\valof\tmrthree} \mlasrelnafex \ntmTWO\isub\var{\valof\tmrfour}$.
		
		\item \emph {$\scequivx$} \[\infer[\scequivx]{\ntm \mlasrel \ntmtwo}{\ntm \mlasrel \ntmONE & \ntmONE \equivx \ntmtwo}\]
		
		By \ih, $\ntm \mlasrelnafex \ntmONE$ which means by \refprop{relnafex-equivx-subseteq-relnafex} $\ntm \mlasrelnafex \ntmtwo$ since $\ntmtwo \equivx \ntmONE$.\qedhere
	\end{itemize}
\end{proof}

\subsection{Syntactic lemmas to relate bisimilar normal forms}
The lemmas in this subsection are needed because mirrored simulations are defined \emph{top-down}, whereas one needs at some point to be able to reason bottom-up, especially in the presence of explicit substitutions and reduction \emph{at a distance}.

\begin{lemma}
	\label{l:lasrelnafex-inert-style-lists-induction}
	If $\isctxp\var \mlasrelnafex \isctxtwop\var$ (or $\isctxp\val \mlasrelnafex \isctxtwop\valtwo$) then either $\isctx,\isctxtwo =\ctxhole,\ctxhole$ or  (for $\tm = \var$ or $ \val$ and $\tmp=\var$ or $\val$) $\isctxp\tm = \isctxONEp\tm\esub\vartwo\itm$ and $\isctxtwop\tmp \equivx\isctxONEtwop\tmp\esub\vartwo\itmtwo$ with $\isctxONEp\var \mlasrelnafex \isctxONEtwop\var$ (or $\isctxONEp\val \mlasrelnafex \isctxONEtwop\valtwo$) and $\itm \mlasrelnafex \itmtwo$.
\end{lemma}

\begin{proof}
	By contradiction and case exhaustion, these are the only possibilities for\\ $\isctxp\var \mlasrelnafex \isctxtwop\var$ (or $\isctxp\val \mlasrelnafex \isctxtwop\valtwo$).
	\begin{itemize}
		\item If only one of the lists is empty: $\var \mlasrelnafex \itm\esub\vartwo\itmtwo$ (or $\var \mlasrelnafex \ntm\esub\vartwo\itmtwo$) is not possible given the definition of nafex. 
		\item If $\isctx= \isctxONE\esub\vartwo\itm$ and $\isctxtwo \equivx\isctxONEtwo\esub\varthree\itmtwo$, we again have $\neg (\ntm\esub\vartwo\itm \mlasrelnafex \ntmtwo\esub\varthree\itmtwo)$.\qedhere
	\end{itemize}
\end{proof}

%
%

This \reflemma{lasrelnafex-inert-style-lists-induction} gives us an "induction principle" on $\isctx,\isctxtwo$ when $\isctxp\var \mlasrelnafex \isctxtwop\var$ (or $\isctxp\val \mlasrelnafex \isctxtwop\valtwo$.

\begin{lemma}
	\label{l:lasrelnafex-values-isctx-decomposition}
	$\isctxp\val \mlasrelnafex \isctxtwop\valtwo \iff \val \mlasrelnafex \valtwo$ and $\isctxp\var \mlasrelnafex \isctxtwop\var$ ($\var$ is fresh)
\end{lemma}

\begin{proof}
	By "induction" on lists $\isctx, \isctxtwo$. (They are of the same size because of how \nafex is defined - see \reflemma{lasrelnafex-inert-style-lists-induction})
	\begin{itemize}
		\item $\isctx,\isctxtwo =\ctxhole,\ctxhole$ then $\val \mlasrelnafex \valtwo \iff \val \mlasrelnafex \valtwo$ and $\var\mlasrelnafex\var$ is always true by case (\nafex 2).
		\item $\isctxp\val = \isctxONEp\val\esub\vartwo\itm$ and $\isctxtwop\valtwo \equivx\isctxONEtwop\valtwo\esub\vartwo\itmtwo$ with $\isctxONEp\valtwo \mlasrelnafex \isctxONEtwop\valtwo$ and $\itm \mlasrelnafex \itmtwo$, then
		
		
		
		by \ih, $\isctxONEp\valtwo \mlasrelnafex \isctxONEtwop\valtwo\iff \isctxONEp\var \mlasrelnafex \isctxONEtwop\var$ and $\val \mlasrelnafex \valtwo$
		
		by \reflemma{lasrelnafex-normal-forms-lasrel-left-to-right}, $\isctxONEp\valtwo \mlasrelnafex \isctxONEtwop\valtwo\iff \isctxONEp\var \mlasrel \isctxONEtwop\var$, $\val \mlasrelnafex \valtwo$ and $\itm \mlasrel \itmtwo$
		
		and finally by case (nafex 5)\\ $\isctxp\val \mlasrelnafex \isctxtwop\valtwo\iff \isctxp\var=\isctxONEp\var\esub\vartwo\itm \mlasrelnafex \isctxONEtwop\var\esub\vartwo\itmtwo\equivx \isctxtwop\var$, $\val \mlasrelnafex \valtwo$.
	\end{itemize}
\end{proof}

\begin{lemma}
	\label{l:values-fireballs-stable-lasrelnafex}
	If $\fire \mlasrelnafex \firep$ then ($\fire =\isctxp\valt$ $\iff$ $\firep = \isctxtwop\valttwo$)
\end{lemma}

\begin{proof}
	Proof by induction on $\isctx$ using \ref{l:lasrelnafex-normal-forms-lasrel-left-to-right} and \ref{l:lasrelnafex-normal-forms-lasrel-right-to-left}.
\end{proof}

\begin{corollary}
	\label{l:inerts-fireballs-stable-lasrelnafex}
	If $\fire \mlasrelnafex \firep$ then ($\fire =\itm$ $\iff$ $\firep = \itmtwo$)
\end{corollary}

\begin{lemma}
	\label{l:lasrelnafex-normal-forms-isctx-decomposition}
	Let $\ntm$ and $\ntmtwo$ be normal forms and $\isctx$ and $\isctxtwo$ inert substitution contexts which may capture free variables of $\ntm$ and $\ntmtwo$. $\ntm \mlasrelnafex \ntmtwo$ and $ \isctxp\var \mlasrelnafex \isctxtwop\var$ ($\var$ is fresh) $\Rightarrow \isctxp\ntm \mlasrelnafex \isctxtwop\ntmtwo$
\end{lemma}
\begin{proof}
	By induction on $\isctx,\isctxtwo$.
\end{proof}

\begin{lemma}
	\label{l:lasrelnafex-normal-forms-isctx-decomposition-2}
	$\ntm \mlasrelnafex \ntmtwo$ and $ \isctxp\var \mlasrelnafex \isctxtwop\var$ $\Rightarrow \isctxp{\var\ntm} \mlasrelnafex \isctxtwop{\var\ntmtwo}$
\end{lemma}
\begin{proof}
	By induction on $\isctx,\isctxtwo$.
\end{proof}

\begin{lemma}
	\label{l:lasrelnafex-normal-forms-isctx-decomposition-applied-inerts}
	$\ntm \mlasrelnafex \ntmtwo$ and $ \isctxp\itmapp \mlasrelnafex \isctxtwop\itmapptwo$ $\Rightarrow \isctxp{\itmapp\ntm} \mlasrelnafex \isctxtwop{\itmapptwo\ntmtwo}$
\end{lemma}
\begin{proof}
	By induction on $\isctx$.
	
	Since $ \isctxp\itmapp \mlasrelnafex \isctxtwop\itmapptwo$, two cases:
	\begin{itemize}
		\item $\isctx=\ctxhole$ and we fall into case (\nafex 4), that is $\isctxtwop\itmapptwo \equivx \itmapptwo_1$ and $\itmapp \mlasrel \itmapptwo_1$. Hence $\itmapp \ntm \mlasrelnafex \itmapptwo_1\ntmtwo \equivx \isctxtwop{\itmapptwo\ntmtwo}$ follows.
		
		\item  $\isctx=\isctxONE\esub\vartwo\itm$ and we fall into case (\nafex 5), that is $\isctxtwop\itmapptwo \equivx \isctxONEtwop{\itmapptwo_1 }\esub\var\itmtwo$ (we cannot move from inerts to answers and we choose the normal form modulo $\equivsone$) where $\isctxONEp\itmapp \mlasrel \isctxONEtwop{\itmapptwo_1 }$ and $\itm \mlasrel \itmtwo$. Apply \reflemma{lasrelnafex-normal-forms-lasrel-right-to-left} and obtain by \ih that $\isctxONEp{\itmapp\ntm} \mlasrelnafex \isctxONEtwop{\itmapptwo_1\ntmtwo}$. By \reflemma{lasrelnafex-normal-forms-lasrel-left-to-right} and definition of \nafex, we can conclude.\qedhere
	\end{itemize}
\end{proof}

\subsection{Normal substituted terms characterization}
The following two lemmas characterize normal terms that are still normal when a variable is substituted by a value -- this characterization does not depend on $\equivx$ (by strong commutation and preservation of normal forms).
\begin{lemma}
	$\ntm\isub\var\val$ is normal iff $\ntm \not = \evctxp{\var\fire}$. (and $\val \not = \vartwo$)
\end{lemma}

\begin{lemma}
	\label{l:normal-subst-inert-are-inert}
	$\itm\isub\var\val$ is inert iff ($\itm$ is inert,) $\itm\isub\var\val$ is normal.
\end{lemma}

\subsection{Coherence of the \nafex simulation, reduction and substitution}
We split the Coherence Proposition of the \nafex simulation into two lemmas to prove, knowing that part of the first statement has already been proven in \reflemma{lasrelnafex-normal-forms-substitutive}.

\begin{lemma}
	\label{l:lasrelnafex-on-normal-subs_vsce}
	Let $\rel$ be an \nafex simulation, $\ntm \mlasrelnafex \ntmtwo$, and $\val\mlasrelnafex\valtwo$. If $\ntm\isub\var\val$ is $\tovsce$-normal then $\ntmtwo\isub\var\valtwo$ is $\tovsce$-normal.
\end{lemma}

\begin{proof}
	By induction on normal forms $\ntm$ for which $\ntm\isub\var\val$ is $\tovsce$-normal:
	\begin{itemize}
		\item \emph{Variable}. Two sub-cases:
		\begin{itemize}
			\item $\ntm= \var$ and so $\ntm\isub\var\val = \val$. Then $\ntmtwo = \var$ by case (\nafex 2) and $\ntmtwo\isub\var\valtwo = \valtwo$, which is $\tovsce$-normal.

			\item $\ntm= \vartwo$ and so $\ntm\isub\var\val = \vartwo$. Then $\ntmtwo = \vartwo$ by case (\nafex 2) and $\ntmtwo\isub\var\valtwo = \vartwo$, which is $\tovsce$-normal. 
		\end{itemize}
		
		\item \emph{Abstraction}, that is, $\ntm = \la\vartwo\tm$ and so $\ntm\isub\var\val = \la\vartwo\tm\isub\var\val$. Then $\ntmtwo= \la\vartwo\tmp$ with $\tm \mlasrel \tmp$. We have that $\ntmtwo\isub\var\valtwo = \la\vartwo\tmp\isub\var\valtwo$, which is $\tovsce$-normal.
		
		\item \emph{Substituted Inert}, that is, $\ntm = \ntmONE\esub\vartwo\itmONE$ and so $\ntm\isub\var\val = \ntmONE\isub\var\val\esub\vartwo{\itmONE\isub\var\val}$.
		$\ntm\isub\var\val$ is normal is equivalent to $\itmONE\isub\var\val$ and $\ntmONE\isub\var\val$ are normal.
		Then $\ntmtwo \equivx \ntmONEtwo\esub\vartwo\itmONEtwo$ with $\itmONE \mlasrel \itmONEtwo$ and $\ntmONE \mlasrel \ntmONEtwo$, which implies $\itmONE \mlasrelnafex \itmONEtwo$ and $\ntmONE \mlasrelnafex \ntmONEtwo$ by \reflemma{lasrelnafex-normal-forms-lasrel-right-to-left}.
		By \ih we then have $\itmONEtwo\isub\var\valtwo$ is $\tovsce$-normal and $\ntmONEtwo\isub\var\valtwo$ is $\tovsce$-normal : which is equivalent to  $\ntmtwo\isub\var\valtwo$ is $\tovsce$-normal by \refprop{equivx-preserves-normal-forms}.

		\item \emph{Applied Normal forms}, that is, $\ntm = \ntmONE\ntmTWO$. Then we have three sub-cases:
		\begin{itemize}
			\item $\ntm = \var\ntmTWO$ and $\val$ is not an abstraction,
			Then $\ntmtwo \equivx \var\ntmTWOtwo$ with $\ntmTWO \mlasrel \ntmTWOtwo$, which implies $\ntmTWO \mlasrelnafex \ntmTWOtwo$ by \reflemma{lasrelnafex-normal-forms-lasrel-right-to-left}.
			By \ih we then have $\ntmTWOtwo\isub\var\valtwo$ is $\tovsce$-normal: which is equivalent to $\ntmtwo\isub\var\valtwo$ is $\tovsce$-normal by \refprop{equivx-preserves-normal-forms}.
			\item $\ntm = \vartwo\ntmTWO$,
			Then $\ntmtwo \equivx \vartwo\ntmTWOtwo$ with $\ntmTWO \mlasrel \ntmTWOtwo$, which implies $\ntmTWO \mlasrelnafex \ntmTWOtwo$ by \reflemma{lasrelnafex-normal-forms-lasrel-right-to-left}.
			By \ih we then have $\ntmTWOtwo\isub\var\valtwo$ is $\tovsce$-normal: which is equivalent to $\ntmtwo\isub\var\valtwo$ is $\tovsce$-normal by \refprop{equivx-preserves-normal-forms}.
			\item $\ntm = \itmapp\ntmTWO$, and $\itmapp\isub\var\val$ is an applied inert,
			Then $\ntmtwo \equivx \itmapptwo\ntmTWOtwo$ with $\itmapp \mlasrel \itmapptwo$ and $\ntmTWO \mlasrel \ntmTWOtwo$, which implies $\itmapp \mlasrelnafex \itmapptwo$ and $\ntmTWO \mlasrelnafex \ntmTWOtwo$ by \reflemma{lasrelnafex-normal-forms-lasrel-right-to-left}.
			
			By \ih we then have $\itmapptwo\isub\var\valtwo$ is $\tovsce$-normal, $\ntmTWOtwo\isub\var\valtwo$ is $\tovsce$-normal : which is equivalent to $\ntmtwo\isub\var\valtwo$ is $\tovsce$-normal by \refprop{equivx-preserves-normal-forms} and by \reflemma{normal-subst-inert-are-inert}.\qedhere
		\end{itemize}

	\end{itemize}
\end{proof}

\begin{lemma} 
	\label{l:lasrelnafex-not-normal-subs}
	If $\ntm_\tmrone \mlasrelnafex \ntm_\tmrtwo$, $\val \mlasrel \valtwo$
	and $\ntm_\tmrone\isub\var{\val} \tovsc \tmronep$
	then $\ntm_\tmrtwo\isub\var{\valtwo}  \tovsc \tmrtwop$ and $\tmronep \mlasrel \tmrtwop$
\end{lemma}

\begin{proof}
	(We write $\val=\la\vartwo\tmfour$ and $\valtwo = \la\vartwo\tmfourp$ with $\tmfour \mlasrel \tmfourp$ - using the fact that ($\val \mlasrel \valtwo \Rightarrow \val \mlasrelnafex \valtwo$) by \ref{l:lasrelnafex-normal-forms-lasrel-right-to-left}.)
	
	If $\ntm_\tmrone\isub\var{\val} \tovsc \tmronep$ then $\ntm_\tmrone = \evctxp{\var\fire}$ ($\evctx$ is the context where the reduction has been done).
	Show (by induction on $\evctx$) that $\tmronep \mlasrel \tmrtwop$.
	\begin{itemize}
		\item $\evctx = \ctxhole$
		then $\ntm_\tmrone = {\var\fire}$ , 
		by $\ntm_\tmrone \mlasrelnafex \ntm_\tmrtwo$ we have $\ntm_\tmrtwo \equivx {\var\firetwo}$ with $\fire \mlasrel \firetwo$.

		By substitutivity of $\equivx$ (\refprop{equivx-is-substitutive}), $\ntm_\tmrtwo\isub\var\valtwo \equivx {\valtwo\firetwo\isub\var\valtwo}\tovsc {\tmfourp\esub\vartwo{\firetwo\isub\var\valtwo}}$
		
		By \refprop{equivx-is-a-strong-bisimulation}, there exists $\tmrtwop$ such that ${\tmfourp\esub\vartwo{\firetwo\isub\var\valtwo}} \equivx \tmrtwop$ and $\ntm_\tmrtwo\isub\var\valtwo \tovsc\tmrtwop$.
		
		We also have $\ntm_\tmrone\isub\var\val \tovsc {\tmfour\esub\vartwo{\fire\isub\var\val}} = \tmronep$ and we can build the following derivation:
		\[ \infer[\scequivx]{\tmronep\mlasrel\tmrtwop}{\infer[\scesub]{ \tmfour\esub\vartwo{\fire\isub\var\val} \mlasrel \tmfourp\esub\vartwo{\firetwo\isub\var\valtwo}}{\tmfour \mlasrel \tmfourp & \infer{\fire\isub\var\val \mlasrel \firetwo\isub\var\valtwo}{\fire \mlasrel \firetwo & \val \mlasrel \valtwo}} & {\tmfourp\esub\vartwo{\firetwo\isub\var\valtwo}} \equivx \tmrtwop} \]
		
		hence the result $\ntm_\tmrone\isub\var\val \tovsc \tmronep$, $\ntm_\tmrtwo\isub\var\valtwo \tovsc \tmrtwop$ and $\tmronep \mlasrel \tmrtwop$.

		\item $\evctx = \tmtwo\evctxONE$ ($\tmtwo = \itm$ because $\ntm_\tmrone$ is normal) 	then $\ntm_\tmrone = \itm\evctxONEp\tmthree$ (where $\tmthree = \var\firetwo$).
		Then by $\ntm_\tmrone \mlasrelnafex \ntm_\tmrtwo$, $\ntm_\tmrtwo \equivx \itmtwo\ntmONE$ with $\itm \mlasrel \itmtwo$ and $\evctxONEp\tmthree \mlasrel \ntmONE$.
		$\evctxONEp\tmthree\isub\var\val \tovsc \tm$ by hypothesis ($\ntm_\tmrone\isub\var\val \tovsc \tmronep$ is in this case $\itm\isub\var\val\evctxONEp\tmthree\isub\var\val \tovsc \itm\isub\var\val\tm$) and $\evctxONEp\tmthree \mlasrelnafex \ntmONE$ (normal forms, apply lemma \ref{l:lasrelnafex-normal-forms-lasrel-right-to-left}), hence by \ih $\ntmONE\isub\var\valtwo \tovsc \tmp$ with $\tm \mlasrel \tmp$.
		
		By substitutivity and strong commutation of $\equivx$,
		$\ntm_\tmrtwo\isub\var\valtwo \equivx \itmtwo\isub\var\valtwo\ntmONE\isub\var\valtwo\tovsc \itmtwo\isub\var\valtwo\tmp$ implies $\ntm_\tmrtwo\isub\var\valtwo\tovsc\tmrtwop$ with 
		$\itmtwo\isub\var\valtwo\tmp \equivx \tmrtwop$.
		\[\infer{\tmronep\mlasrel\tmrtwop}{\infer{\itm\isub\var\val\tm \mlasrel \itmtwo\isub\var\valtwo\tmp }{ \infer{\itm\isub\var\val \mlasrel \itmtwo\isub\var\valtwo}{\itm \mlasrel \itmtwo & \val \mlasrel \valtwo} & \tm \mlasrel \tmp} & \itmtwo\isub\var\valtwo\tmp \equivx \tmrtwop}\]
		
		hence $\ntm_\tmrone\isub\var\val \tovsc \tmronep$, $\ntm_\tmrtwo\isub\var\valtwo \tovsc \tmrtwop$ and $\tmronep \mlasrel \tmrtwop$.

		\item The rest of the induction cases ($\evctx = \evctxONE\tmtwo$, $\evctx = \tmtwo\esub\varthree\evctxONE$ or $\evctx = \evctxONE\esub\varthree\tmtwo$) follow from very similar arguments.\qedhere
	\end{itemize}
\end{proof}

\subsection{Mirrored Lassen's Closure preserves \nafex simulations}
After all these preliminaries, we can finally prove \nafex compatibility, with a very similar proof than in the case of naive similarity.

\begin{proposition}
	\label{prop:main-lemma_vsce}
	Let $\relsym$ be a \nafex simulation.
	\begin{enumerate}
		\item \emph{Technical auxiliary statement}: if $\tmrone\mlasrel\tmrtwo$ and $\tmrone \bsvsct k \ntm$ then $\tmrtwo\bsvscts \ntmtwo$ and $\ntm \mlasrelnafex \ntmtwo$.		
		\item \emph{Mirrored Lassen's closure preserves \nafex simulations}:  $\mlassenop\relsym$ is a \nafex simulation.
	\end{enumerate}
\end{proposition}

\begin{proof}
	\begin{enumerate}
		\item 
	By induction on $(k,d)$ where $d$ is the size of the derivation of $\tmrone \mlasrel \tmrtwo$.


By case analysis on the last rule of the derivation $\tmrone\mlasrel \tmrtwo$.

\begin{enumerate}
	\item \emph{Lifting}:
	\[ \infer[(\sclift) ]{\tmrone \mlasrel \tmrtwo} {\tmrone \rel \tmrtwo}\text{ and }\tmrone\bsvsct k \ntm\]
	Since $\relsym$ is a \nafex simulation, we have $\tmrone\relnafex\tmrtwo$ and $\tmrtwo \bsvscts \ntmtwo$ for some $\ntmtwo$ such that $\ntm\relnafex\ntmtwo$. Hence  $\ntm\mlasrelnafex\ntmtwo$ by monotonicity of $\opnafex$.
	
	\item \emph{Variables}:
	\[\infer[(\scvar) ]{\var \mlasrel \var}	{} \text{ and } \var\bsvsct 0 \var\]
	
	hence the result ($\var\bsvsct 0 \var$) and by the definition of \nafex, $\var \mlasrelnafex \var$.
	
	\item \emph{Abstraction}:
	\[\infer[(\scabs) ]{\la\var\tmrone \mlasrel \la\var\tmrtwo} {\tmrone \mlasrel \tmrtwo} \text{ and } \la\var\tmrone \bsvsct 0 \la\var\tmrone \]
	
	hence the result ($\la\var\tmrtwo \bsvsct 0 \la\var\tmrtwo$) and by the definition of \nafex, $\la\var\tmrone \mlasrelnafex \la\var\tmrtwo$.
	\item \emph{Application}:
	\[ \infer[(sc.app) ] 
	{\tmrone\tmrthree  \mlasrel  \tmrtwo\tmrfour} {\tmrone  \mlasrel \tmrtwo & \tmrthree \mlasrel \tmrfour } \text{ and }\tmrone\tmrthree \bsvsct k \ntm \]
	
	then, by case analysis on the last rule of the big-step derivation,
	\begin{enumerate}
		\item \emph{Applied variable}: 
		\[\infer{\tmrone\tmrthree \bsvsct {k+h} \isctxp{\var\ntm}}{
			\tmrone \bsvsct k \isctxp\var
			&
			\tmrthree \bsvsct h \ntm
		}\]
		
		by inductive hypothesis ($d$ strictly decreasing, first component not increasing) we obtain $\tmrtwo \bsvscts \firep$ with $\isctxp{\var} \mlasrelnafex \firep$ (hence by \reflemma{values-fireballs-stable-lasrelnafex}, $\firep = \isctxtwop{\var}$)and $\tmrfour \bsvscts \ntmtwo$ with $\ntm \mlasrelnafex \ntmtwo$. Then:
		\[\infer{\tmrtwo\tmrfour \bsvscts \isctxtwop{\var\ntmtwo}}{
			\tmrtwo \bsvscts  \itmtwo = \isctxtwop{\var}
			&
			\tmrfour \bsvscts \ntmtwo
		}\]
		By definition of \nafex, and 
		by \reflemma{lasrelnafex-normal-forms-isctx-decomposition-2}
		we have $\isctxp{\var\ntm} \mlasrelnafex \isctxtwop{\var\ntmtwo}$.
		\item \emph{Applied inert}: 
		\[\infer{\tmrone\tmrthree \bsvsct {k+h} \isctxp{\itmappONE\ntm}}{
			\tmrone \bsvsct k \isctxp\itmappONE = \itm
			&
			\tmrthree \bsvsct h \ntm
		}\]
		
		by inductive hypothesis ($d$ strictly decreasing, first component not increasing) we obtain $\tmrtwo \bsvscts \itmtwo = \isctxtwop{\itmappONEtwo}$ (the $\tovsce$ normal form, and it is an inert by Corollary \ref{l:inerts-fireballs-stable-lasrelnafex}) and $\tmrfour \bsvscts \ntmtwo$ with $\itm\mlasrelnafex\itmtwo,~\ntm \mlasrelnafex \ntmtwo$. Then:
		\[\infer{\tmrtwo\tmrfour \bsvscts \isctxtwop{\itmappONEtwo\ntmtwo}}{
			\tmrtwo \bsvscts  \itmtwo = \isctxtwop{\itmappONEtwo}
			&
			\tmrfour \bsvscts \ntmtwo
		}\]
		
		By \reflemma{lasrelnafex-normal-forms-isctx-decomposition-applied-inerts}, $\isctxp{\itmappONE\ntm} \mlasrelnafex \isctxtwop{\itmappONEtwo\ntmtwo}$.
		
		\item \emph{Substitution of an inert}:
		not applicable.

		\item \emph{$m$ step}:
		\[\infer{\tmrone\tmrthree \bsvsct {k+i+1} \isctxp\ntm}{
			\tmrone \bsvsct k \isctxp{\la\var\tmronep}
			&
			{\tmronep\esub\var\tmrthree} \bsvsct i \ntm
		}\]

		then by inductive hypothesis ($d$ strictly decreasing, first component non increasing) on $\tmrone$ we get
		$\tmrtwo\bsvscts \ntm_\tmrtwo$ with $\isctxp{\la\var\tmronep} \mlasrelnafex \ntm_\tmrtwo$ (hence by  \reflemma{values-fireballs-stable-lasrelnafex} $\ntm_\tmrtwo = \isctxtwop{\la\var\tmrtwop}$) \ie $\isctxp{\la\var\tmronep} \mlasrel \isctxtwop{\la\var\tmrtwop}$.
		
		Hence by \reflemma{lasrelnafex-normal-forms-lasrel-right-to-left}, $\isctxp{\la\var\tmronep} \mlasrelnafex \isctxtwop{\la\var\tmrtwop}$ and by \reflemma{lasrelnafex-values-isctx-decomposition} we get\\ $\isctxp{\var} \mlasrelnafex \isctxtwop{\var}$ and $\la\var\tmronep\mlasrelnafex \la\var\tmrtwop$ then $\tmronep\mlasrel\tmrtwop$ by case (\nafex 3).
		
		Then:
		\[\infer{\tmronep\esub\var\tmrthree \mlasrel \tmrtwop\esub\var\tmrfour}{\tmronep\mlasrel\tmrtwop&\tmrthree\mlasrel\tmrfour}\]
		
		since $\tmronep\esub\var\tmrthree \bsvsct i \ntm$ with $i < k+i+1$ we can apply the inductive hypothesis on the first component for $\tmronep\esub\var\tmrthree$ obtaining $\tmrtwop\esub\var\tmrfour \bsvscts  \ntmtwo$ for some $\ntmtwo$ such that $\ntm\mlasrelnafex\ntmtwo$. Since $\isctxp{\var} \mlasrelnafex \isctxtwop{\var}$ and $\ntm\mlasrelnafex\ntmtwo$, by \reflemma{lasrelnafex-normal-forms-isctx-decomposition}, we get $\isctxp{\ntm} \mlasrelnafex \isctxtwop{\ntmtwo}$.
		Last, note that $\tmrtwo\tmrfour\bsvscts\isctxtwop\ntmtwo$ by 
		\[\infer{\tmrtwo\tmrfour \bsvscts \isctxtwop\ntmtwo}{
			\tmrtwo \bsvscts \isctxtwop{\la\var\tmrtwop}
			&
			\tmrtwop\esub\var\tmrfour \bsvscts \ntmtwo
		}\]

		\item \emph{$e$ step}:
		not applicable.
	\end{enumerate}
	
	\item \emph{Explicit Substitution}: 
	\[ \infer[(sc.esubst) ]{\tmrone\esub\var{\tmrthree} \mlasrel \tmrtwo\esub\var{\tmrfour}} {\tmrone \mlasrel \tmrtwo & \tmrthree \mlasrel \tmrfour }\text{ and }\tmrone\esub\var{\tmrthree} \bsvsct k \ntm \]
	
	\begin{enumerate}
		\item \emph{Applied inert}: not applicable.
		\item \emph{Substitution of an inert}:
		\[ \infer{\tmrone\esub\var\tmrthree \bsvsct {k+h} \ntm\esub\var\itm}{
			\tmrone \bsvsct k \ntm
			&
			\tmrthree \bsvsct h \itm
		} \]
		
		by inductive hypothesis ($d$ strictly decreasing, first component not increasing) we obtain $\tmrtwo \bsvscts \ntmtwo$ and $\tmrfour \bsvscts \itmtwo$ with $\itm\mlasrelnafex\itmtwo$ and $\ntm \mlasrelnafex \ntmtwo$. By \reflemma{lasrelnafex-normal-forms-lasrel-left-to-right},  $\itm\mlasrel\itmtwo$ and $\ntm \mlasrel \ntmtwo$. Then:
		\[\infer{\tmrtwo\esub\var\tmrfour \bsvscts \ntmtwo\esub\var\itmtwo}{
			\tmrtwo \bsvscts  \ntmtwo
			&
			\tmrfour \bsvscts \itmtwo
		}\]
		and $\ntm\esub\var\itm \mlasrelnafex \ntmtwo\esub\var\itmtwo$ by definition of \nafex.
		
		\item \emph{m step}: not applicable.
		\item \emph{e step}: 		\[ \infer{\tmrone\esub\var{\tmrthree} \bsvsct {k+i+1} \isctxp\ntm}{
			\tmrthree \bsvsct k \isctxp{\valof{\tmrthree}}
			&
			\tmrone\isub\var{\valof{\tmrthree}}\bsvsct i \ntm
		} \]
		
		then by inductive hypothesis ($d$ strictly decreasing, first component non increasing) on $\tmrone$ and $\tmrthree$ we get
		$\tmrfour\bsvscts\tmrfourp$ with $\isctxp{\valof\tmrthree} \mlasrelnafex \tmrfourp$ (hence by \reflemma{values-fireballs-stable-lasrelnafex} $\tmrfourp = \isctxtwop{\valof\tmrfour}$) \ie $\isctxp{\valof\tmrthree} \mlasrelnafex \isctxtwop{\valof\tmrfour}$.
		Hence by \reflemma{lasrelnafex-values-isctx-decomposition} we get $\isctxp{\var} \mlasrelnafex \isctxtwop{\var}$ and ${\valof\tmrthree}\mlasrelnafex{\valof\tmrfour}$ \ie ${\valof\tmrthree}\mlasrel{\valof\tmrfour}$ by \ref{l:lasrelnafex-normal-forms-lasrel-left-to-right}.
		
		Then:
		\[\infer{\tmrone\isub\var{\valof\tmrthree} \mlasrel \tmrtwo\isub\var{\valof\tmrfour}}{\tmrone\mlasrel\tmrtwo& {\valof\tmrthree}\mlasrel{\valof\tmrfour}}\]
		
		since $\tmrone\isub\var{\valof\tmrthree} \bsvsct i \ntm$ with $i < k+i+1$ we can apply the inductive hypothesis on the first component for $\tmrone\isub\var{\valof\tmrthree}$ obtaining $\tmrtwo\isub\var{\valof\tmrfour} \bsvscts  \ntmtwo$ for some $\ntmtwo$ such that $\ntm\mlasrelnafex\ntmtwo$. 
		Since $\isctxp{\var} \mlasrelnafex \isctxtwop{\var}$ and $\ntm\mlasrelnafex\ntmtwo$, by \reflemma{lasrelnafex-normal-forms-isctx-decomposition}, we get $\isctxp{\ntm} \mlasrelnafex \isctxtwop{\ntmtwo}$.
		
		Last, note that $\tmrtwo\tmrfour\bsvscts\isctxtwop\ntmtwo$ by 
		\[\infer{\tmrtwo\tmrfour \bsvscts \isctxtwop\ntmtwo}{
			\tmrfour \bsvscts \isctxtwop{\valof\tmrfour}
			&
			\tmrtwo\isub\var{\valof\tmrfour} \bsvscts \ntmtwo
		}\]
		
	\end{enumerate}

	\item \emph{Implicit Substitution}: 
	\[ \infer[(sc.subst) ]{\tmrone\isub\var{\valof\tmrthree} \mlasrel \tmrtwo\isub\var{\valof\tmrfour}} {\tmrone \mlasrel \tmrtwo & \valof\tmrthree \mlasrel \valof\tmrfour }\text{ and }\tmrone\isub\var{\valof\tmrthree} \bsvsct k \ntm \]
	then by applying the splitting lemma (\reflemma{splitting_vsce}), we obtain $\tmrone \bsvsct {k_1} \ntm_\tmrone$ and $\ntm_\tmrone\isub\var{\valof\tmrthree} \bsvsct {k_2} \ntm$ with $k=k_1+k_2$. Hence by inductive hypothesis ($d$ strictly decreasing, first component non increasing) $\tmrtwo\bsvscts\ntm_\tmrtwo$ and $\ntm_\tmrone \mlasrelnafex \ntm_\tmrtwo$, and by \ih $\valof\tmrthree \mlasrelnafex \valof\tmrfour $. By applying \reflemma{lasrelnafex-normal-forms-lasrel-left-to-right}, we have $\ntm_\tmrone \mlasrel \ntm_\tmrtwo$.
	We then have: 
	\[\infer{\ntm_\tmrone\isub\var{\valof\tmrthree} \mlasrel \ntm_\tmrtwo\isub\var{\valof\tmrfour}} {\ntm_\tmrone \mlasrel \ntm_\tmrtwo & \valof\tmrthree \mlasrel \valof\tmrfour }\]
	
	
	Two cases.
	\begin{enumerate}

		\item \emph{$\tmrone$ is not normal}, that is, $k_1>0$ and $k_2<k$. Then by applying the induction hypothesis to $k_2$ (first component)  and $\ntm_\tmrone\isub\var{\valof\tmrthree}$ we obtain $\ntm_\tmrtwo\isub\var{\valof\tmrfour} \bsvscts \ntmtwo$ with $\ntm\mlasrelnafex\ntmtwo$. We conclude using stability \ref{l:stability_vsce} and the equivalence between big and small steps, because $\tmrtwo\isub\var{\valof\tmrfour} \to^* \ntm_\tmrtwo\isub\var{\valof\tmrfour} \to^* \ntmtwo$.

		\item \emph{$\tmrone$ is normal}, that is, $k_1=0$ and $k_2=k$. Then $\tmrone=\ntm_\tmrone$. Two sub-cases:
		\begin{itemize}
			\item \emph{$\tmrone\isub\var{\valof\tmrthree} = \ntm_\tmrone\isub\var{\valof\tmrthree}$ is also normal}

			Since we know that $\ntm_\tmrone \mlasrelnafex \ntm_\tmrtwo$  and $\valof\tmrthree \mlasrelnafex \valof\tmrfour $, we can apply \reflemma{lasrelnafex-on-normal-subs_vsce}, and obtain that $\ntm_\tmrtwo\isub\var{\valof\tmrfour}$ is $\tovsce$-normal and by \reflemma{lasrelnafex-normal-forms-substitutive}, $\tmrone\isub\var{\valof\tmrthree} \mlasrelnafex \ntm_\tmrtwo\isub\var{\valof\tmrfour}$. It is only left to show that $\tmrtwo\isub\var{\valof\tmrfour} \bsvscts \ntm_\tmrtwo\isub\var{\valof\tmrfour}$, which follows from $\tmrtwo\bsvscts \ntm_\tmrtwo$, stability of reduction under substitution (\reflemma{stability_vsce}) and the fact that $\ntm_\tmrtwo\isub\var{\valof\tmrfour}$ is $\tovsce$-normal (plus the equivalence of small and big steps).
			
			\item   \emph{$\tmrone\isub\var{\valof\tmrthree} = \ntm_\tmrone\isub\var{\valof\tmrthree}$ is not normal}
			
			hence $\ntm_\tmrone\isub\var{\valof\tmrthree} \to \tmronep \to ^ {k-1} \ntm$ (the reduction is diamond, all reductions are of the same length, we pick any first step possible). Then by \reflemma{lasrelnafex-not-normal-subs} , $\ntm_\tmrtwo\isub\var{\valof\tmrfour} \to \tmrtwop$ with $\tmronep \mlasrel \tmrtwop$.
			
			We can apply the inductive hypothesis to $\tmronep$ (first component is decreasing, as $k-1<k$) and we obtain $\tmrtwop \bsvscts \ntmtwo$ with $\ntm\mlasrelnafex\ntmtwo$.
			The statement is then proved, since (using \reflemma{stability_vsce})
			$$\tmrtwo\isub\var{\valof\tmrfour} \to^* \ntm_\tmrtwo\isub\var{\valof\tmrfour} \to \tmrtwop \to^* \ntmtwo$$ that is, $\tmrtwo\isub\var{\valof\tmrfour} \bsvscts \ntmtwo$ by \reflemma{ss-bs-equivalence_vsce}.
			
		\end{itemize}
	\end{enumerate}

	\item \emph{Equivalent $X$} \[	\infer[\scequivx]{\tmrone \mlasrel \tmrtwop} {\tmrone \mlasrel \tmrtwo & \tmrtwo \equivx \tmrtwop} \text{ and }\tmrone \bsvsct k \ntm \]
	
	by \ih, $\tmrtwo \bsvscts \ntmtwo$ and $\ntm \mlasrelnafex \ntmtwo$.
	
	Since $\tmrtwo$ has a normal form and $\tmrtwo \equivx \tmrtwop$ then $\tmrtwop \bsvscts \ntmthree$ and $\ntmtwo \equivx \ntmthree$ by \refprop{equivx-is-a-strong-bisimulation}, hence $\ntm \mlasrelnafex \ntmthree$ by \refprop{relnafex-equivx-subseteq-relnafex}.
	
\end{enumerate}

	\item Reformulation of the first point.\qedhere
\end{enumerate}
\end{proof}

\begin{proposition}[\nafex similarity is adequate]
	\label{prop:adequacy-nafex}
	Suppose $\equivx$ is a mirror for $\tovsc$. If $\tm \leqnafex \tmtwo$ then $\tm \bsvscs$ implies $\tmtwo\bsvscs$
\end{proposition}

\begin{proof}
	Without using equivalences in the definition this fact was obvious. With equivalences, we need the fact that $\equivx$ preserves normal forms to conclude (Proposition \ref{prop:equivx-preserves-normal-forms}).
\end{proof}

\gettoappendix{thm:nafex-included-leqc}

\begin{proof}
	\begin{enumerate}
		\item By Proposition \ref{prop:main-lemma_vsce} and coinductive definition of $\leqvscx$.
		\item Compatibility comes from the first point of the theorem and the fact that the mirrored Lassen's closure is compatible. Inclusion in contextual preorder by compatibility and adequacy (\refprop{adequacy-nafex}).
		\item Similar argument, here we rely on the fact that $\streq$ is substitutive and a strong commutation for the VSC. \qedhere
	\end{enumerate}
\end{proof}


\section{Proofs from \refsect{type-preorder} (From Operational to Denotational Semantics: the Type Preorder)}
In this section, we prove the propositions from the Section 12 \emph{From Operational to Denotational Semantics: the Type Preorder}.

\subsection{Type preorder is compatible}
Here we show that the type preorder is compatible. The proof is quite trivial, as the type preorder is somehow \emph{compositional}. We first prove a lemma about compositionality of syntax, then compatibility follows by an induction on contexts.

\begin{lemma} Type preorder verifies:
	\begin{itemize}
		\item \emph{(applicative)}  $\tm \leqtype \tmp ~\&~ \tmtwo\leqtype\tmtwop \Rightarrow {\tm\tmtwo}  \leqtype{\tmp\tmtwop}$.
		\item \emph{(abstractive)} $\tm \leqtype \tmtwo \Rightarrow {\la\var\tm} \leqtype {\la\var\tmtwo}$.
		\item \emph{(explicitly substitutive)} $\tm \leqtype \tmp ~\&~ \tmtwo \leqtype \tmtwop \Rightarrow {\tm\esub\var\tmtwo} \leqtype {\tmp\esub\var\tmtwop}$.
		\item \emph{($\alpha$-equivalence)} ${\la\var\tm} \equivtype {\la\vartwo\tm\isub\var\vartwo}$.
		\item \emph{($\alpha$-equivalence')} ${\tm\esub\var\tmtwo} \equivtype {(\tm\isub\var\vartwo)\esub\vartwo\tmtwo}$.
	\end{itemize}
\end{lemma}

\begin{proof}
\hfill
	\begin{itemize}
		\item \emph{(applicative)}, \emph{(explicitly substitutive)} and \emph{(abstractive)} properties can be done by looking at trees since they are very syntax driven. We look at the (ES) case to sketch the idea:
		\\
		
		Let $\typeder$ be a type derivation for $\tm\esub\varthree\tmtwo$, $\typeder :~ \typectx \types \tm\esub\varthree\tmtwo \hastype \mtype$.
		
		We show that there exists a type derivation $\typederp$ such that $\typederp :~ \typectx \types \tmp\esub\varthree\tmtwop \hastype \mtype$.
		
		Since the term $\tm\esub\varthree\tmtwo$ is not a value, there is only one possibility for the last rule of the derivation: ($\typingruleES$).
		\[\infer[\typingruleES]{\typectx \types \tm\esub\varthree\tmtwo \hastype \mtype}{\typectx_1, \varthree \hastype \mtypetwo \types \tm \hastype \mtype & \typectx_2 \types \tmtwo \hastype \mtypetwo}\]
		
		where $\typectx = \typectx_1 \uplus \typectx_2$.
		
		Since $\tm$ is type equivalent to $\tmp$ and $\tmtwo$ is type equivalent to $\tmtwop$, there exists two derivations $\typectx_2 \types \tmtwop \hastype \mtypetwo$ and $\typectx_1, \varthree \hastype \mtypetwo \types \tmp \hastype \mtype$. Hence we can construct the appropriate derivation $\typederp$ for $\tmp\esub\varthree\tmtwop$.
		\[\infer[\typingruleES]{\typectx \types \tmp\esub\varthree\tmtwop \hastype \mtype}{\typectx_1, \varthree \hastype \mtypetwo \types \tmp \hastype \mtype & \typectx_2 \types \tmtwop \hastype \mtypetwo}\]

		Hence $\tm\esub\varthree\tmtwo\leqtype \tmp\esub\varthree\tmtwop$.
		
		\item \emph{($\alpha$-equivalence)} and \emph{($\alpha$-equivalence')} are obvious, because typing judgments do not depend on the representation of bounded variables.\qedhere
	\end{itemize}
\end{proof}

\gettoappendix{prop:type-preorder-is-compatible}

\begin{proof}
	\begin{enumerate}
		\item By induction on $\ctx$.
	\begin{itemize}
		\item $\ctx = \ctxhole$, $\tm\leqtype\tmp$.
		\item $\ctx \leqtype \tmtwo\ctxtwo$, then by induction ${\ctxtwop\tm} \leqtype{\ctxtwop\tmp}$ and obviously $\tmtwo \leqtype \tmtwo$, hence by the (applicative) property ${\tmtwo\ctxtwop\tm} \leqtype{\tmtwo\ctxtwop\tmp}$.
		\item $\ctx = \ctxtwo\tmtwo$, $\ctx \leqtype \ctxtwo\esub\var\tmtwo$ and $\ctx \leqtype \tmtwo\esub\var\ctxtwo$ are similar to the previous case.
		\item $\ctx = \la\var\ctxtwo$ then by induction ${\ctxtwop\tm} \leqtype{\ctxtwop\tmp}$, hence by the (abstractive) property ${\la\var\ctxtwop\tm} \leqtype{\la\var\ctxtwop\tmp}$.
	\end{itemize}
\item By compatibility and adequacy.\qedhere

\end{enumerate}
\end{proof}

\subsection{\Net similarity is included into the Type preorder}
In fact, for all \nafex similarities, such that the following propoposition is true (and not necessarily that $\equivx$ is a mirror), we get that $\leqnafex \subseteq \leqtype$. The only prerequisite is that $\equivx$ satisfies the following proposition.

\begin{proposition}[$\equivx$-equivalence implies typability-equivalence]
	\label{prop:equivx-subseteq-equivtype}
	If $\tm\equivx\tmtwo$ then $\tm\equivtype\tmtwo$.
\end{proposition}

For $\leqtype$, this proposition is true, see Theorem \ref{thm:invariance-and-adequacy}.

\gettoappendix{l:bisimulation-preserves-typeder}

\begin{proof}
	
	\begin{enumerate}
		\item By induction on the size of the derivation $\typeder: \typectx \types \tm \hastype \mtype$.

	The term $\tm$ is typable by a derivation $\typeder: \typectx \types \tm \hastype \mtype$ therefore it is normalizable by \refthm{invariance-and-adequacy}. Hence we have $\tm\tovsc^k\ntm$ and therefore (since $\relsym$ is a bisimulation) $\tmp\tovsc^*\ntmtwo$ with $\ntm \relnafex \ntmtwo$. Instead of looking for a derivation $\typederp$  of $\tmp$, we can look for a derivation $\typederp_1$ of $\ntmtwo$ and conclude by (typability) expansion for the $\tovsc$ reduction.
	
	There is a derivation $\typeder_1: \typectx \types \ntm \hastype \mtype$ whose size is at most the size of $\typeder$.
	
	By case analysis on the last rule of the derivation $\typeder_1$.
	
	\begin{enumerate}
		\item \emph{Axiom rule.} \[\typeder_1 :~~~~~ \infer[\typingruleAx]{\var \hastype [\ltype] \types \ntm = \var \hastype \ltype}{}\]
		
		Then by $\ntm=\var\relnafex \ntmtwo$, $\ntmtwo = \var$ and $\typederp_1 \defeq \typeder_1$ types $\ntmtwo$ accordingly.
		\item \emph{Abstraction rule.} \[\typeder_1 :~~~~~ \infer[\typingruleAbs]{\typectx \types \ntm = \la\var\tmtwo \hastype \mtype \multimap \mtypetwo}{\typectx, \var \hastype \mtype \types  \tmtwo \hastype \mtypetwo}\]
		
		Then by $\ntm=\la\var\tmtwo\relnafex \ntmtwo$, $\ntmtwo = \la\var\tmtwop$ with $\tmtwo\rel\tmtwop$.
		
		The derivation $\typeder_2 : \typectx, \var \hastype \mtype \types  \tmtwo \hastype \mtypetwo$ is of a strictly smaller size than $\typeder$. By induction, since $\tmtwo\rel\tmtwop$, there is a derivation $\typederp_2 : \typectx, \var \hastype \mtype \types  \tmtwop \hastype \mtypetwo$.
		
		Then, \[\typederp_1 :~~~~~ \infer[\typingruleAbs]{\typectx \types \ntmtwo = \la\var\tmtwop \hastype \mtype \multimap \mtypetwo}{\infer*{\typectx, \var \hastype \mtype \types  \tmtwop \hastype \mtypetwo}{\typederp_2}}\]
		
		\item \emph{Many rule.} \[\typeder_1 : ~~~~~
		\infer[\typingruleMany]{\biguplus_{i\in I} \typectx_i \types \ntm = \val \hastype \biguplus_{i\in I} \ltype_i}{\infer*{(\typectx_i \types \ntm= \val \hastype \ltype_i)_{i\in I}}{\typedertwo_i}  & I~ \text{finite} } \]
		
		Then by $\ntm=\val\relnafex \ntmtwo$, $\ntmtwo = \valtwo$ with $\val\relnafex\valtwo$.
		
		Two sub-cases depending on the value nature of $\val$:
		\begin{itemize}
			\item \emph{Variable}. If $\val=\var$ then, $\valtwo=\var$ as well.
			Then, $\typederp_1\defeq\typeder_1$ is a correct derivation for $\valtwo$ and concludes the proof in this case.  
			\item \emph{Abstract}. If $\val=\la\var\tmtwo$ then, $\valtwo=\la\var\tmtwop$ with $\tmtwo\rel\tmtwop$.
			
			Suppose there is at least a $\typedertwo_i$ derivation (if there are none the result is trivial).
			
			Since $\ltype_i$ is a linear type the only possibility for the last rule of $\typedertwo_i$ is a ($\typingruleAbs$) rule. 
			
			Suppose $\ltype_i = \mtype_i\multimap \mtypetwo_i$.
			\[{\typedertwo_i} :~~~~~ \infer[\typingruleAbs]{\typectx \types \la\var\tmtwo \hastype \mtype_i\multimap \mtypetwo_i}{\infer*{\typectx, \var \hastype \mtype_i\types \tmtwo \hastype\mtypetwo_i}{\typederthree_i}}\]
			
			We know that $\tmtwo\rel\tmtwop$. By \ih on $\typederthree_i$ (whose size is strictly smaller than the size of $\typeder$), we get $\typederthreep_i : \typectx, \var \hastype \mtype_i\types \tmtwop \hastype\mtypetwo_i$. Hence we can reconstruct the appropriate $\typederp_1$ derivation.
			\[\typederp_1 : ~~~~~
			\infer[\typingruleMany]{\biguplus_{i\in I} \typectx_i \types \ntmtwo = \valtwo \hastype \biguplus_{i\in I} \ltype_i}{(\infer{\typectx_i \types \valtwo = \la\var\tmtwop \hastype \ltype_i}{{\infer*{\typectx, \var \hastype \mtype_i\types \tmtwo \hastype\mtypetwo_i}{\typederthreep_i}}})_{i\in I}  & I~ \text{finite} } \]
		\end{itemize}

		\item \emph{Application rule.} \[	\infer[\typingruleApp]{\typectx \uplus \typectxtwo \types \ntm=\ntmONE\ntmTWO \hastype \mtypetwo}{ \typectx \types \ntmONE \hastype [\mtype \multimap \mtypetwo] & \typectxtwo \types \ntmTWO \hastype \mtype }
		\]
		
		Then by $\ntm=\ntmONE\ntmTWO\relnafex \ntmtwo$, $\ntmtwo \equivx \ntmONEtwo\ntmTWOtwo$ with $\ntmONE\rel\ntmONEtwo$ and $\ntmTWO\rel\ntmTWOtwo$.
		
		By \refprop{equivx-subseteq-equivtype}, $\ntmtwo$ is type equivalent with $\ntmONEtwo\ntmTWOtwo$. Hence it is enough to construct an appropriate $\typederp$ for $\ntmONEtwo\ntmTWOtwo$.
		
		By induction on $\ntmONE,\ntmONEtwo$ and $\ntmTWO,\ntmTWOtwo$, we get the appropriate derivation.
		
		\item \emph{Explicit Substitution rule.} \[
		\infer[\typingruleES]{\typectx \uplus \typectxtwo \types \ntm = \ntmONE\esub\var\ntmTWO \hastype \mtypetwo}{ \typectx, \var \hastype \mtype \types \ntmONE \hastype \mtypetwo & \typectxtwo \types \ntmTWO \hastype \mtype }\]
		
		Then by $\ntm=\ntmONE\esub\var\ntmTWO\relnafex \ntmtwo$, $\ntmtwo \equivx \ntmONEtwo\esub\var\ntmTWOtwo$ with $\ntmONE\rel\ntmONEtwo$ and $\ntmTWO\rel\ntmTWOtwo$.
		
		By \refprop{equivx-subseteq-equivtype}, $\ntmtwo$ is type equivalent with $\ntmONEtwo\esub\var\ntmTWOtwo$. Hence it is enough to construct an appropriate $\typederp$ for $\ntmONEtwo\esub\var\ntmTWOtwo$.
		
		By induction on $\ntmONE,\ntmONEtwo$ and $\ntmTWO,\ntmTWOtwo$, we get the appropriate derivation.
	\end{enumerate}
\item 
Let $\tm,\tmp$ terms such that $\tm \leqnafex\tmp$.
By the first part of the proposition, $\tm\leqtype\tmp$.\qedhere
	\end{enumerate}
\end{proof}

\subsection{Lassen's Enf similarity is included in the Type preorder}
For Lassen's bisimilarity, the proof is a little more intricated, but relies on the same reasoning. We consider Plotkin's normal form and type them using the VSC multi types. Since multi types are invariant by reduction, this does not affect type equivalence.

\gettoappendix{l:smaller-derivations-stuck}

\begin{proof}
	By induction on $\levctx$.
	\begin{itemize}
		\item $\levctx = \ctxhole$
		\[\typeder =~~~	\infer[\typingruleApp]{\typectx \uplus \typectxtwo \types \var\val \hastype \mtype}{ \typectx \types \var \hastype [\mtypetwo_1 \multimap \mtype] & \infer*{\typectxtwo \types \val \hastype \mtypetwo_1}{\typeder_{\val}} }
		\]
		
		Then, $\typeder_{\levctx}$ is of size $1$, \ie clearly smaller than $\typeder$, and it is easy to see that $\typeder_{\val}$ is of size $\size{\typeder}-2$.
		
		\item $\levctx= \valtwo\levctxtwo$
		\[\typeder =~~~	\infer[\typingruleApp]{\typectx \uplus \typectxtwo \types \levctxp{\var\val}=\valtwo\levctxtwop{\var\val} \hastype \mtype}{ \typectx \types \valtwo \hastype [\mtype_1 \multimap \mtype] & \infer*{\typectxtwo \types \levctxtwop{\var\val} \hastype \mtype_1}{\typeder_1} }
		\]
		
		By \ih, $\exists \typeder_{1\levctx} :  \typectxtwo_{a}, \varthree \hastype \mtypetwo \types  \levctxtwop{\varthree} \hastype \mtype_1$ and $\exists \typeder_{\val} : \typectx_{\val} \types \val : \mtypetwo_1$ with $\size{\typeder_{1\levctx}}<\size{\typeder_1}$ and $\size{\typeder_{\val}}<\size{\typeder_1}$.
		
		Therefore $\size{\typeder_{\val}}<\size{\typeder}$. We complete $\typeder_{\levctx}$ in the following way:
		\[\typeder_{\levctx} =~~~	\infer[\typingruleApp]{\typectx \uplus \typectxtwo_{\levctx},\varthree \hastype \mtypetwo \types \levctxp\varthree=\valtwo\levctxtwop{\varthree} \hastype \mtype}{ \typectx \types \valtwo \hastype [\mtype_1 \multimap \mtype] & \infer*{\typectxtwo_{\levctx},\varthree \hastype \mtypetwo \types \levctxtwop{\varthree} \hastype \mtype_1}{\typeder_{1\levctx}} }
		\]
		
		Then $\size{\typeder_{\levctx}}<\size\typeder$.
		
		\item $\levctx=\levctxtwo\tm$
		\[\typeder =~~~	\infer[\typingruleApp]{\typectx \uplus \typectxtwo \types \levctxtwop{\var\val}\tm \hastype \mtype}{    \infer*{\typectxtwo \types \levctxtwop{\var\val} \hastype [\mtype_1 \multimap \mtype]}{\typeder_1} &\typectx \types \tm \hastype \mtype_1}
		\]
		
		By \ih, $\exists \typeder_{1\levctx} :  \typectxtwo_{a}, \varthree \hastype \mtypetwo \types  \levctxtwop{\varthree} \hastype [\mtype_1 \multimap \mtype]$ and $\exists \typeder_{\val} : \typectx_{\val} \types \val : \mtypetwo_1$ with $\size{\typeder_{1\levctx}}<\size{\typeder_1}$ and $\size{\typeder_{\val}}<\size{\typeder_1}$.
		
		Therefore $\size{\typeder_{\val}}<\size\typeder$. We complete $\typeder_{\levctx}$ in the following way:
		\[\typeder_{\levctx} =~~~	\infer[\typingruleApp]{\typectx \uplus \typectxtwo_{\levctx},\varthree\hastype\mtypetwo \types \levctxtwop{\varthree}\tm=\levctxp\varthree \hastype \mtype}{    \infer*{\typectxtwo_{\levctx},\varthree\hastype\mtypetwo \types \levctxtwop{\varthree} \hastype [\mtype_1 \multimap \mtype]}{\typeder_{1\levctx}} &\typectx \types \tm \hastype \mtype_1}
		\]
		
		Then $\size{\typeder_{\levctx}}<\size\typeder$.\qedhere
	\end{itemize}
\end{proof}

\begin{proposition}
	\label{prop:stuck-typed-applicative-form}
	${\levctxp{\var\val}}\equivtype{(\la\varthree\levctxp{\varthree})(\var\val)}$ where $\varthree$ does not appear in $\levctx$.
\end{proposition}

\begin{proof}
	Similar arguments apply, proof by induction on $\levctx$.
\end{proof}

\gettoappendix{l:enf-bisimulation-preserves-typeder}

\begin{proof}
\begin{enumerate}
	\item 
	By induction on the size of the derivation $\typeder: \typectx \types \tm \hastype \mtype$.
	
	\newcommand{\ntmleft}{\ntm_{\mathsf{left}}}
	\newcommand{\ntmlefttwo}{\ntmtwo_{\mathsf{left}}}
	The term $\tm$ is typable by a derivation $\typeder: \typectx \types \tm \hastype \mtype$ therefore it is normalizable by \refthm{invariance-and-adequacy}. Hence we have $\tm\tovsc^k\ntm$. By the fact that all VSC terminating terms are Plotkin's left terminating, we have that $\tm\tolw\ntmleft$ and therefore (since $\relsym$ is an enf bisimulation) $\tmp\tolw^*\ntmlefttwo$ with $\ntmleft \relenf\ntmlefttwo$. Instead of looking for a derivation $\typederp$  of $\tmp$, we can look for a derivation $\typederp_1$ of $\ntmlefttwo$ and conclude by (typability) expansion for the $\tolw$ reduction.
	
	There is a derivation $\typeder_1: \typectx \types \ntmleft \hastype \mtype$ whose size is at most the size of $\typeder$.
	
	By case analysis on the last rule of the derivation $\typeder_1$.
	
	\begin{enumerate}
		\item \emph{Axiom rule.} \[\typeder_1 :~~~~~ \infer[\typingruleAx]{\var \hastype [\ltype] \types \ntmleft = \var \hastype \ltype}{}\]
		
		Then by $\ntmleft=\var\relenf \ntmlefttwo$, $\ntmlefttwo = \var$ and $\typederp_1 \defeq \typeder_1$ types $\ntmlefttwo$ accordingly.
		\item \emph{Abstraction rule.} \[\typeder_1 :~~~~~ \infer[\typingruleAbs]{\typectx \types \ntmleft = \la\var\tmtwo \hastype \mtype \multimap \mtypetwo}{\typectx, \var \hastype \mtype \types  \tmtwo \hastype \mtypetwo}\]
		
		Then by $\ntmleft=\la\var\tmtwo\relenf \ntmtwo$, $\ntmlefttwo = \la\var\tmtwop$ with $\tmtwo\rel\tmtwop$.
		
		The derivation $\typeder_2 : \typectx, \var \hastype \mtype \types  \tmtwo \hastype \mtypetwo$ is of a strictly smaller size than $\typeder$. By induction, since $\tmtwo\rel\tmtwop$, there is a derivation $\typederp_2 : \typectx, \var \hastype \mtype \types  \tmtwop \hastype \mtypetwo$.
		
		Then, \[\typederp_1 :~~~~~ \infer[\typingruleAbs]{\typectx \types \ntmlefttwo = \la\var\tmtwop \hastype \mtype \multimap \mtypetwo}{\infer*{\typectx, \var \hastype \mtype \types  \tmtwop \hastype \mtypetwo}{\typederp_2}}\]
		
		\item \emph{Many rule.} \[\typeder_1 : ~~~~~
		\infer[\typingruleMany]{\biguplus_{i\in I} \typectx_i \types \ntmleft = \val \hastype \biguplus_{i\in I} \ltype_i}{\infer*{(\typectx_i \types \ntmleft= \val \hastype \ltype_i)_{i\in I}}{\typedertwo_i}  & I~ \text{finite} } \]
		
		Then by $\ntmleft=\val\relenf \ntmtwo$, $\ntmlefttwo = \valtwo$ with $\val\relenf\valtwo$.
		
		Two sub-cases depending on the value nature of $\val$:
		\begin{itemize}
			\item \emph{Variable}. If $\val=\var$ then, $\valtwo=\var$ as well.
			Then, $\typederp_1\defeq\typeder_1$ is a correct derivation for $\valtwo$ and concludes the proof in this case.  
			\item \emph{Abstract}. If $\val=\la\var\tmtwo$ then, $\valtwo=\la\var\tmtwop$ with $\tmtwo\rel\tmtwop$.
			
			Suppose there is at least a $\typedertwo_i$ derivation (if there are none the result is trivial).
			
			Since $\ltype_i$ is a linear type the only possibility for the last rule of $\typedertwo_i$ is a ($\typingruleAbs$) rule. 
			
			Suppose $\ltype_i = \mtype_i\multimap \mtypetwo_i$.
			\[{\typedertwo_i} :~~~~~ \infer[\typingruleAbs]{\typectx \types \la\var\tmtwo \hastype \mtype_i\multimap \mtypetwo_i}{\infer*{\typectx, \var \hastype \mtype_i\types \tmtwo \hastype\mtypetwo_i}{\typederthree_i}}\]
			
			We know that $\tmtwo\rel\tmtwop$. By \ih on $\typederthree_i$ (whose size is strictly smaller than the size of $\typeder$), we get $\typederthreep_i : \typectx, \var \hastype \mtype_i\types \tmtwop \hastype\mtypetwo_i$. Hence we can reconstruct the appropriate $\typederp_1$ derivation.
			\[\typederp_1 : ~~~~~
			\infer[\typingruleMany]{\biguplus_{i\in I} \typectx_i \types \ntmlefttwo = \valtwo \hastype \biguplus_{i\in I} \ltype_i}{(\infer{\typectx_i \types \valtwo = \la\var\tmtwop \hastype \ltype_i}{{\infer*{\typectx, \var \hastype \mtype_i\types \tmtwo \hastype\mtypetwo_i}{\typederthreep_i}}})_{i\in I}  & I~ \text{finite} } \]
		\end{itemize}

		\item \emph{Application rule.} \[	\infer[\typingruleApp]{\typectx \uplus \typectxtwo \types \ntmleft=\tmrone\tmrtwo \hastype \mtypetwo}{ \typectx \types \tmrone \hastype [\mtype \multimap \mtypetwo] & \typectxtwo \types \tmrtwo \hastype \mtype }
		\]
		
		Then by $\ntmleft=\tmrone\tmrtwo\relenf \ntmtwo$, $\ntmleft=\levctxp{\var\val}$ and $\ntmtwo = \levctxtwop{\var\valtwo}$ with $\levctxp\varthree\rel\levctxp\varthree$ and $\val\rel\valtwo$ with $\varthree$ fresh.
		
		By \refprop{stuck-typed-applicative-form}, typability of $\ntmleft$ is equivalent to typability of $(\la\varthree\levctxp{\varthree})(\var\val)$, which the derivation (for this type) is:
		\[	\infer[\typingruleApp]{\typectx \uplus \typectxtwo \types (\la\varthree\levctxp{\varthree})(\var\val) \hastype \mtypetwo}{ \infer[\typingruleAbs]{\typectx \types \la\varthree\levctxp{\varthree} \hastype [\mtype \multimap \mtypetwo]}{\infer*{\typectx, \varthree \hastype \mtype \types  \levctxp\varthree \hastype \mtypetwo}{\typeder_{\levctx}}} & \infer[\typingruleApp]{\typectxtwo',\var \hastype [\mtype_1 \multimap \mtype] \types \var\val \hastype \mtype}{\var \hastype [\mtype_1 \multimap \mtype] \types \var \hastype [\mtype_1 \multimap \mtype] & \infer*{\typectxtwo' \types \val \hastype \mtype_1}{\typeder_{\val}}} }
		\]
		By \reflemma{smaller-derivations-stuck}, $\typeder_{\levctx}$ and $\typeder_{\val}$ are strictly smaller than $\typeder$. Hence, by \ih, there exists $\typederp_{\levctx}$ and $\typederp_{\val}$ such that we can build the following derivation tree:
			\[\typederp=~~~	\infer[\typingruleApp]{\typectx \uplus \typectxtwo \types (\la\varthree\levctxtwop{\varthree})(\var\valtwo) \hastype \mtypetwo}{ \infer[\typingruleAbs]{\typectx \types \la\varthree\levctxtwop{\varthree} \hastype [\mtype \multimap \mtypetwo]}{\infer*{\typectx, \varthree \hastype \mtype \types  \levctxtwop\varthree \hastype \mtypetwo}{\typeder_{\levctx}}} & \infer[\typingruleApp]{\typectxtwo',\var \hastype [\mtype_1 \multimap \mtype] \types \var\valtwo \hastype \mtype}{\var \hastype [\mtype_1 \multimap \mtype] \types \var \hastype [\mtype_1 \multimap \mtype] & \infer*{\typectxtwo' \types \valtwo \hastype \mtype_1}{\typeder_{\val}}} }
		\]
		
		Which concludes the proof, since by \refprop{stuck-typed-applicative-form}, $(\la\varthree\levctxtwop{\varthree})(\var\valtwo)$ and $\levctxtwop{\var\valtwo}$ are type equivalent.
		
		\item \emph{Explicit Substitution rule.} \[
		\infer[\typingruleES]{\typectx \uplus \typectxtwo \types \ntmleft = \ntmONE\esub\var\ntmTWO \hastype \mtypetwo}{ \typectx, \var \hastype \mtype \types \ntmONE \hastype \mtypetwo & \typectxtwo \types \ntmTWO \hastype \mtype }\]
		This case is not possible: $\ntmleft$ is a term without explicit substitutions.
	\end{enumerate}
\item 
Let $\tm,\tmp$ terms such that $\tm \leqenf\tmp$.
By the first part, $\tm \leqtype\tmp$.\qedhere
\end{enumerate}
\end{proof}

\subsection{Multi types by value, regarding $\eta_v$ equivalence}
We now prove that multi types as they are defined in this paper, validate $\eta_v$ equivalence. We only need to prove the result for variables, as $\eta_v$ equivalence on abstractions is contained in $\equivbetav$.

\ignore{\subsubsection{Validating $\eta_v$ preorder}

\gettoappendix{prop:etav-for-leqtype}

\begin{proof}\begin{enumerate}
		\item \begin{itemize}
		\item 
	Let $(\typectx,\mtype)$ and $\typeder$ such that $\typeder : \typectx \types \la\vartwo\var\vartwo \hastype \mtype$.
	
	It is easy to type $\var$ with the multi type $\mtype$, with the context $\typectxtwo = \var \hastype \mtype$. We then need to prove that $\typectx = \typectxtwo$ to conclude.
	
	By unfolding the derivation $\typeder$, we get:

	Let $n$, $(\ltype_i)_{1\leq i\leq n}$ such that $\mtype = \multitype{n}{\ltype}$. Let $\mtypetwo_i$ and $\mtypetwo'_i$ such that $\ltype_i = \mtypetwo_i \multimap \mtypetwo'_i$.
	Let $m_i$, $((\ltypetwo_{i,j})_{1\leq j\leq m_i})_{0 \leq i \leq n}$ such that $\mtypetwo_i = \multitype{m_i}{\ltypetwo_{i,}}$.
	
	\[\infer[\typingruleMany]{\biguplus_{0 \leq i\leq n}\typectx_i \types \la\vartwo\var\vartwo \hastype \mtype}{\ldots & \infer[\typingruleAbs]{\typectx_i \types \la\vartwo\var\vartwo \hastype \ltype_i}{\infer[\typingruleApp]{\typectx_i,\vartwo \hastype \biguplus_{0 \leq j\leq m}[\ltypetwo_{i,j}] \types \var\vartwo \hastype \mtypetwo'_i}{\typectx_i \types \var \hastype [\biguplus_{0 \leq j\leq m}[\ltypetwo_{i,j}] \multimap \mtypetwo'_i] & \infer[\typingruleMany]{\vartwo \hastype \biguplus_{0 \leq j\leq m}[\ltypetwo_{i,j}] \types \vartwo \hastype \biguplus_{0 \leq j\leq m}\ltypetwo_{i,j}}{\left(\infer[\typingruleAx]{\vartwo \hastype [\ltypetwo_{i,j}] \types \vartwo \hastype \ltypetwo_{i,j}}{}\right)_{0\leq j \leq m}}}} & \ldots ~{0 \leq i \leq n}}\]
	
	Hence, if we keep unfolding on $\typectx_i \types \var \hastype [\biguplus_{0 \leq j\leq m}[\ltypetwo_{i,j}] \multimap \mtypetwo'_i]$, we get that $\biguplus_{0 \leq i\leq n}\typectx_i = \var \hastype \mtype$.
	\item $\var \hastype [\vartype] \types \var \hastype [\vartype]$ but $\var \hastype [\vartype] \not \types \la\vartwo\var\vartwo \hastype [\vartype]$, as an abstraction cannot be typed without a linear type of the shape $\mtype \multimap \mtypetwo$.
	\end{itemize}
\item For abstractions, the result is easier as $\eta_v$ equivalence amounts then to reducing under lambdas and $\alpha$-equivalence.
	\end{enumerate}
\end{proof}

\subsubsection{Removing ground types, validating $\eta_v$ equivalence}

Consider the same multi type system but with a new grammar of types, removing ground types (the empty multi type $\emptytype$ can be seen as acting as a ground type):

\begin{center}
	$\begin{array}{ccccc}
	\textsc{Linear Types} & \ltype, \ltypetwo &\grameq& \mtype \multimap \mtypetwo
	\\
	\textsc{Multi Types} & \mtype, \mtypetwo &\grameq& \multitype{n}{\ltype} & n\geq 0
	\end{array}$
\end{center}
}

\gettoappendix{prop:etav-for-leqtypetwo}

\begin{proof}
\hfill
\begin{enumerate}
		
		\item \emph{$\var \leqtype \la\vartwo\var\vartwo$}
		
		Let $\typeder$ be a type derivation such that $\typeder : \typectx \types \var \hastype \mtype$.
		
		We show that $\la\vartwo\var\vartwo$ can be typed in the same context and with the same type.
		
		Let $n$, $(\ltype_i)_{1\leq i\leq n}$ such that $\mtype = \multitype{n}{\ltype}$. Let $\mtypetwo_i$ and $\mtypetwo'_i$ such that $\ltype_i = \mtypetwo_i \multimap \mtypetwo'_i$.
		Let $m_i$, $((\ltypetwo_{i,j})_{1\leq j\leq m_i})_{0 \leq i \leq n}$ such that $\mtypetwo_i = \multitype{m_i}{\ltypetwo_{i,}}$.
		
		\[\infer[\typingruleMany]{\biguplus_{0 \leq i\leq n}\typectx_i \types \la\vartwo\var\vartwo \hastype \mtype}{\ldots & \infer[\typingruleAbs]{\typectx_i \types \la\vartwo\var\vartwo \hastype \ltype_i}{\infer[\typingruleApp]{\typectx_i,\vartwo \hastype \biguplus_{0 \leq j\leq m}[\ltypetwo_{i,j}] \types \var\vartwo \hastype \mtypetwo'_i}{\typectx_i \types \var \hastype [\biguplus_{0 \leq j\leq m}[\ltypetwo_{i,j}] \multimap \mtypetwo'_i] & \infer[\typingruleMany]{\vartwo \hastype \biguplus_{0 \leq j\leq m}[\ltypetwo_{i,j}] \types \vartwo \hastype \biguplus_{0 \leq j\leq m}\ltypetwo_{i,j}}{\left(\infer[\typingruleAx]{\vartwo \hastype [\ltypetwo_{i,j}] \types \vartwo \hastype \ltypetwo_{i,j}}{}\right)_{0\leq j \leq m}}}} & \ldots ~{0 \leq i \leq n}}\]
		
		It only remains to show that $\typectx = \var \hastype \biguplus_{0 \leq i\leq n} [\biguplus_{0 \leq j\leq m}[\ltypetwo_{i,j}] \multimap \mtypetwo'_i]$ and then we are done.
		
		Let's unfold the derivation $\typeder : \typectx \types \var \hastype \mtype$:
		\[\infer[\typingruleMany]{\biguplus_{0 \leq i\leq n}\typectxtwo_i \types \var \hastype \mtype}{(\infer[\typingruleAx]{\typectxtwo_i \types \var \hastype \mtypetwo_i \multimap \mtypetwo'_i}{})_{0 \leq i \leq n}}\]
		
		Hence $\typectxtwo_i = \var \hastype [\mtypetwo_i \multimap \mtypetwo'_i]$, which concludes the proof.
		
		\item \emph{$\la\vartwo\var\vartwo \leqtype \var$}, 
		
		Let $(\typectx,\mtype)$ and $\typeder$ such that $\typeder : \typectx \types \la\vartwo\var\vartwo \hastype \mtype$.
		
		It is easy to type $\var$ with the multi type $\mtype$, with the context $\typectxtwo = \var \hastype \mtype$. We then need to prove that $\typectx = \typectxtwo$ to conclude.
		
		By unfolding the derivation $\typeder$, we get:

		Let $n$, $(\ltype_i)_{1\leq i\leq n}$ such that $\mtype = \multitype{n}{\ltype}$. Let $\mtypetwo_i$ and $\mtypetwo'_i$ such that $\ltype_i = \mtypetwo_i \multimap \mtypetwo'_i$.
		Let $m_i$, $((\ltypetwo_{i,j})_{1\leq j\leq m_i})_{0 \leq i \leq n}$ such that $\mtypetwo_i = \multitype{m_i}{\ltypetwo_{i,}}$.
		\[\infer[\typingruleMany]{\biguplus_{0 \leq i\leq n}\typectx_i \types \la\vartwo\var\vartwo \hastype \mtype}{\ldots & \infer[\typingruleAbs]{\typectx_i \types \la\vartwo\var\vartwo \hastype \ltype_i}{\infer[\typingruleApp]{\typectx_i,\vartwo \hastype \biguplus_{0 \leq j\leq m}[\ltypetwo_{i,j}] \types \var\vartwo \hastype \mtypetwo'_i}{\typectx_i \types \var \hastype [\biguplus_{0 \leq j\leq m}[\ltypetwo_{i,j}] \multimap \mtypetwo'_i] & \infer[\typingruleMany]{\vartwo \hastype \biguplus_{0 \leq j\leq m}[\ltypetwo_{i,j}] \types \vartwo \hastype \biguplus_{0 \leq j\leq m}\ltypetwo_{i,j}}{\left(\infer[\typingruleAx]{\vartwo \hastype [\ltypetwo_{i,j}] \types \vartwo \hastype \ltypetwo_{i,j}}{}\right)_{0\leq j \leq m}}}} & \ldots ~{0 \leq i \leq n}}\]
		
		Hence, if we keep unfolding on $\typectx_i \types \var \hastype [\biguplus_{0 \leq j\leq m}[\ltypetwo_{i,j}] \multimap \mtypetwo'_i]$, we get that $\biguplus_{0 \leq i\leq n}\typectx_i = \var \hastype \mtype$.\qedhere
		\end{enumerate}
\end{proof}

\end{document}